\documentclass[10pt, reqno]{amsart}
\usepackage{fullpage}
\usepackage{amsmath, amsthm}
\usepackage{amscd}
\usepackage{enumerate}
\usepackage{float,amsmath,amssymb,mathrsfs,bm,multirow,graphics}
\usepackage[dvips]{graphicx}
\usepackage[percent]{overpic}
\usepackage[numbers,sort&compress]{natbib}
\usepackage{xcolor}
\usepackage{todonotes}
\usepackage{amsaddr} 
\usepackage{mathtools}
\usepackage{subcaption}
\usepackage{mathtools}
\usepackage{stackengine} 
 
\usepackage{tikz}
\usepackage{tikz-cd}
\usepackage{tqft}
\usetikzlibrary{positioning, shapes,arrows}
\usetikzlibrary{decorations.pathreplacing,decorations.markings}

\newtheorem{proposition}{Proposition}[section]
\newtheorem{theorem}[proposition]{Theorem}
\newtheorem{lemma}[proposition]{Lemma}

\newtheorem{assumption}[proposition]{Assumption}
\newtheorem{remark}[proposition]{Remark}

\newcommand{\re}{\text{\upshape Re} \,}
\newcommand{\im}{\text{\upshape Im} \,}
\newcommand{\ren}{\text{\upshape \normalfont ren}}

 \newcommand\ben{\begin{equation*}}
 \newcommand\ebn{\end{equation*}}
 \newcommand\beq{\begin{equation}}
 \newcommand\eeq{\end{equation}}
 \newcommand\lb{\left(}
  \newcommand\rb{\right)} 

\numberwithin{equation}{section}

\allowdisplaybreaks[1]

\usepackage[colorlinks=true]{hyperref}
\hypersetup{urlcolor=blue, citecolor=red, linkcolor=blue}

\begin{document}
\title[The family of confluent Virasoro fusion kernels]{The family of confluent Virasoro fusion kernels \\ and a non-polynomial $q$-Askey scheme}
\author{Jonatan Lenells and Julien Roussillon}
\address{Department of Mathematics, KTH Royal Institute of Technology, \\ 100 44 Stockholm, Sweden.}
\email{jlenells@kth.se}
\email{julienro@kth.se}

\begin{abstract}
We study the recently introduced family of confluent Virasoro fusion kernels $\mathcal{C}_k(b,\boldsymbol{\theta},\sigma_s,\nu)$. We study their eigenfunction properties and show that 
they can be viewed as non-polynomial generalizations of both the continuous dual $q$-Hahn and the big $q$-Jacobi polynomials. More precisely, we prove that: (i) $\mathcal{C}_k$ is a joint eigenfunction of four different difference operators for any positive integer $k$, (ii) $\mathcal{C}_k$ degenerates to the continuous dual $q$-Hahn polynomials when $\nu$ is suitably discretized, and (iii) $\mathcal{C}_k$ degenerates to the big $q$-Jacobi polynomials when $\sigma_s$ is suitably discretized. These observations lead us to propose the existence of a non-polynomial generalization of the $q$-Askey scheme. The top member of this non-polynomial scheme is the Virasoro fusion kernel (or, equivalently, Ruijsenaars' hypergeometric function), and its first confluence is given by the $\mathcal{C}_k$.
\end{abstract}

\maketitle

\tableofcontents

\section{Introduction}

The Askey--Wilson polynomials are a four-parameter family of 
$q$-hypergeometric series \cite{AW1985}. They satisfy one three-term recurrence relation and one difference equation of Askey--Wilson type. Moreover, they form the top element of a five-level hierarchy of $q$-orthogonal polynomials called the $q$-Askey scheme \cite{KLS2010, KS}. Each element of this scheme is a family of orthogonal polynomials satisfying one three-term recurrence relation and one difference equation. 
The families at a given level of the scheme arise as limits of the families at the level above. In particular, the two families at the second level of the $q$-Askey scheme are the continuous dual $q$-Hahn and the big $q$-Jacobi polynomials, and both of them arise as limits of the Askey--Wilson polynomials \cite{KS,KM}, see Figure \ref{schemefig1} (left).

In this article, we propose a non-polynomial version of the $q$-Askey scheme motivated by two-dimensional conformal field theory, see Figure \ref{schemefig1} (right). The top element of the proposed scheme is the Virasoro fusion kernel. The second level is made up of a family of confluent Virasoro fusion kernels which was recently introduced in \cite{LR}. The members of the non-polynomial scheme are associated with a quantum deformation parameter $q$ which is related to the central charge $c$ of the Virasoro algebra by
\begin{align}\label{qcQdef}
q=e^{2i\pi b^2}, \qquad c=1+6 Q^2, \qquad Q = b+b^{-1}.
\end{align}
Moreover, each member is a joint eigenfunction of four difference operators.
The proposed scheme is a generalization of the $q$-Askey scheme in the sense that its members reduce to members of the $q$-Askey scheme in appropriate limits when certain variables are discretized, see Figure \ref{schemefig2}. In this paper, we consider in detail the first two levels of the non-polynomial scheme and their relation to the first two levels of the $q$-Askey scheme; results on lower levels will be presented elsewhere. 

\vspace{.1cm}
\begin{figure}[h!]
 \hspace{.5cm}
\begin{subfigure}{.3\textwidth}
\centering
\tikzstyle{block} = [rectangle, draw, fill=blue!20, 
    text width=10em, text centered, rounded corners, minimum height=2em]
\tikzstyle{line} = [draw, -latex', shorten >= 4pt, shorten <= 0pt]
    
\begin{tikzpicture}[node distance = .8cm, auto]
    \node [block] (AW) {Askey--Wilson polynomials};
    \node [block, below left=0.5cm and -1.2cm of AW] (Hahn) {Continuous dual $q$-Hahn polynomials};
    \node [block, below right=0.5cm and -1.2cm of AW] (Jacobi) {Big $q$-Jacobi polynomials};
    \path [line] (AW) -- (Hahn);
    \path [line] (AW) -- (Jacobi);
\end{tikzpicture}
\end{subfigure}
\hspace{5cm}
\begin{subfigure}{.3\textwidth}
\centering
\tikzstyle{block} = [rectangle, draw, fill=blue!20, 
    text width=10em, text centered, rounded corners, minimum height=2em]
\tikzstyle{line} = [draw, -latex', shorten >= 4pt, shorten <= 0pt]
    
    \begin{tikzpicture}[node distance = .8cm, auto]
    \node [block] (F) {Virasoro fusion kernel};
    \node [block, below =of F] (Cn) {confluent Virasoro \\ fusion kernels};
    \path [line] (F) -- (Cn);
\end{tikzpicture}
\end{subfigure}
\caption{First two levels of the $q$-Askey scheme (left) and of the non-polynomial scheme (right). \label{schemefig1}}
\end{figure}
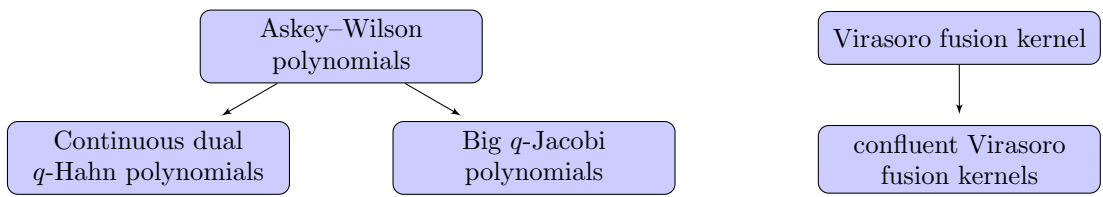

In the first part of the paper (Sections \ref{Fsec}-\ref{Ckdifferencesec}), we study the Virasoro fusion kernel $F\left[\substack{\theta_1\;\;\theta_t\vspace{0.1cm}\\ \theta_{\infty}\;\theta_0};\substack{\sigma_s \vspace{0.15cm} \\  \sigma_t}\right] $ as well as the confluent Virasoro fusion kernels $\mathcal{C}_k(b,\boldsymbol{\theta},\sigma_s,\nu)$ which were introduced in \cite{LR} as confluent limits of $F$. The Virasoro fusion kernel $F$ plays a fundamental role in the conformal bootstrap approach to two-dimensional conformal field theories \cite{pt1,Ribault}. It was first constructed in \cite{pt1, pt2} and later revisited in \cite{TV14} where it was interpreted as $b$-6j symbols associated to the quantum group $U_q(sl_2(\mathbb{R}))$, with $b$ characterizing the central charge of the Virasoro algebra according to (\ref{qcQdef}). 
The function $F$ satisfies two different pairs of difference equations, the first pair involving operators acting on $\sigma_s$ and the second pair operators acting on $\sigma_t$. 
On the other hand, the kernels $\mathcal{C}_k$ form a family of functions indexed by the integer $k \geq 1$. They are obtained from  $F$ by letting the variables $\theta_1$, $\theta_\infty$, and $\sigma_t$ tend to infinity in a prescribed way, see Section \ref{Ckdifferencesec}, and this confluent limit can be viewed as the first degeneration limit in the non-polynomial scheme. We show in Theorem \ref{thm6p2} and Theorem \ref{thm6p4} that, just like $F$, each of the kernels $\mathcal{C}_k$ is a joint eigenfunction of two different pairs of difference operators, with the operators in the first pair acting on $\nu$ and the operators in the second pair acting on $\sigma_s$. These difference equations are obtained by studying the confluent limits of the difference equations satisfied by $F$.

In the second part of the paper (Sections \ref{FtoAW}-\ref{CktoJacobisec}), we study the relation between the non-polynomial scheme and the $q$-Askey scheme. 
In the confluent limit, the parameter $\sigma_s$ is left unchanged while $\sigma_t$ is sent to infinity. As a result, the two pairs of difference equations satisfied by $\mathcal{C}_k$ are of different nature: the first pair is of the form satisfied by the continuous dual $q$-Hahn polynomials, 
while the second pair is of the form satisfied by the big $q$-Jacobi polynomials. 
This suggests that there is a relationship between the $\mathcal{C}_k$ and these polynomials. This relationship is made precise in Theorem \ref{thhahn} and Theorem \ref{thjacobi}, which together with Theorem \ref{FAWthm} form the main results of the second part of the paper. Theorem \ref{FAWthm} shows that $F$ reduces (up to normalization) to the Askey--Wilson polynomials when the variable $\sigma_s$  is suitably discretized. 
Similarly, Theorem \ref{thhahn} and Theorem \ref{thjacobi} show that $\mathcal{C}_k(b,\boldsymbol{\theta},\sigma_s,\nu)$ reduces (again up to normalization) to the continuous dual $q$-Hahn polynomials when $\nu$ is suitably discretized and to the big $q$-Jacobi polynomials when $\sigma_s$ is suitably discretized. Depending on whether $k$ is an odd or an even integer, the resulting polynomials are associated with the quantum deformation parameter $q$ or its inverse $q^{-1}$. 

The main results in Sections \ref{FtoAW}-\ref{CktoJacobisec} are summarized in Figure \ref{schemefig2}. 
In Figure \ref{schemefig2}, the limits from elements of the non-polynomial scheme to elements of the $q$-Askey scheme are indicated by dashed arrows; we refer to these limits as polynomial limits, since the limiting functions are polynomials. The limits within the two schemes are indicated by solid arrows. It is important to note that the diagram in Figure \ref{schemefig2} commutes: The confluent limit which was considered (with an entirely different goal in mind) in \cite{LR} to define the confluent Virasoro fusion kernels $\mathcal{C}_k$ descends to the degeneration limits of the $q$-Askey scheme in the polynomial limit.

 \begin{figure}[h!]
\centering
 \tikzstyle{block} = [rectangle, draw, fill=blue!20, 
    text width=10em, text centered, rounded corners, minimum height=2em]
\tikzstyle{line} = [draw, -latex', shorten >= 4pt, shorten <= 0pt]
       
\begin{tikzpicture}[node distance = .8cm, auto]
    \node [block] (F) {Virasoro fusion kernel};
      \node[block] at (10,0) (Cn) {confluent Virasoro \\ fusion kernels};
    \node[block] at (0,-3.5) (AW) {Askey--Wilson \\ polynomials};
   
    \node[block] at (5,-2.2) (Hn) {Continuous dual \\ q-Hahn polynomials};
    
      \node[block] at (5,-4.8) (Jn)  {big-q Jacobi \\ polynomials};     
   
    \path [line] (F) -- node [text width=1.5cm,midway,above]{Eq.\,\eqref{confluentlimit}} (Cn);
            \path [line] (AW) -- node [text width=2.5cm,midway,above]{Eq.\,\eqref{qhahn}} (Hn);
        \path [line] (AW) -- node [text width=2.7cm,midway,below]{Eq.\,\eqref{Jn}} (Jn);
    \path [line,dashed](F) -- node [text width=2cm,midway,left]{Theorem \ref{FAWthm}} (AW);
    \path[line,dashed] (Cn) -- node [text width=2.4cm,midway,left]{Theorem \ref{thhahn}} (Hn);
    \path[line,dashed] (Cn) -- node [text width=2.5cm,midway,right]{\; Theorem \ref{thjacobi}} (Jn);
\end{tikzpicture}
\caption{Illustration of the relationship between the non-polynomial scheme introduced in this paper and the $q$-Askey scheme. The diagram summarizes the main results of Sections \ref{FtoAW}-\ref{CktoJacobisec}. The solid and dashed arrow correspond to confluent and polynomial limits, respectively.
 \label{schemefig2}}
\end{figure}
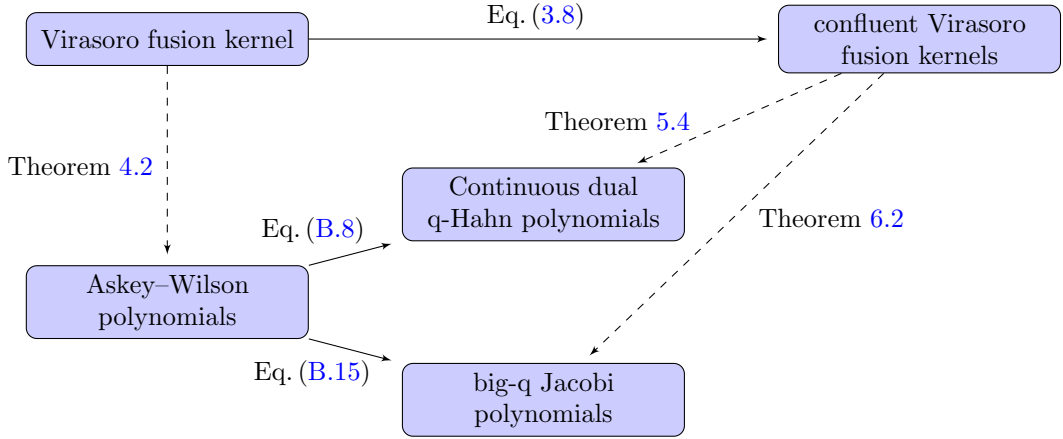

In the context of the $q$-Askey scheme, it is usually assumed that $0 < q <1$, or at least that $|q| < 1$ \cite{KLS2010}. However, in the non-polynomial setting it is natural to assume that $q$ lies on the unit circle; according to (\ref{qcQdef}), this corresponds to $b > 0$ and a central charge $c$ satisfying $c > 1$. Thus, even though we expect many of our results to analytically extend to other values of $q$, we will focus on the case when $q$ lies on the unit circle.
It is interesting to note that although the orthogonal polynomials in the $q$-Askey scheme are typically not defined (at least not in the standard way) when $q$ is a root of unity, no such restriction is necessary for the non-polynomial scheme. Consequently, we only need to impose the assumption that $q$ is not a root of unity in the second part of the paper where the polynomial limits are considered.

Sometimes the family of continuous $q$-Hahn polynomials are included in the second level of the $q$-Askey scheme, along with the continuous dual $q$-Hahn and the big $q$-Jacobi polynomials. However, since the continuous $q$-Hahn polynomials are related to the Askey--Wilson polynomials by simple phase shifts (see \cite[Eq. (14.1.17)]{KLS2010}), they can be obtained from the Virasoro fusion kernel in the same way as the Askey--Wilson polynomials. We will therefore not discuss them further in this paper.

\subsection{Relation to earlier work}
Our construction of a non-polynomial scheme with the Virasoro fusion kernel as its top member has been inspired by results presented in \cite{R1999} and \cite{R20}. First, it was shown in \cite{R1999}  that the Ruijsenaars hypergeometric function $R(a_+,a_-,\boldsymbol{c}; v, \hat{v})$ (also referred to as the $R$-function) reduces to the Askey--Wilson polynomials in a certain limit when either $v$ or $\hat{v}$ is suitably discretized. The key to the proof of this fact in \cite{R1999} is that, in this limit, one of the four difference equations satisfied by the $R$-function reduces to the three-term recurrence relation satisfied by the Askey--Wilson polynomials. 
Second, it was recently understood in \cite{R20} that, up to normalization, the Virasoro fusion kernel $F$ is equal to the $R$-function when the parameters of the two functions are appropriately identified. By combining these two results, it follows that the Virasoro fusion kernel also reduces to the Askey--Wilson polynomials in an appropriate limit. This observation is made precise in Theorem \ref{FAWthm}, which therefore can be viewed as a reformulation of the result from \cite{R1999} expressed in the language of the Virasoro fusion kernel using the identification put forth in \cite{R20}. However, in Section \ref{FtoAW} we give a direct and self-contained proof of Theorem \ref{FAWthm} based on the idea of \cite{R1999}, because this is easier than to explain how the assertion follows from \cite{R1999} and \cite{R20}.

To the best of our knowledge, no attempt has previously been made to derive an Askey type scheme with the $R$-function or the Virasoro fusion kernel as its top member. 
A non-polynomial generalization of the $q$-Askey scheme whose top member is the Askey--Wilson function was constructed in \cite{KS99}. It is however not clear if there is any relation between the non-polynomial scheme presented here and the scheme of \cite{KS99}. The Askey--Wilson function is a non-polynomial generalization of the Askey--Wilson polynomials \cite{KS99}; it is a joint eigenfunction of two difference operators of Askey--Wilson type \cite{BRS,IR} and it was shown in \cite{BRS} that it is proportional to a trigonometric Barnes integral, whose building block is the  $q$-Gamma function $\Gamma_q$ \cite{R2001}. In particular, the function $\Gamma_q$, and consequently the Askey--Wilson function, are well defined only for $|q| <1$. On the other hand, the $R$-function, which was introduced in \cite{R1994} and studied in greater detail in \cite{R1999,R2003,R2003bis}, is proportional to a hyperbolic Barnes integral \cite{BRS}, whose building block consists of the hyperbolic gamma function $G(a_+,a_-,z)$ \cite{R2001}, and it is defined for more general values of $q=e^{i\pi a_+/a_-}$ in the complex plane. The $R$-function can be expressed as a sum of two terms, where each term is proportional to a product of two Askey--Wilson functions \cite[Theorem 6.5]{BRS}.



\subsection{Organization of the paper}
We recall the definition and eigenfunction properties of the Virasoro fusion kernel $F$ in Section \ref{Fsec}. Eigenfunction properties of the confluent Virasoro fusion kernels $\mathcal{C}_k$ are derived in Section \ref{Ckdifferencesec}. In Section \ref{FtoAW}, we consider the reduction of $F$ to the Askey--Wilson polynomials. In Section \ref{CktoHahnsec}, we show that a renormalized version of $\mathcal{C}_k$ reduces to the continuous dual $q$-Hahn polynomials when $\nu$ is suitably discretized. In Section \ref{CktoJacobisec}, we prove  that a renormalized version of $\mathcal{C}_k$ reduces to the big $q$-Jacobi polynomials when $\sigma_s$ is suitably discretized. Section \ref{conclusionssec} contains some conclusions and perspectives. In Appendix \ref{appendixA}, the definition of $q$-hypergeometric series is recalled. In Appendix \ref{appendixB}, we review the properties of the first two levels of the $q$-Askey scheme that are needed for the proofs in Sections \ref{FtoAW}-\ref{CktoJacobisec}.

\section{The Virasoro fusion kernel} \label{Fsec}

The Virasoro fusion kernel, denoted by $F$, is defined by
\beq \label{fusion01}
\begin{split}
F\left[\substack{\theta_1\;\;\theta_t\vspace{0.1cm}\\ \theta_{\infty}\;\theta_0};\substack{\sigma_s \vspace{0.15cm} \\  \sigma_t}\right] = &  \prod_{\epsilon,\epsilon'=\pm1} \frac{g_b \lb \epsilon \theta_1+\theta_{t}+\epsilon' \sigma_t\rb g_b \lb \epsilon \theta_0-\theta_\infty+\epsilon' \sigma_t \rb}{g_b \lb \epsilon \theta_0 + \theta_t + \epsilon' \sigma_s \rb g_b \lb \epsilon \theta_1-\theta_\infty+\epsilon' \sigma_s \rb} \prod_{\epsilon=\pm1} \frac{g_b(\frac{iQ}2+2\epsilon \sigma_s)}{g_b(-\frac{iQ}2+2\epsilon \sigma_t)}
	\\
& \times \int_{\mathsf{F}} dx~\prod_{\epsilon=\pm1} \frac{s_b \lb x+ \epsilon \theta_1 \rb s_b \lb x+\epsilon\theta_0+\theta_\infty+\theta_t \rb}{s_b \lb x+\frac{iQ}{2}+\theta_\infty+\epsilon \sigma_s \rb s_b \lb x+\frac{i Q}{2}+\theta_t+\epsilon \sigma_t \rb},
\end{split}
\eeq
where $s_b(z)$ and $g_b(z)$ are the special functions defined by
\begin{equation}\label{defsb}
s_b(z)=\operatorname{exp}{\left[  i \int_0^\infty \frac{dy}{y} \left(\frac{\operatorname{sin}{2yz}}{2\operatorname{sinh}{b^{-1}y}\operatorname{sinh}{by}}-\frac{z}{y}\right)\right]}, \qquad |\im z|<\frac{Q}{2},
\end{equation}
and
\beq \label{gb}
g_b(z)=\operatorname{exp}{\left\{ \int_0^\infty \frac{dt}{t}\left[\frac{e^{2i z t}-1}{4 \operatorname{sinh}{b t} \operatorname{sinh}{b^{-1} t}}+\frac{1}{4}z^2 \lb e^{-2bt}+e^{-\frac{2t}b}\rb-\frac{iz}{2t} \right]\right\}}, \qquad \im z>-\frac{Q}{2}.
\eeq
In order to specify the contour of integration $\mathsf{F}$ in (\ref{fusion01}), we need to first recall some properties of $g_b(z)$ and $s_b(z)$. These functions are related to the functions $G$ and  $E$ in \cite[Eq. (A.3)]{R1999} and \cite[Eq. (A.43)]{R1999} by
\begin{align}\label{sbgbGE}
s_b(z) = G(b, b^{-1}; z), \qquad g_b(z) = \frac{1}{E(b, b^{-1}; -z)};
\end{align}
thus it follows from \cite{R1999} that 
\begin{itemize}
\item $s_b$ and $g_b$ satisfy the relation $s_b(z)= g_b(z)/g_b(-z)$ and
\item the function $g_b(z)$ has no zeros, but it has poles located at 
\beq\label{polegb}
z_{k,l}=-\frac{i Q}{2} -i k b-il b^{-1}, \qquad k,l = 0,1, 2, \dots.
\eeq
\end{itemize}
Consequently, the function $s_b(z)$ is a meromorphic function of $z \in \mathbb{C}$ with zeros $\{z_{m,l}\}_{m,l=0}^\infty$ and poles $\{p_{m,l}\}_{m,l=0}^\infty$ located at 
\begin{equation}\label{polesb}
\begin{split}
&z_{m,l}=\frac{i Q}{2}+i m b +il b^{-1}, \qquad m,l = 0,1, 2,\dots, \qquad (\text{zeros}), 
	\\
&p_{m,l}=-\frac{i Q}{2} -i m b-il b^{-1}, \qquad m,l = 0,1, 2,\dots, \qquad (\text{poles}).
\end{split}
\end{equation}
The multiplicity of the zero $z_{m,l}$ in (\ref{polesb}) is given by the number of distinct pairs $(m_i,l_i) \in \mathbb{Z}_{\geq 0} \times \mathbb{Z}_{\geq 0}$ such that $z_{m_i,l_i}=z_{m,l}$. 
The pole $p_{m,l}$ has the same multiplicity as the zero $z_{m,l}$.
In particular, if $b^2$ is an irrational real number, then all the zeros and poles in (\ref{polesb}) are distinct and simple.
We deduce that the integrand in (\ref{fusion01}) has eight semi-infinite sequences of poles in the complex $x$-plane. Assuming that $b > 0$, there are four downward sequences starting at $x = \pm \theta_1 -\frac{i Q}{2}$ and $x = \pm \theta_0 - \theta_\infty - \theta_t -\frac{i Q}{2}$, and four upward sequences starting at $x = -\theta_\infty \pm \sigma_s$ and $x = -\theta_t \pm \sigma_t$.
The contour $\mathsf{F}$ in (\ref{fusion01}) is any curve from  $-\infty$  to $+\infty$ which separates the four upward from the four downward sequences of poles. We can ensure the existence of such a contour by imposing the following restrictions on the parameters.  



\begin{assumption}[Restrictions on the parameters]\label{assumption}
Throughout the paper, we assume that
\beq \label{restrictions}
b>0,\quad (\theta_0,\theta_t,\theta_1,\theta_\infty)\in \mathbb{R}^4.
\eeq
\end{assumption}
Let us temporarily also assume that $\sigma_s, \sigma_t \in \mathbb{R}$. Then Assumption \ref{assumption} implies that $\mathsf{F}$ can be chosen to be any curve from $-\infty$ to $+\infty$ lying in the open strip $\im x\in (-Q/2,0)$.
Moreover, with this choice of $\mathsf{F}$, the integrand in (\ref{fusion01}) has exponential decay as $\re x \to \pm \infty$, so the integral in (\ref{fusion01}) is well-defined. The decay of the integrand follows from the following asymptotic formula which is a consequence of \cite[Theorem A.1]{R1999} and (\ref{sbgbGE}): 
For each $\epsilon >0$, 
\begin{align}\label{sbasymptotics}
\pm \ln s_b(z) = -\frac{i\pi z^2}{2} - \frac{i\pi}{24}(b^2 + b^{-2}) + O(e^{-\frac{2\pi(1-\epsilon)}{\max(b, b^{-1})}|\re z|}), \qquad \re z \to \pm \infty,
\end{align}
uniformly for $(b,\im z)$ in compact subsets of $(0,\infty) \times \mathbb{R}$. Note that Assumption \ref{assumption} is made primarily for simplicity; we expect all our results to admit analytic continuations to more general values of the parameters. In the following subsection we use that $F$ can be defined for complex values of $\sigma_s$ and $\sigma_t$ by analytic continuation. 

\subsection{First pair of difference equations}
Define a translation operator $e^{\pm ib \partial_{\sigma_s}}$ which formally acts on a meromorphic function $f(\sigma_s)$ by $e^{\pm ib \partial_{\sigma_s}} f(\sigma_s)=f(\sigma_s \pm ib)$. Define the difference operator $H_F$ acting on the variable $\sigma_s$ by
\begin{equation} \label{HF}
H_F\left[\substack{\theta_1\;\;\theta_t\vspace{0.1cm}\\ \theta_{\infty}\;\theta_0};\substack{b,\;\sigma_s}\right] = H_F^+\left[\substack{\theta_1\;\;\theta_t\vspace{0.1cm}\\ \theta_{\infty}\;\theta_0};\substack{b,\;\sigma_s}\right] e^{ib\partial_{\sigma_s}}+H_F^+\left[\substack{\theta_1\;\;\theta_t\vspace{0.1cm}\\ \theta_{\infty}\;\theta_0};\substack{b,\; -\sigma_s}\right] e^{-ib\partial_{\sigma_s}}+H_F^0\left[\substack{\theta_1\;\;\theta_t\vspace{0.1cm}\\ \theta_{\infty}\;\theta_0};\substack{b,\;\sigma_s}\right],
\end{equation}
where
\beq \label{K}
H_F^+\left[\substack{\theta_1\;\;\theta_t\vspace{0.1cm}\\ \theta_{\infty}\;\theta_0};\substack{b,\;\sigma_s}\right]= \frac{4\pi^2\Gamma \lb 1+2b^2-2ib\sigma_s \rb \Gamma \lb b^2-2i b \sigma_s \rb \Gamma \lb -2i b \sigma_s \rb \Gamma \lb 1+b^2-2ib \sigma_s \rb}{\prod_{\epsilon,\epsilon'=\pm1}\Gamma \lb \tfrac{bQ}2-ib(\sigma_s+\epsilon \theta_1+\epsilon' \theta_\infty)\rb \Gamma \lb \tfrac{bQ}2-ib (\sigma_s+\epsilon \theta_0+\epsilon' \theta_t )\rb}
\eeq
and
\beq\label{H0} \begin{split}
H_F^0\left[\substack{\theta_1\;\;\theta_t\vspace{0.1cm}\\ \theta_{\infty}\;\theta_0};\substack{b,\;\sigma_s}\right]= & -2\cosh{( 2\pi b (\theta_1+\theta_t+\tfrac{ib}2))}\\
&+4 \displaystyle \sum_{k=\pm 1} \frac{\prod_{\epsilon=\pm 1} \cosh{(\pi b(\epsilon \theta_\infty-\tfrac{ib}2-\theta_1-k\sigma_s))} \cosh{(\pi b(\epsilon \theta_0-\tfrac{ib}2-\theta_t-k\sigma_s))}}{\operatorname{sinh}{\lb 2\pi b(k\sigma_s+\frac{ib}2)\rb}\operatorname{sinh}{\lb 2\pi b k \sigma_s \rb}}.
\end{split} \eeq
It was shown in \cite[Proposition 4.3]{R20} that the Virasoro fusion kernel $F$ is a joint eigenfunction of two copies of $H_F$:
\begin{subequations} \label{difference1} \begin{align} 
\label{differenceF1} & H_F\left[\substack{\theta_1\;\;\theta_t\vspace{0.1cm}\\ \theta_{\infty}\;\theta_0};\substack{b,\;\sigma_s}\right] F\left[\substack{\theta_1\;\;\theta_t\vspace{0.1cm}\\ \theta_{\infty}\;\theta_0};\substack{\sigma_s \vspace{0.15cm} \\  \sigma_t}\right] = 2\cosh{\lb 2\pi b \sigma_t\rb} ~ F\left[\substack{\theta_1\;\;\theta_t\vspace{0.1cm}\\ \theta_{\infty}\;\theta_0};\substack{\sigma_s \vspace{0.15cm} \\  \sigma_t}\right], \\
\label{differenceF2} & H_F\left[\substack{\theta_1\;\;\theta_t\vspace{0.1cm}\\ \theta_{\infty}\;\theta_0};\substack{b^{-1},\;\sigma_s}\right] F\left[\substack{\theta_1\;\;\theta_t\vspace{0.1cm}\\ \theta_{\infty}\;\theta_0};\substack{\sigma_s \vspace{0.15cm} \\  \sigma_t}\right] = 2\cosh{\lb 2\pi b^{-1} \sigma_t\rb} ~ F\left[\substack{\theta_1\;\;\;\theta_t\vspace{0.1cm}\\ \theta_{\infty}\;\;\theta_0};\substack{\sigma_s \vspace{0.15cm} \\  \sigma_t}\right].
\end{align}\end{subequations}

\subsection{Second pair of difference equations}
Define the dual difference operator $\tilde{H}_F$ acting on the variable $\sigma_t$ by 
\beq \label{tildeHF} \tilde H_F\left[\substack{\theta_0\;\;\theta_t\vspace{0.1cm}\\ \theta_{\infty}\;\theta_1};\substack{b,\;\sigma_t}\right] = \tilde H_F^+\left[\substack{\theta_0\;\;\theta_t\vspace{0.1cm}\\ \theta_{\infty}\;\theta_1};\substack{b,\;\sigma_t}\right] e^{ib\partial_{\sigma_t}}+\tilde H_F^+\left[\substack{\theta_0\;\;\theta_t\vspace{0.1cm}\\ \theta_{\infty}\;\theta_1};\substack{b,\;-\sigma_t}\right] e^{-ib\partial_{\sigma_t}}+H_F^0\left[\substack{\theta_0\;\;\theta_t\vspace{0.1cm}\\ \theta_{\infty}\;\theta_1};\substack{b,\;\sigma_t}\right], \eeq
where
  \beq \label{tildeK}\begin{split}
 \tilde H_F^+\left[\substack{\theta_0\;\;\theta_t\vspace{0.1cm}\\ \theta_{\infty}\;\theta_1};\substack{b,\;\sigma_t}\right] = \frac{4\pi^2~\Gamma \left(1-b^2+2 i b\sigma_t \right) \Gamma (1+2ib\sigma_t) \Gamma \left(2 i b \sigma_t-2 b^2\right) \Gamma \left(2 i b \sigma_t-b^2\right)}{\prod _{\epsilon,\epsilon'=\pm 1} \Gamma \left(\frac{1-b^2}2+ib \left(\sigma_t+\epsilon \theta_0+\epsilon' \theta_\infty\right)\right) \Gamma \left(\frac{1-b^2}2+ib \left(\sigma_t+\epsilon \theta_1+\epsilon' \theta_t\right)\right)}.
     \end{split}\eeq
It was shown in \cite[Proposition 4.4]{R20} that $F$ also satisfies the dual pair of difference equations
\begin{subequations}\label{difference2} \begin{align}
\label{differenceF3} & \tilde H_F\left[\substack{\theta_0\;\;\theta_t\vspace{0.1cm}\\ \theta_{\infty}\;\theta_1};\substack{b,\;\sigma_t}\right] F\left[\substack{\theta_1\;\;\theta_t\vspace{0.1cm}\\ \theta_{\infty}\;\theta_0};\substack{\sigma_s \vspace{0.15cm} \\  \sigma_t}\right] = 2\cosh{\lb 2\pi b \sigma_s\rb} ~ F\left[\substack{\theta_1\;\;\theta_t\vspace{0.1cm}\\ \theta_{\infty}\;\theta_0};\substack{\sigma_s \vspace{0.15cm} \\  \sigma_t}\right], \\
\label{differenceF4} & \tilde H_F\left[\substack{\theta_0\;\;\theta_t\vspace{0.1cm}\\ \theta_{\infty}\;\theta_1};\substack{b^{-1},\;\sigma_t}\right] F\left[\substack{\theta_1\;\;\theta_t\vspace{0.1cm}\\ \theta_{\infty}\;\theta_0};\substack{\sigma_s \vspace{0.15cm} \\  \sigma_t}\right] = 2\cosh{\lb 2\pi b^{-1} \sigma_s\rb} ~ F\left[\substack{\theta_1\;\;\theta_t\vspace{0.1cm}\\ \theta_{\infty}\;\theta_0};\substack{\sigma_s \vspace{0.15cm} \\  \sigma_t}\right].
\end{align}\end{subequations}

\section{Difference equations for the confluent Virasoro fusion kernels} \label{Ckdifferencesec}
In the previous section, we recalled the definition of the Virasoro fusion kernel and its eigenfunction properties. In this section, we describe the limit of the Virasoro fusion kernel leading to the family of confluent Virasoro fusion kernels $\mathcal{C}_k$, and we show that the kernel $\mathcal{C}_k$ satisfies two different pairs of difference equations for each $k$.

\subsection{Confluent Virasoro fusion kernels} 
The confluent Virasoro fusion kernels $\{\mathcal{C}_k\}_{k=1}^\infty$ were introduced in \cite{LR} as confluent limits of the Virasoro fusion kernel. Let $\boldsymbol{\theta}=(\theta_0,\theta_t,\theta_*) \in \mathbb{R}^3$ be a vector of parameters and let $\Delta(x)$ be the function defined by
\beq
\Delta(x)=\frac{Q^2}4+x^2.
\eeq
The kernel $\mathcal{C}_k$ is defined for any integer $k \geq 1$ by \cite[Eq. (5.5)]{LR}
\begin{equation}\label{gnm}
\begin{split}
\mathcal{C}_k\lb b,\boldsymbol{\theta},\nu, \sigma_s\rb  = P^{(k)}\lb \boldsymbol{\theta},\nu, \sigma_s\rb \displaystyle \int_{\mathsf{C}} dx ~ I^{(k)}\lb x,\boldsymbol{\theta},\nu, \sigma_s\rb\end{split}, 
\end{equation}
where the prefactor $P^{(k)}$ is given by\footnote{Complex powers are defined on the universal cover of $\mathbb{C}\setminus \{0\}$, i.e., $$(e^{2i\pi(\lfloor \frac{k}2 \rfloor-\frac12)}b)^\alpha = e^{2i\pi \alpha (\lfloor \frac{k}2 \rfloor-\frac12)}b^\alpha.$$}
\beq \label{P} \begin{split}
P^{(k)}\lb \boldsymbol{\theta},\nu, \sigma_s\rb = & \left(e^{2i\pi(\lfloor \frac{k}2 \rfloor-\frac12)}b\right)^{\Delta(\theta_0)+\Delta(\theta_t)-\Delta(\sigma_s)+\frac{\theta_*^2}2-2\nu^2}  
	\\
& \times \displaystyle \prod_{\epsilon=\pm 1}\frac{g_b\lb \epsilon \sigma_s-\theta_* \rb g_b \lb \epsilon \sigma_s-\theta_0-\theta_t\rb g_b \lb \epsilon \sigma_s+\theta_0-\theta_t \rb}{g_b \lb -\frac{iQ}{2}+2\epsilon \sigma_s \rb g_b\left(\nu-\frac{\theta_*}{2}+\epsilon \theta_0\right)g_b\left(-\theta_t+\epsilon(\nu+\frac{\theta_*}{2})\right)},
\end{split} \eeq
and the integrand $I^{(k)}$ is given by 
\beq \label{I}
I^{(k)}\lb x,\boldsymbol{\theta},\nu, \sigma_s\rb = e^{(-1)^{k+1}i\pi x\lb \frac{iQ}{2}+\frac{\theta_*}{2}+\theta_t+\nu \rb}\frac{s_b \lb x+\frac{\theta_*}{2}-\theta_t+\nu \rb}{s_b \big( x+ \frac{i Q}{2} \big) } \prod_{\epsilon=\pm 1} \frac{s_b \lb x+\epsilon \theta_0 +\nu-\frac{\theta_*}{2}\rb }{s_b\big( x+ \frac{iQ}{2}+\nu-\frac{\theta_*}{2}-\theta_t + \epsilon \sigma_s \big)}.
\eeq
The integration contour $\mathsf{C}$ in (\ref{gnm}) is defined as follows. As a function of $x \in \mathbb{C}$, the numerator in the integrand has three decreasing semi-infinite sequences of poles, while the denominator has three increasing semi-infinite sequences of zeros. The contour $\mathsf{C}$ in (\ref{gnm}) is any curve from $-\infty$ and $+\infty$ which separates the increasing from the decreasing sequences. In addition to the restrictions that $b>0$ and $(\theta_0,\theta_t,\theta_1,\theta_\infty) \in \mathbb{R}^4$ imposed by Assumption  \ref{assumption}, we temporarily also assume that $(\theta_*, \nu, \sigma_s) \in \mathbb{R}^3$. In this case, the contour of integration $\mathsf{C}$ can be any curve from $-\infty$ to $+\infty$ lying within the strip $\im{x} \in (-Q/2,0)$. Moreover, with this choice of $\mathsf{C}$, the integral in (\ref{gnm}) converges. Indeed, it follows from the asymptotic formula (\ref{sbasymptotics}) for $s_b$ that the integrand $I^{(k)}$ obeys the estimate
\begin{align}\label{Ikestimate}
I^{(k)}\lb x,\boldsymbol{\theta},\nu, \sigma_s\rb = O(e^{-\pi Q |\re x|}), \qquad \re x \to \pm \infty,
\end{align}
uniformly for $(b, \im x, \theta_0, \theta_t, \theta_*, \nu, \sigma_s)$ in compact subsets of $\mathbb{R}_{>0} \times \mathbb{R} \times \mathbb{C}^5$. As in the case of $F$, the functions $\mathcal{C}_k$ are defined for more general values of the variables by analytic continuation.


The kernels $\mathcal{C}_k$ are confluent limits of the Virasoro fusion kernel \eqref{fusion01}. Indeed, define the function $M$ by
\beq\label{Mdef}
M\left[\substack{\theta_0\;\;\theta_t\vspace{0.1cm}\\ \theta_\infty\;\theta_1};\substack{\sigma_t\vspace{0.15cm} \\  \sigma_s}\right] = e^{i\pi\lb \Delta(\sigma_t)-\Delta(\theta_1)-\Delta(\theta_t) \rb} ~
~ F\left[\substack{\theta_0\;\;\theta_t\vspace{0.1cm}\\ \theta_\infty\;\theta_1};\substack{\sigma_t\vspace{0.15cm} \\  \sigma_s}\right],
\eeq
where $F$ is defined by \eqref{fusion01}. 
Moreover, define the normalization factor $L_k(\Lambda,\nu,\sigma_s)$ for $\Lambda > 0$ and $k \geq 1$ by
\beq \label{Lj}
L_k(\Lambda,\nu,\sigma_s) = e^{-i \pi  \left(-\Delta \left(\frac{\theta_*+\Lambda }{2}\right)-\Delta (\theta_t)+\Delta \left(\frac{\Lambda }{2}-\nu \right)\right)} \lb e^{2 i \pi  (k-1)} ib\Lambda\rb^{\left(\Delta (\theta_0)+\Delta (\theta_t)-\Delta (\sigma_s)+\frac{\theta_*^2}{2}-2 \nu ^2\right)}.
\eeq
Then, for any integer $k \geq 1$, 
\beq \label{confluentlimit}
 \lim_{\Lambda \to + \infty} \lb L_k(\epsilon\Lambda,\nu,\sigma_s) M \left[\substack{\theta_0\;\;\;\;\;\;\;\;\theta_t \vspace{0.1cm}\\ \frac{\epsilon\Lambda-\theta_*}2\;\frac{\epsilon\Lambda+\theta_*}2};\substack{\frac{\epsilon\Lambda}2-\nu\vspace{0.15cm} \\ \sigma_s}\right] \rb = \begin{cases} \mathcal{C}_{2k} \lb b,\boldsymbol{\theta},\nu, \sigma_s\rb, & \epsilon=+1, \\
\mathcal{C}_{2k-1} \lb b,\boldsymbol{\theta},\nu, \sigma_s\rb,& \epsilon=-1.
 \end{cases}
 \eeq
The relation between $F$ and $\mathcal{C}_k$ was not stated in this form in \cite{LR}, but equation (\ref{confluentlimit}) is convenient for our present purposes and can be deduced from \cite[Section 6.1]{LR}.

In the remainder of this section, we show that the four difference equations in  \eqref{difference1} and \eqref{difference2} satisfied by the Virasoro fusion kernel survive in the confluent limit \eqref{confluentlimit}; this leads to four difference equations satisfied by $\mathcal{C}_k$ for each $k$. 
Since we seek to use (\ref{confluentlimit}), we rewrite the difference equations \eqref{difference1} and \eqref{difference2} for $F$ in terms of $M$ using (\ref{Mdef}). Define the difference operators $H_M$ and $\tilde H_M$ by
\begin{subequations}\begin{align}\nonumber
H_M\left[\substack{\theta_0\;\;\theta_t\vspace{0.1cm}\\ \theta_\infty \;\theta_1};\substack{b,\;\sigma_t}\right] = &\; e^{2\pi b(\sigma_t+\frac{ib}2)} H_F^+\left[\substack{\theta_0\;\;\theta_t\vspace{0.1cm}\\ \theta_{\infty}\;\theta_1};\substack{b,\;\sigma_t}\right] e^{i b \partial_{\sigma_t}} 
+ e^{2\pi b(-\sigma_t+\frac{ib}2)} H_F^+\left[\substack{\theta_0\;\;\theta_t\vspace{0.1cm}\\ \theta_{\infty}\;\theta_1};\substack{b,\;-\sigma_t}\right] e^{-i b \partial_{\sigma_t}}
	\\ \label{HM}  
& + H_F^0\left[\substack{\theta_0\;\;\theta_t\vspace{0.1cm}\\ \theta_{\infty}\;\theta_1};\substack{b,\;\sigma_t}\right], 
	\\ \label{HMtilde}  
\tilde H_M\left[\substack{\theta_1\;\;\theta_t\vspace{0.1cm}\\ \theta_\infty \;\theta_0};\substack{b,\;\sigma_s}\right] = &\; \tilde H_F\left[\substack{\theta_1\;\;\theta_t\vspace{0.1cm}\\ \theta_{\infty}\;\theta_0};\substack{b,\;\sigma_s}\right].
\end{align} \end{subequations}
It follows from \eqref{difference1} and \eqref{difference2}  that $M$ satisfies the pair of difference equations
\begin{subequations}\label{differenceM}\begin{align}
\label{differenceM1} & H_M\left[\substack{\theta_0\;\;\theta_t\vspace{0.1cm}\\ \theta_\infty \;\theta_1};\substack{b,\;\sigma_t}\right]~M \left[\substack{\theta_0\;\;\theta_t\vspace{0.1cm}\\ \theta_\infty\;\theta_1};\substack{\sigma_t\vspace{0.15cm} \\  \sigma_s}\right] = 2\cosh{(2\pi b \sigma_s)}~M \left[\substack{\theta_0\;\;\theta_t\vspace{0.1cm}\\ \theta_\infty\;\theta_1};\substack{\sigma_t\vspace{0.15cm} \\  \sigma_s}\right], \\
\label{differenceM2} & \tilde H_M\left[\substack{\theta_1 \;\;\theta_t\vspace{0.1cm}\\ \theta_\infty \;\theta_0};\substack{b,\;\sigma_s}\right]~M \left[\substack{\theta_0\;\;\theta_t\vspace{0.1cm}\\ \theta_\infty\;\theta_1};\substack{\sigma_t\vspace{0.15cm} \\  \sigma_s}\right] = 2\cosh{(2\pi b \sigma_t)}~M \left[\substack{\theta_0\;\;\theta_t\vspace{0.1cm}\\ \theta_\infty\;\theta_1};\substack{\sigma_t\vspace{0.15cm} \\  \sigma_s}\right],
\end{align}\end{subequations}
as well as the pair of difference equations obtained by replacing $b \to b^{-1}$ in \eqref{differenceM}.

\subsection{First pair of difference equations}\label{section6p2}  The first pair of difference equations for $\mathcal{C}_k$ is found by applying the limit \eqref{confluentlimit} to the difference equation \eqref{differenceM1}. Equation \eqref{differenceM1} can be written as
\beq\label{4p8}\begin{split}
\lb L_k(\epsilon\Lambda,\nu,\sigma_s) H_M\left[\substack{\theta_0\;\;\theta_t\vspace{0.1cm}\\ \theta_\infty \;\theta_1};\substack{b,\;\sigma_t}\right]  L_k(\epsilon\Lambda,\nu,\sigma_s)^{-1}\rb 
&L_k(\epsilon\Lambda,\nu,\sigma_s) M \left[\substack{\theta_0\;\;\theta_t\vspace{0.1cm}\\ \theta_\infty\;\theta_1};\substack{\sigma_t\vspace{0.15cm} \\  \sigma_s}\right] 
	\\
     & = 2\cosh{\lb 2\pi b \sigma_s\rb} ~ L_k(\epsilon\Lambda,\nu,\sigma_s) M \left[\substack{\theta_0\;\;\theta_t\vspace{0.1cm}\\ \theta_\infty\;\theta_1};\substack{\sigma_t\vspace{0.15cm} \\  \sigma_s}\right].
\end{split}\eeq
Introduce the difference operator $H_{\mathcal{C}_k}$ by
\beq\label{Hcn}
H_{\mathcal{C}_k}(b,\nu) = H_{\mathcal{C}_k}^{+}(\nu) e^{ib \partial_{\nu}} + H^{-}_{\mathcal{C}_k}(\nu) e^{-ib \partial_{\nu}}+H^{0}_{\mathcal{C}_k}(\nu),
\eeq
where
\beq \label{Hcnplusminus} \begin{split}
& H_{\mathcal{C}_k}^{+}(\nu) = \frac{4 \pi ^2 e^{-4 \pi b \lfloor \frac{k-1}2 \rfloor (i b+2\nu)}}{\prod _{\epsilon =\pm1} \Gamma \lb \frac{bQ}2 +ib \lb \epsilon \theta_0+\frac{\theta_*}2-\nu \rb \rb \Gamma \lb \frac{bQ}2+ib\lb \epsilon \theta_t-\frac{\theta_*}2-\nu \rb \rb}, \\
& H_{\mathcal{C}_k}^{-}(\nu) = \frac{4 \pi ^2 e^{-4\pi b\lb \lfloor \frac{k}2 \rfloor-\frac12\rb (ib-2\nu)}}{\prod _{\epsilon =\pm1} \Gamma \lb \frac{bQ}2 +ib \lb \epsilon \theta_0-\frac{\theta_*}2+\nu \rb \rb \Gamma \lb \frac{bQ}2+ib\lb \epsilon \theta_t+\frac{\theta_*}2+\nu \rb \rb}, 
\end{split} \eeq
and
\beq \label{Hn0} \begin{split}
 H_{\mathcal{C}_k}^{0}(\nu) = & ~ 4 e^{(-1)^k\pi  b (\theta_0+\theta_t+2 \nu )} \cosh(\pi  b(\tfrac{i b}{2}+\theta_0+\tfrac{\theta_*}{2}-\nu)) \cosh(\pi  b(\tfrac{i b}{2}-\tfrac{\theta_*}{2}+\theta_t-\nu)) \\ 
&+4 e^{(-1)^{k+1}\pi b (\theta_0+\theta_t-2 \nu )} \cosh(\pi  b(\tfrac{i b}{2}+\theta_t+\tfrac{\theta_*}{2}+\nu)) \cosh(\pi  b(\tfrac{i b}{2}-\tfrac{\theta_*}{2}+\theta_0+\nu))\\
& -2 \cosh(2 \pi  b(\tfrac{i b}{2}+\theta_0+\theta_t)).
\end{split}\eeq
The next lemma shows that $H_{\mathcal{C}_k}$ is the limit as $\Lambda \to + \infty$ of the operator in round brackets on the left-hand side of \eqref{4p8} evaluated with
\beq\label{paramconf}
\theta_\infty = \frac{\epsilon\Lambda-\theta_*}2, \quad \theta_1 = \frac{\epsilon\Lambda+\theta_*}2, \quad \sigma_t = \frac{\epsilon\Lambda}2-\nu.
\eeq

\begin{lemma}\label{confluentlimHMlemma}
For each integer $k \geq 1$,
\beq\label{confluentlimHM}
\lim_{\Lambda\to +\infty}\lb L_k(\epsilon\Lambda,\nu,\sigma_s) H_M\left[\substack{\theta_0\;\;\;\;\;\;\;\theta_t\vspace{0.1cm}\\\frac{\epsilon \Lambda-\theta_*}2  \;\;\frac{\epsilon \Lambda+\theta_*}2};\substack{b,\;\frac{\epsilon \Lambda}2-\nu}\right]  L_k(\epsilon\Lambda,\nu,\sigma_s)^{-1}\rb=\begin{cases} H_{\mathcal{C}_{2k}}(b,\nu), & \epsilon=+1, \\ H_{\mathcal{C}_{2k-1}}(b,\nu),& \epsilon=-1. \end{cases}
\eeq
\end{lemma}
\begin{proof}
Let us consider the case $\epsilon=+1$. Using \eqref{HM}, we can write
\beq\label{6p14}\begin{split}
& L_k(\Lambda,\nu,\sigma_s)  H_M\left[\substack{\theta_0\;\;\;\;\;\;\;\theta_t\vspace{0.1cm}\\\frac{ \Lambda-\theta_*}2  \;\;\frac{\Lambda+\theta_*}2};\substack{b,\;\frac{\Lambda}2-\nu}\right]  L_k(\Lambda,\nu,\sigma_s)^{-1} = H_F^0\left[\substack{\theta_0 \quad\;\;\;\theta_t\vspace{0.1cm}\\ \frac{\Lambda-\theta_*}2 \;\; \frac{\Lambda+\theta_*}2 };\substack{b,\;\frac{\Lambda}2-\nu}\right]  \\
& + e^{\pi b \Lambda} \frac{L_k(\Lambda,\nu,\sigma_s)}{L_k(\Lambda,\nu-ib,\sigma_s)} H_F^+\left[\substack{\theta_0 \quad\;\;\;\theta_t\vspace{0.1cm}\\ \frac{\Lambda-\theta_*}2 \;\; \frac{\Lambda+\theta_*}2 };\substack{b,\;\frac{\Lambda}2-\nu}\right] e^{\pi b(ib-2\nu)} e^{-i b \partial_\nu} \\
& + e^{-\pi b \Lambda} \frac{L_k(\Lambda,\nu,\sigma_s)}{L_k(\Lambda,\nu+ib,\sigma_s)} H_F^+\left[\substack{\theta_0 \quad\;\;\;\theta_t\vspace{0.1cm}\\ \frac{\Lambda-\theta_*}2 \;\; \frac{\Lambda+\theta_*}2 };\substack{b,\;-\frac{\Lambda}2+\nu}\right] e^{\pi b(ib+2\nu)} e^{i b \partial_\nu},
\end{split}\eeq
where the coefficients $H_F^+$ and $L_k$ are given in \eqref{K} and \eqref{Lj}, respectively, and $H_0$ is defined in \eqref{H0}. It is straightforward to verify that \eqref{6p14} can be brought to the following form:
 \beq\label{6p15}\begin{split}
& L_k(\Lambda,\nu,\sigma_s)  H_M\left[\substack{\theta_0\;\;\;\;\;\;\;\theta_t\vspace{0.1cm}\\\frac{ \Lambda-\theta_*}2  \;\;\frac{\Lambda+\theta_*}2};\substack{b,\;\frac{\Lambda}2-\nu}\right] L_k(\Lambda,\nu,\sigma_s)^{-1}   \\
&= H_F^0\left[\substack{\theta_0 \quad\;\;\;\theta_t\vspace{0.1cm}\\ \frac{\Lambda-\theta_*}2 \;\; \frac{\Lambda+\theta_*}2 };\substack{b,\;\frac{\Lambda}2-\nu}\right]  + X_{+1}(\Lambda,\nu) ~H_{\mathcal{C}_{2k}}^{+}(\nu)~e^{ib\partial_\nu} + X_{-1}(\Lambda,\nu)~H_{\mathcal{C}_{2k}}^{-}(\nu)~e^{-ib\partial_\nu},
\end{split}\eeq
where $H_{\mathcal{C}_k}^{\pm}(\nu)$ are given in \eqref{Hcnplusminus} and the coefficients $X_j(\Lambda,\nu)$, $j = \pm 1$, are defined by
\beq\label{Xk}\begin{split}
&X_j(\Lambda,\nu)=(j i b \Lambda )^{-2 b (b-2 i j \nu )}\\
&\times\frac{\Gamma (b (b+i j (\Lambda -2 \nu ))) \Gamma \left(b^2+i j (\Lambda -2 \nu ) b+1\right) \Gamma \left(2 b^2+i j (\Lambda -2 \nu ) b+1\right) \Gamma (j i b (\Lambda -2 \nu ))}{\prod_{\epsilon=\pm1}\Gamma \left(\frac{b Q}{2}-j i b \left(\epsilon \theta_0+\frac{\theta_*}{2}-\Lambda +\nu \right)\right) \Gamma \left(\frac{b Q}{2}+b i j \left(\frac{\theta_*}{2}+\epsilon  \theta_t+\Lambda -\nu \right)\right)}.
\end{split}\eeq
It remains to compute the limit $\Lambda \to +\infty$ of \eqref{6p15}. First, the asymptotic formula
\beq \label{asympgamma}
\Gamma(z+a) \sim \sqrt{2\pi}~z^{z+a-\frac12}~e^{-z}, \qquad z \to \infty, ~ z + a \in \mathbb{C} \setminus \mathbb{R}_{\leq 0}, ~  |a| < |z|
\eeq
shows that
\beq\label{Xpmlimit}
\lim_{\Lambda \to +\infty} X_{\pm1}(\Lambda,\nu) = 1.
\eeq
Second, the first term on the right-hand side of \eqref{6p15} takes the form
\beq\label{limH01}\begin{split}
H_F^0\left[\substack{\theta_0 \quad\;\;\;\theta_t\vspace{0.1cm}\\ \frac{\Lambda-\theta_*}2 \;\; \frac{\Lambda+\theta_*}2 };\substack{b,\;\frac{\Lambda}2-\nu}\right]  
= & -2 \cosh(2 \pi  b(\tfrac{i b}{2}+\theta_0+\theta_t)) 
	\\
& + J_{+1}(\Lambda,\nu) \cosh(\pi  b(-\tfrac{ib}2-\theta_0-\tfrac{\theta_*}{2}+\nu)) \cosh(\pi  b(-\tfrac{ib}2+\tfrac{\theta_*}{2}-\theta_t+\nu)) \\
& + J_{-1}(\Lambda,\nu) \cosh(\pi b(-\tfrac{ib}2-\theta_0+\tfrac{\theta_*}{2}-\nu)) \cosh(\pi  b(\tfrac{i b}{2}+\tfrac{\theta_*}{2}+\theta_t+\nu)),
\end{split} \eeq
where the coefficients $J_j(\Lambda,\nu)$, $j=\pm1$, are defined by
\beq\label{Jk}
J_j(\Lambda,\nu) = 4j \frac{\cosh \left(\pi  b \left(-\frac{ib}2-\theta_0+\frac{\theta_* j}{2}-j \Lambda +j \nu \right)\right) \cosh \left(\pi  b \left(-\frac{ib}2-\theta_t-\frac{\theta_* j}{2}-j \Lambda +j \nu \right)\right)}{\sinh (\pi  b (\Lambda -2 \nu )) \sinh (\pi  b (i b+j \Lambda -2 j \nu ))}.
\eeq
The limit of $J_j$ as $\Lambda \to +\infty$ is easily computed by expressing the hyperbolic functions in terms of exponentials:
\beq\label{limitjk}
\lim_{\Lambda\to+\infty}J_j(\Lambda,\nu) = 4 e^{j \pi b(\theta_0+\theta_t+2j\nu)}, \qquad j = \pm 1.
\eeq
It follows from \eqref{limH01} and \eqref{limitjk} that
\beq\label{HF0limit}
\lim_{\Lambda\to+\infty} H_F^0\left[\substack{\theta_0 \quad\;\;\;\theta_t\vspace{0.1cm}\\ \frac{\Lambda-\theta_*}2 \;\; \frac{\Lambda+\theta_*}2 };\substack{b,\;\frac{\Lambda}2-\nu}\right]  = H_{\mathcal{C}_{2k}}^{0}(\nu). 
\eeq
Using (\ref{Xpmlimit}) and (\ref{HF0limit}), we can compute the limit of (\ref{6p15}) as $\Lambda \to +\infty$. Comparing the result with (\ref{Hcn}), we obtain \eqref{confluentlimHM} for $\epsilon=+1$. The case $\epsilon=-1$ is treated in a similar way.
\end{proof}

We can now state the first pair of difference equations for $\mathcal{C}_k$.

\begin{theorem}[First pair of difference equations for $\mathcal{C}_k$]\label{thm6p2} 
For each integer $k \geq 1$, the confluent fusion kernel $\mathcal{C}_k \lb b,\boldsymbol{\theta},\nu, \sigma_s\rb$ defined in \eqref{gnm} satisfies the following pair of difference equations: \begin{subequations} \label{eigenvaluecn} 
\begin{align} \label{eigenvaluecn1} & H_{\mathcal{C}_k}(b,\nu) ~ \mathcal{C}_k \lb b,\boldsymbol{\theta},\nu, \sigma_s\rb = 2\cosh{(2\pi b \sigma_s)} ~ \mathcal{C}_k \lb b,\boldsymbol{\theta}, \nu, \sigma_s\rb, \\
\label{eigenvaluecn2} & \lb b^{-4\nu^2} H_{\mathcal{C}_k}(b^{-1},\nu) b^{4\nu^2}\rb ~ \mathcal{C}_k \lb b,\boldsymbol{\theta},\nu, \sigma_s\rb = 2\cosh{(2\pi b^{-1} \sigma_s)} ~ \mathcal{C}_k \lb b, \boldsymbol{\theta},\nu, \sigma_s\rb,
\end{align}
\end{subequations}
where $H_{\mathcal{C}_k}$ is given in \eqref{Hcn}.
\end{theorem}
\begin{proof} 
We only have to prove \eqref{eigenvaluecn1}, because the confluent fusion kernel satisfies
\beq\label{4p25}
\mathcal{C}_k \lb b^{-1},\boldsymbol{\theta},\nu, \sigma_s\rb = b^{-2(\Delta(\theta_0)+\Delta(\theta_t)-\Delta(\sigma_s)+\frac{\theta_*^2}2-2\nu^2)} \mathcal{C}_k \lb b,\boldsymbol{\theta},\nu, \sigma_s\rb.
\eeq
We will present two different proofs of (\ref{eigenvaluecn1}). The first proof has the advantage of being constructive and is based on Lemma \ref{confluentlimHMlemma}; this is the approach we first used to arrive at (\ref{eigenvaluecn}). The second proof is more direct and avoids the use of Lemma \ref{confluentlimHMlemma} and of confluent limits, but it assumes that the structure of (\ref{eigenvaluecn}) is already known.

{\it First proof of (\ref{eigenvaluecn1}).} With the help of \eqref{confluentlimit} and (\ref{confluentlimHM}), we can easily compute the limit $\Lambda \to +\infty$ of equation \eqref{4p8} for any integer $k \geq 1$ and $\epsilon = \pm1$. This gives (\ref{eigenvaluecn1}).

{\it Second proof of (\ref{eigenvaluecn1}).} 
Let us rewrite the integral representation \eqref{gnm} for $\mathcal{C}_k$ as follows:
\beq \label{cn}
\mathcal{C}_k \lb b,\boldsymbol{\theta},\nu, \sigma_s\rb=\displaystyle \int_{\mathsf{C}} dx ~ X_k(x,\nu) Y_k(x,\sigma_s) Z_k(x),
\eeq
where the dependence of the functions $X_k,Y_k,Z_k$ on $\boldsymbol{\theta}$ is omitted for simplicity. Performing the change of variables $x \to x-\nu$ in \eqref{gnm}, we find that $X_k,Y_k,Z_k$ are given by
\begin{align} \nonumber 
X_k(x,\nu)= &\; b^{-2\nu^2} e^{-4i\pi \nu^2(\lfloor \frac{k}2\rfloor-\frac12)}e^{(-1)^k i \pi \nu(\frac{iQ}2-x+\frac{\theta_*}2+\theta_t+\nu)} \tfrac{s_b(x+\frac{iQ}2-\nu)^{-1}}{\prod_{\epsilon=\pm1}g_b\left(\nu-\frac{\theta_*}{2} +\epsilon\theta_0\right) g_b\left(\epsilon \left(\frac{\theta_*}{2}+\nu \right)-\theta_t\right)},
	\\\nonumber
Y_k(x,\sigma_s) = &\; b^{-\Delta (\sigma_s)} e^{-2 i \pi  \Delta (\sigma_s) \left(\left\lfloor \frac{k}{2}\right\rfloor -\frac{1}{2}\right)}  \prod_{\epsilon=\pm1} \tfrac{g_b \lb \epsilon \sigma_s-\theta_* \rb g_b \lb \epsilon \sigma_s-\theta_0-\theta_t \rb g_b \lb \epsilon \sigma_s+\theta_0-\theta_t \rb}{g_b \lb 2\epsilon \sigma_s-\frac{iQ}2 \rb s_b \lb x+\frac{iQ}2-\frac{\theta_*}2-\theta_t+\epsilon \sigma_s \rb}, 
	\\\nonumber
Z_k(x) = &\; e^{i \pi  (-1)^{k+1} x \left(\frac{\theta_*}{2}+\theta_t+\frac{i Q}{2}\right)} b^{\Delta (\theta_0)+\Delta (\theta_t)+\frac{\theta_*^2}{2}} e^{2 i \pi  \left(\left\lfloor \frac{k}{2}\right\rfloor -\frac{1}{2}\right) \left(\Delta (\theta_0)+\Delta (\theta_t)+\frac{\theta_*^2}{2}\right)} 
	\\ \label{XYZ}
& \times s_b(x-\theta_0-\tfrac{\theta_*}2) s_b(x+\theta_0-\tfrac{\theta_*}2) s_b(x+\tfrac{\theta_*}2-\theta_t). 
\end{align}
As a consequence of the relations
\beq \label{gbsbdifferenceeqs}
\frac{g_b \lb z+\frac{ib}{2}\rb}{g_b\lb z-\frac{ib}{2}\rb}=\frac{b^{-ibz}\sqrt{2\pi}}{\Gamma \lb \frac{1}{2}-ibz \rb}, \qquad
\frac{s_b(z+\frac{ib}{2})}{s_b(z-\frac{ib}{2})}=2\cosh{\pi b z},
\eeq
the following identity follows from long but straightforward computations:
\beq \label{HnX} 
\frac{H_{\mathcal{C}_k}(b,\nu) X_k \lb x, \nu\rb}{X_k \lb x, \nu\rb} = 2\cosh{\lb \pi b (2x-\theta_*-2\theta_t)\rb} + \psi_k(x,\nu), \eeq
where
\beq
\psi_k(x,\nu) = -4 i e^{(-1)^k \pi b\lb \theta_t+\nu+\frac{iQ}2+\tfrac{\theta_*}2 \rb} \frac{\cosh{\lb \pi b (x+\frac{ib}2+\tfrac{\theta_*}2-\theta_t) \rb}}{\operatorname{sinh}{\lb \pi b(x+ib-\nu)\rb}} \displaystyle \prod_{\epsilon=\pm1} \cosh{(\pi b( x+\tfrac{ib}2+\epsilon \theta_0-\tfrac{\theta_*}2))}.\eeq
Using \eqref{HnX}, we obtain
\begin{align} \nonumber
H_{\mathcal{C}_k}(b,\nu) ~ \mathcal{C}_k \lb b,\boldsymbol{\theta},\nu, \sigma_s\rb = &  \int_{\mathsf{C}} dx ~ 2\cosh{\lb \pi b (2x-\theta_*-2\theta_t)\rb}   X_k(x,\nu) Y_k(x,\sigma_s) Z_k(x) 
	\\ \label{HnCn} 
& +  \int_{\mathsf{C}} dx ~ \psi_k(x,\nu) X_k(x,\nu) Y_k(x,\sigma_s) Z_k(x).
\end{align}
Moreover, the identity satisfied by $s_b$ in \eqref{gbsbdifferenceeqs} implies that the building blocks $X_k,Y_k,Z_k$ of the integral possess the following properties:
\begin{subequations} \label{identities} 
\begin{align}\label{identitiesX}
&\frac{X_k(x-ib,\nu)}{X_k(x,\nu)}=2ie^{(-1)^{k+1}\pi  b \nu  } \sinh (\pi  b (x-\nu )),
	\\\label{identitiesY}
&\frac{Y_k(x-ib,\sigma_s)}{Y_k(x,\sigma_s)}=2\cosh{\lb 2\pi b \sigma_s \rb} - 2\cosh{\lb \pi b(2x-\theta_*-2\theta_t) \rb},
	\\\label{identitiesZ}
&\frac{Z_k(x-ib)}{Z_k(x)} = \frac{e^{(-1)^{k+1}\frac{\pi b}{2} (\theta_*+2 \theta_t+i Q)}}{8\cosh{\lb \pi b(x-\frac{ib}2-\theta_0-\frac{\theta_*}2) \rb} \cosh{\lb \pi b(x-\frac{ib}2+\theta_0-\frac{\theta_*}2) \rb} \cosh{\lb \pi b (x-\frac{ib}2+\frac{\theta_*}2-\theta_t \rb}}.
\end{align}
\end{subequations}
Performing the change of variables $x \to x-ib$ in the second integral in \eqref{HnCn} and using \eqref{identitiesY}, we obtain
\begin{align}\nonumber 
H_{\mathcal{C}_k}(b,\nu) \mathcal{C}_k \lb\nu, \sigma_s\rb = & \int_{\mathsf{C}} dx ~ 2\cosh{\lb \pi b (2x-\theta_*-2\theta_t)\rb}  X_k(x,\nu) Y_k(x,\sigma_s) Z_k(x) 
	\\ \nonumber
& - \displaystyle \int_{\mathsf{C}} dx ~ 2\cosh{\lb \pi b (2x-\theta_*-2\theta_t)\rb} \psi_k(x-ib,\nu) X_k(x-ib,\nu) Y_k(x,\sigma_s) Z_k(x-ib) 
	\\ \label{4p33} 
& + 2\cosh{\lb 2\pi b \sigma_s\rb} \displaystyle \int_{\mathsf{C}} dx ~ \psi_k(x-ib,\nu) X_k(x-ib,\nu) Y_k(x,\sigma_s) Z_k(x-ib).
\end{align}
On the other hand, the identities \eqref{identitiesX} and \eqref{identitiesZ} imply that 
\beq\label{psinXnZn}
\psi_k(x-ib,\nu) X_k(x-ib,\nu) Z_k(x-ib) = X_k(x,\nu) Z_k(x).
\eeq
Equation \eqref{psinXnZn} ensures that the first two lines in \eqref{4p33} cancel and thus \eqref{eigenvaluecn1} follows from \eqref{4p33}. 
 \end{proof}

\subsection{Second pair of difference equations}
To derive the second pair of difference equations satisfied by $\mathcal{C}_k$, we rewrite \eqref{differenceM2} as
\beq\label{conflim2}\begin{split}
\lb e^{-\pi b \Lambda}  L_k(\Lambda,\nu,\sigma_s) \tilde H_M\left[\substack{\theta_1\;\;\theta_t\vspace{0.1cm}\\ \theta_\infty\;\theta_0};\substack{b,\;\sigma_s}\right]  L_k(\Lambda,\nu,\sigma_s)^{-1}\rb 
& L_k(\Lambda,\nu,\sigma_s)
 M\left[\substack{\theta_0\;\;\theta_t\vspace{0.1cm}\\ \theta_\infty\;\theta_1};\substack{\sigma_t\vspace{0.15cm} \\  \sigma_s}\right] 
 	\\
&  = 2e^{-\pi b \Lambda}\cosh{(2\pi b \sigma_t)}L_k(\Lambda,\nu,\sigma_s)M \left[\substack{\theta_0\;\;\theta_t\vspace{0.1cm}\\ \theta_\infty\;\theta_1};\substack{\sigma_t\vspace{0.15cm} \\  \sigma_s}\right].
 \end{split} \eeq
  It is easy to observe that
 \beq
 \lim_{\Lambda \to +\infty} 2e^{-\pi b \Lambda}\cosh{(2\pi b \sigma_t)}|_{\sigma_t = \frac{\epsilon\Lambda}2-\nu}=\begin{cases} e^{-2\pi b \nu}, & \epsilon=+1, \\  e^{2\pi b \nu}, & \epsilon = -1.  \end{cases}
 \eeq
Introduce the dual difference operator $\tilde{H}_{\mathcal{C}_k}$ by
\beq\label{Hcntilde}
\tilde{H}_{\mathcal{C}_k}(b,\sigma_s) =  \tilde{H}^+_{\mathcal{C}_k}(\sigma_s) e^{ib \partial_{\sigma_s}} +  \tilde{H}^+_{\mathcal{C}_k}(-\sigma_s) e^{-ib \partial_{\sigma_s}} + \tilde{H}^0_{\mathcal{C}_k}(\sigma_s),
\eeq
where
\beq \begin{split}
\tilde{H}^+_{\mathcal{C}_k}(\sigma_s) =&\; 2\pi e^{-2 \pi  b \left(\sigma_s+\frac{i b}{2}\right) \left(\left\lfloor \frac{k-1}{2}\right\rfloor +\left\lfloor \frac{k}{2}\right\rfloor -\frac{1}{2}\right)} \\
& \times \frac{\Gamma (1+2ib \sigma_s) \Gamma \left(1-b^2+2 ib \sigma_s\right) \Gamma (-2 b (b-i \sigma_s)) \Gamma (-b (b-2 i \sigma_s))}{\prod_{\epsilon_1=\pm1}\left\{\Gamma \left(\frac{1-b^2}2-i b \left(\epsilon_1 \theta_*-\sigma_s\right)\right) \prod_{\epsilon_2=\pm1} \Gamma \left(\frac{1-b^2}2-i b \left(\epsilon_1 \theta_0+\epsilon_2 \theta_t-\sigma_s\right)\right)\right\}},
\end{split} \eeq
and 
\beq \label{Htilden0} \begin{split}
\tilde{H}^0_{\mathcal{C}_k}(\sigma_s) &=-e^{(-1)^k \pi b(ib+\theta_*+2\theta_t)}+V_k(\sigma_s,\theta_t)+V_k(-\sigma_s,\theta_t),
\end{split} \eeq
with
\beq\label{Vkdef}
V_k(\sigma_s,\theta_t) = 2 e^{(-1)^{k+1}\pi b\left(\sigma_s-\frac{i b}{2}\right)}\cosh(\pi  b (\tfrac{i b}{2}+\theta_*-\sigma_s))\frac{\prod _{\epsilon =\pm1} \cosh \left(\pi  b \left(-\frac{ib}{2}-\theta_t+\sigma_s+\epsilon\theta_0 \right)\right)}{\sinh (\pi  b (2 \sigma_s-i b)) \sinh (2 \pi  b \sigma_s)}.
\eeq
The next lemma shows that $\tilde{H}_{\mathcal{C}_k}$ is the limit as $\Lambda \to + \infty$ of the operator in round brackets on the left-hand side of \eqref{conflim2} with the parameters chosen as in (\ref{paramconf}).

\begin{lemma}\label{confluentlimHMlemma2}
For each integer $k \geq 1$,
\beq\label{limtilde2}
\lim_{\Lambda \to +\infty} \lb  e^{-\pi b \Lambda}  L_k(\Lambda,\nu,\sigma_s) \tilde H_M\left[\substack{\frac{\epsilon\Lambda+\theta_*}2 \;\;\theta_t\vspace{0.1cm}\\ \frac{\epsilon\Lambda-\theta_*}2 \;\; \theta_0};\substack{b,\;\sigma_s}\right]  L_k(\Lambda,\nu,\sigma_s)^{-1}\rb = \begin{cases}
     \tilde{H}_{\mathcal{C}_{2k}}(b,\sigma_s), & \epsilon=+1, \\ 
     \tilde{H}_{\mathcal{C}_{2k-1}}(b,\sigma_s), & \epsilon=-1.
     \end{cases} 
\eeq
\begin{proof}
The proofs for $\epsilon=+1$ and $\epsilon=-1$ involve similar computations, so we only give the proof for $\epsilon=+1$. Using \eqref{HMtilde} and \eqref{Lj}, we  can write 
\beq \begin{split}
& e^{-\pi b \Lambda}  L_k(\Lambda,\nu,\sigma_s) \tilde H_M\left[\substack{\frac{\Lambda+\theta_*}2 \;\;\theta_t\vspace{0.1cm}\\ \frac{\Lambda-\theta_*}2 \;\; \theta_0};\substack{b,\;\sigma_s}\right]   L_k(\Lambda,\nu,\sigma_s)^{-1} \\
     & = e^{-\pi  b \Lambda } \chi(\Lambda,\sigma_s) \tilde{H}^+_{\mathcal{C}_{2k}}(\sigma_s) + e^{-\pi  b \Lambda} \chi(\Lambda,-\sigma_s) \tilde{H}^+_{\mathcal{C}_{2k}}(-\sigma_s) + e^{-\pi b \Lambda} H_F^0\left[\substack{\frac{\Lambda+\theta_*}2\;\;\theta_t\vspace{0.1cm}\\ \frac{\Lambda-\theta_*}2 \;\; \theta_0 };\substack{b,\;\sigma_s}\right],
\end{split}\eeq
where $H_0$ is given in \eqref{H0} and
\beq\label{chi}
\chi(\Lambda,\sigma_s)=\frac{2\pi e^{\pi  b \left(\sigma_s+\frac{i b}{2}\right)} (i b \Lambda )^{-b (b-2 i \sigma_s)}}{\Gamma \left(\frac{1-b^2}2-i b (\Lambda -\sigma_s)\right) \Gamma \left(\frac{1-b^2}2+ib (\Lambda +\sigma_s)\right)}.
\eeq
The asymptotics \eqref{asympgamma} of the gamma function implies that $\lim_{\Lambda\to+\infty} \chi(\Lambda,\pm\sigma_s)=e^{\pi b \Lambda}$.
Moreover, using that
\beq\begin{split}
& e^{-\pi b \Lambda} H_F^0\left[\substack{\frac{\Lambda+\theta_*}2\;\;\theta_t\vspace{0.1cm}\\ \frac{\Lambda-\theta_*}2 \;\; \theta_0 };\substack{b,\;\sigma_s}\right] = -2e^{-\pi b \Lambda} \cosh (\pi  b (i b+\theta_*+2 \theta_t+\Lambda ))  \\
     & +2e^{-\pi b \Lambda}e^{\pi b(\sigma_s-\frac{ib}2)}\cosh{(\pi b(\tfrac{ib}2+\Lambda-\sigma_s)}V_{2k}(\sigma_s,\theta_t) + 2e^{-\pi b \Lambda}e^{-\pi b(\sigma_s+\frac{ib}2)}\cosh{(\pi b(\tfrac{ib}2+\Lambda+\sigma_s)}V_{2k}(-\sigma_s,\theta_t),
\end{split}\eeq
where $V_k$ is given in \eqref{Vkdef}, it is easy to verify that
\beq
\lim_{\Lambda\to+\infty} e^{-\pi b \Lambda} H_F^0\left[\substack{\frac{\Lambda+\theta_*}2\;\;\theta_t\vspace{0.1cm}\\ \frac{\Lambda-\theta_*}2 \;\; \theta_0 };\substack{b,\;\sigma_s}\right] = \tilde{H}^0_{\mathcal{C}_{2k}}(\sigma_s).
\eeq
Recalling the definition (\ref{Hcntilde}) of $\tilde{H}_{\mathcal{C}_k}$, this proves (\ref{limtilde2}) for $\epsilon = +1$.
\end{proof}
\end{lemma}

The following theorem follows from Lemma \ref{confluentlimHMlemma2} and (\ref{conflim2}) in the same way that Theorem \ref{thm6p2} followed from Lemma \ref{confluentlimHMlemma} and (\ref{4p8}). 

\begin{theorem}[Second pair of difference equations for $\mathcal{C}_k$]\label{thm6p4} 
For each integer $k \geq 1$, the confluent fusion kernel $\mathcal{C}_k \lb b,\nu, \sigma_s\rb$ defined in \eqref{cn} satisfies the following pair of difference equations:
\begin{subequations}\label{eigenvaluetildecn} \begin{align}
 & \tilde{H}_{\mathcal{C}_k}(b,\sigma_s) ~ \mathcal{C}_k \lb b,\boldsymbol{\theta},\nu, \sigma_s\rb = e^{(-1)^{k+1}2\pi b \nu}\mathcal{C}_k \lb b,\boldsymbol{\theta},\nu, \sigma_s\rb, \\
 & \lb b^{-2\Delta(\sigma_s)}  \tilde{H}_{\mathcal{C}_k}(b^{-1},\sigma_s) b^{2\Delta(\sigma_s)} \rb \mathcal{C}_k \lb b,\boldsymbol{\theta},\nu, \sigma_s\rb = e^{(-1)^{k+1}2\pi b^{-1} \nu}\mathcal{C}_k \lb b,\boldsymbol{\theta},\nu, \sigma_s\rb,
\end{align}\end{subequations}
where $\tilde{H}_{\mathcal{C}_k}$ is given in \eqref{Hcntilde}.
\end{theorem}

\begin{remark}
As in the case of Theorem \ref{thm6p2}, it is possible to give a direct proof of Theorem \ref{thm6p4} which avoids the use of confluent limits.
\end{remark}

\begin{remark}
The two difference operators \eqref{Hcn} and \eqref{Hcntilde} possess different analytic properties: the coefficients in \eqref{Hcn} are holomorphic, while the coefficients in \eqref{Hcntilde} are meromorphic. It is therefore nontrivial that the confluent fusion kernels are eigenfunctions of both of them.
\end{remark}

\section{From the Virasoro fusion kernel to the Askey--Wilson polynomials}\label{FtoAW}

\subsection{A renormalized version of $F$} It was shown in \cite[Theorem 1]{R20} that a renormalized version of the Virasoro fusion kernel is equal to Ruijsenaars' hypergeometric function. In what follows, we rewrite the result of \cite{R20} in a form suitable for our present needs. Introduce the normalization factor $N$ by
\beq
N=K~\frac{g_b\left(-2 \sigma_t-\frac{i Q}{2}\right) g_b\left(2 \sigma_t-\frac{i Q}{2}\right)}{g_b\left(-2 \sigma_s+\frac{i Q}{2}\right) g_b\left(2 \sigma_s+\frac{i Q}{2}\right)} \prod_{\epsilon_1=\pm 1} \prod_{\epsilon_2=\pm 1} \frac{g_b\left(-\theta_t+\epsilon_1\theta_0+\epsilon_2 \sigma_s\right) g_b\left(-\theta_1+\epsilon_1 \theta_\infty +\epsilon_2 \sigma_s\right)}{g_b\left(\theta_0+\epsilon_1\theta_\infty+\epsilon_2 \sigma_t \right) g_b\left(\theta_t+\epsilon_1\theta_1 +\epsilon_2 \sigma_t\right)},
\eeq
where
\beq
K=s_b\lb \tfrac{iQ}2+2\theta_t \rb s_b\lb \tfrac{iQ}2+\theta_0+\theta_1+\theta_\infty+\theta_t \rb s_b\lb \tfrac{iQ}2+\theta_0+\theta_1-\theta_\infty+\theta_t \rb.
\eeq
It was shown in \cite[Theorem 1]{R20} that the renormalized Virasoro fusion kernel $F_{\ren}$ defined by
\beq \label{Fren}
F_{\ren}\left[\substack{\theta_1\;\;\theta_t\vspace{0.1cm}\\ \theta_{\infty}\;\theta_0};\substack{\sigma_s \vspace{0.15cm} \\  \sigma_t}\right] = N F\left[\substack{\theta_1\;\;\theta_t\vspace{0.1cm}\\ \theta_{\infty}\;\theta_0};\substack{\sigma_s \vspace{0.15cm} \\  \sigma_t}\right]
 \eeq
is equal to Ruijsenaars' hypergeometric function under a certain parameter correspondence. One advantage of renormalizing $F$ is that $F_{\ren}$ is symmetric under the exchange $(\sigma_s, \theta_0) \leftrightarrow (\sigma_t, \theta_1)$. Therefore, the four difference equations satisfied by $F_{\ren}$ can be written in a more symmetric form as follows. Define the difference operator $H_\ren$ by
\beq\label{3p18}
H_\ren \left[\substack{\theta_1\;\;\theta_t\vspace{0.1cm}\\ \theta_{\infty}\;\theta_0};\substack{b,\;\sigma_s}\right] = C\left[\substack{\theta_1\;\;\theta_t\vspace{0.1cm}\\ \theta_{\infty}\;\theta_0};\substack{b,\;\sigma_s}\right] e^{-ib \partial_{\sigma_s}} + C\left[\substack{\theta_1\;\;\theta_t\vspace{0.1cm}\\ \theta_{\infty}\;\theta_0};\substack{b,\;-\sigma_s}\right] e^{ib \partial_{\sigma_s}} + H_F^0\left[\substack{\theta_1\;\;\theta_t\vspace{0.1cm}\\ \theta_{\infty}\;\theta_0};\substack{b,\;\sigma_s}\right],
\eeq
where $H_F^0$ is defined by \eqref{H0} and
\beq
C\left[\substack{\theta_1\;\;\theta_t\vspace{0.1cm}\\ \theta_{\infty}\;\theta_0};\substack{b,\;\sigma_s}\right] = \frac{4 \prod _{\epsilon = \pm 1} \cosh \left(\pi  b \left(-\frac{ib}{2} -\theta_t+\sigma_s+\epsilon\theta_0\right)\right) \cosh \left(\pi  b \left(-\frac{ib}{2} -\theta_1+\sigma_s+\epsilon\theta_\infty\right)\right)}{\sinh (2 \pi  b \sigma_s) \sinh (\pi  b (-2 \sigma_s+i b))}.
\eeq
It follows from (\ref{difference1}), (\ref{difference2}), and (\ref{Fren}) that the renormalized Virasoro fusion kernel $F_\ren$ satisfies the following four difference equations \cite{R20}:
\begin{subequations}\label{hreneq}\begin{align}
\label{hreneq1} & H_\ren\left[\substack{\theta_1\;\;\theta_t\vspace{0.1cm}\\ \theta_{\infty}\;\theta_0};\substack{b,\;\sigma_s}\right] F_\ren\left[\substack{\theta_1\;\;\theta_t\vspace{0.1cm}\\ \theta_{\infty}\;\theta_0};\substack{\sigma_s \vspace{0.15cm} \\ \sigma_t}\right] = 2\cosh{(2\pi b \sigma_t)} F_\ren\left[\substack{\theta_1\;\;\theta_t\vspace{0.1cm}\\ \theta_{\infty}\;\theta_0};\substack{\sigma_s \vspace{0.15cm} \\  \sigma_t}\right], \\
\label{hreneq2} & H_\ren\left[\substack{\theta_1\;\;\theta_t\vspace{0.1cm}\\ \theta_{\infty}\;\theta_0};\substack{b^{-1},\;\sigma_s}\right] F_\ren\left[\substack{\theta_1\;\;\theta_t\vspace{0.1cm}\\ \theta_{\infty}\;\theta_0};\substack{\sigma_s \vspace{0.15cm} \\  \sigma_t}\right] = 2\cosh{(2\pi b^{-1} \sigma_t)} F_\ren\left[\substack{\theta_1\;\;\theta_t\vspace{0.1cm}\\ \theta_{\infty}\;\theta_0};\substack{\sigma_s \vspace{0.15cm} \\  \sigma_t}\right], \\
\label{hreneq3} & H_\ren\left[\substack{\theta_0\;\;\theta_t\vspace{0.1cm}\\ \theta_{\infty}\;\theta_1};\substack{b,\;\sigma_t}\right] F_\ren\left[\substack{\theta_1\;\;\theta_t\vspace{0.1cm}\\ \theta_{\infty}\;\theta_0};\substack{\sigma_s \vspace{0.15cm} \\  \sigma_t}\right] = 2\cosh{(2\pi b \sigma_s)} F_\ren\left[\substack{\theta_1\;\;\theta_t\vspace{0.1cm}\\ \theta_{\infty}\;\theta_0};\substack{\sigma_s \vspace{0.15cm} \\  \sigma_t}\right], \\
\label{hreneq4} & H_\ren\left[\substack{\theta_0\;\;\theta_t\vspace{0.1cm}\\ \theta_{\infty}\;\theta_1};\substack{b^{-1},\;\sigma_t}\right] F_\ren\left[\substack{\theta_1\;\;\theta_t\vspace{0.1cm}\\ \theta_{\infty}\;\theta_0};\substack{\sigma_s \vspace{0.15cm} \\  \sigma_t}\right] = 2\cosh{(2\pi b^{-1} \sigma_s)} F_\ren\left[\substack{\theta_1\;\;\theta_t\vspace{0.1cm}\\ \theta_{\infty}\;\theta_0};\substack{\sigma_s \vspace{0.15cm} \\  \sigma_t}\right].
\end{align}\end{subequations}


\subsection{From $F_\text{ren}$ to $A_n$}
Let $A_n$ be the Askey--Wilson polynomials defined in (\ref{AW}).
In this subsection, we show that $F_\text{ren}$ reduces to the polynomials $A_n$ in a certain limit. As a consequence, the Virasoro fusion kernel can be viewed as a non-polynomial generalization of the Askey--Wilson polynomials with quantum deformation parameter $q$ related to the central charge $c$ of the Virasoro algebra according to (\ref{qcQdef}).
In addition to Assumption \ref{assumption}, we need the following assumption.
\begin{assumption}[Restriction on the parameters]\label{assumptionAW}
Assume that $b > 0$ is such that $b^2$ is irrational, and that, for $\epsilon,\epsilon' = \pm 1,$
\beq\label{restrictionAW}\begin{split}
& \sigma_s, \sigma_t, \theta_1, \theta_0 \neq 0, \qquad \theta_\infty - \theta_t + \epsilon \sigma_s + \epsilon' \sigma_t \neq 0, \qquad \theta_\infty+\theta_t+\epsilon\theta_0+\epsilon'\theta_1 \neq 0.
\end{split}\eeq
\end{assumption}
Assumption \ref{assumptionAW} implies that the four increasing and the four decreasing sequences of poles of the integrand in \eqref{fusion01} are vertical and do not overlap. The assumption that $b^2$ is irrational implies that all the poles of the integrand are simple.
It is necessary to assume that $b^2$ is irrational because otherwise $q=e^{2i\pi b^2}$ is a root of unity and then the Askey--Wilson polynomials are not well-defined in general, see Remark \ref{rootofunityremark}.
 
As described in the introduction, the next theorem follows by combining one of the results in \cite{R1999} with the observation of \cite{R20} that $F_\ren = R$.

\begin{theorem}[Virasoro fusion kernel $\to$ Askey--Wilson polynomials]\label{FAWthm}
Suppose that Assumptions \ref{assumption} and \ref{assumptionAW} are satisfied. Define $\sigma_s^{(n)} \in \mathbb{C}$ for $n \geq 0$, by
\beq\label{limsigmas}
  \sigma_s^{(n)} = \tfrac{iQ}2+\theta_0+\theta_t+ibn.
\eeq
Under the parameter correspondence
\beq\label{paramAW}
\alpha =-e^{2 \pi  b \left(\frac{i b}{2}+\theta_1+\theta_t\right)}, \quad \beta =-e^{2 \pi  b \left(\frac{i b}{2}+\theta_0-\theta_\infty\right)}, \quad \gamma=-e^{2 \pi  b \left(\frac{i b}{2}-\theta_1+\theta_t\right)}, \quad \delta=-e^{2 \pi  b \left(\frac{i b}{2}+\theta_0+\theta_\infty\right)}, \quad q=e^{2i\pi b^2},
\eeq
the renormalized fusion kernel defined in \eqref{Fren} satisfies, for each integer $n \geq 0$,
\beq\label{limitAW}
\lim\limits_{\sigma_s \to \sigma_s^{(n)}} F_{\text{\upshape ren}}\left[\substack{\theta_1\;\;\;\theta_t\vspace{0.1cm}\\ \theta_{\infty}\;\;\theta_0};\substack{\sigma_s \vspace{0.15cm} \\  \sigma_t}\right] 
= A_n(e^{2\pi b \sigma_t};\alpha,\beta,\gamma,\delta,q),
\eeq
where $A_n$ are the Askey--Wilson polynomials defined in (\ref{AW}).
\end{theorem}
\begin{proof}
It is easier to give a direct proof than to explain how the assertion follows from \cite{R1999} and \cite{R20}.
In fact, there are two different ways to prove \eqref{limitAW}. The first approach consists of taking the limit $\sigma_s \to \sigma_s^{(n)}$ in the integral representation \eqref{fusion01} for $F$ for each $n$; the second approach only computes this limit for $n = 0$ and then instead uses the limit of one of the four difference equations \eqref{hreneq} to extend the result to other values of $n$. We choose to use the second approach. 

We first prove (\ref{limitAW}) for $n = 0$. The definition of $A_n$ as a hypergeometric series involves a $q$-Pochhammer symbol of the form $(q^{-n};q)_k$ in the numerator, see (\ref{AW}) and (\ref{hypergeometricphidef}). If $n = 0$, only the first term of this $q$-hypergeometric series is nonzero, because $(1;q)_k = 0$ for each $k\geq 1$. Since $(x;q)_0=1$ for all $x$ and $q$ by definition, we conclude that $A_0 = 1$. Thus, to prove (\ref{limitAW}) for $n = 0$, we need to show that the left-hand side of (\ref{limitAW}) equals $1$ when $n = 0$.
By (\ref{fusion01}) and (\ref{Fren}), we have
\beq \label{Fren2}
\begin{split}
F_\ren\left[\substack{\theta_1\;\;\;\theta_t\vspace{0.1cm}\\ \theta_{\infty}\;\;\theta_0};\substack{\sigma_s \vspace{0.15cm} \\  \sigma_t}\right] =  P_1(\sigma_s) \int_{\mathsf{F}} dx~I_1(x,\sigma_s),
\end{split}
\eeq
where
\begin{align*}
& P_1(\sigma_s) = K \prod_{\epsilon_1=\pm 1} \lb \frac{s_b(\epsilon_1\sigma_t-\theta_0-\theta_\infty)}{s_b(\epsilon_1\sigma_s+\theta_1-\theta_\infty)}\prod_{\epsilon_2=\pm 1} s_b(\epsilon_1\sigma_s+\epsilon_2\theta_0-\theta_t) \rb,
	\\
& I_1(x,\sigma_s)=\prod_{\epsilon=\pm1} \frac{s_b \lb x+ \epsilon \theta_1 \rb s_b \lb x+\epsilon\theta_0+\theta_\infty+\theta_t \rb}{s_b \lb x+\frac{iQ}{2}+\theta_\infty+\epsilon \sigma_s \rb s_b \lb x+\frac{i Q}{2}+\theta_t+\epsilon \sigma_t \rb}.
\end{align*}
The function $P_1(\sigma_s)$ has a simple zero at $\sigma_s^{(0)}= \tfrac{iQ}2+\theta_0+\theta_t$ originating from the factor $s_b\left(\sigma_s-\theta_0-\theta_t\right)$. Let us consider the integrand $I_1$. In the limit $\sigma_s \to \sigma_s^{(0)}$, the pole of $s_b(x+\tfrac{iQ}2+\theta_\infty+\sigma_s)^{-1}$ located at $x=-\theta_\infty-\sigma_s$ moves downwards, crosses the contour of integration $\mathsf{F}$, and collides with the pole of $s_b(x+\theta_0+\theta_\infty+\theta_t)$ located at $x = x_f :=-\tfrac{iQ}2-\theta_0-\theta_\infty-\theta_t$. Hence, before taking the limit $\sigma_s \to \sigma_s^{(0)}$, we choose to deform the contour of integration $\mathsf{F}$ into a contour $\mathsf{F}'$ which passes below $x_f$; this gives
\beq\label{2p24}
F_\ren\left[\substack{\theta_1\;\;\;\theta_t\vspace{0.1cm}\\ \theta_{\infty}\;\;\theta_0};\substack{\sigma_s \vspace{0.15cm} \\  \sigma_t}\right] = -2i\pi P_1(\sigma_s) \underset{x=x_f}{\text{Res}}\lb I_1(x,\sigma_s)\rb + P_1(\sigma_s) \int_{\mathsf{F}'} dx~I_1(x,\sigma_s).
\eeq
Using the relation
\beq \label{ressb}
\underset{z=-\frac{iQ}{2}}{\text{Res}} s_b(z) =\frac{i}{2\pi},
\eeq 
a straightforward computation yields
\begin{align}\label{xfres}
-2i\pi \underset{x=x_f}{\text{Res}} I_1(x,\sigma_s) = \frac{s_b\left(-2 \theta_0-\frac{i Q}{2}\right) s_b\left(-\theta_0-\theta_1-\theta_\infty-\theta_t-\frac{i Q}{2}\right) s_b\left(-\theta_0+\theta_1-\theta_\infty-\theta_t-\frac{i Q}{2}\right)}{s_b(-\theta_0-\theta_\infty-\sigma_t) s_b(-\theta_0-\theta_\infty+\sigma_t) s_b(-\theta_0-\theta_t-\sigma_s) s_b(-\theta_0-\theta_t+\sigma_s)}.
\end{align}
The right-hand side of (\ref{xfres}) has a simple pole at $\sigma_s=\sigma_s^{(0)}$ due to the factor $s_b(-\theta_0-\theta_t+\sigma_s)^{-1}$. Moreover, in the limit $\sigma_s \to \sigma_s^{(0)}$, the second term in \eqref{2p24} vanishes thanks to the zero of $P_1(\sigma_s)$. Thus,
\beq
\lim\limits_{\sigma_s\to \sigma_s^{(0)}} F_\ren\left[\substack{\theta_1\;\;\;\theta_t\vspace{0.1cm}\\ \theta_{\infty}\;\;\theta_0};\substack{\sigma_s \vspace{0.15cm} \\  \sigma_t}\right] = -2i\pi \lim\limits_{\sigma_s\to \sigma_s^{(0)}}  P_1(\sigma_s) \underset{x=x_f}{\text{Res}} I_1(x,\sigma_s).
\eeq
A straightforward computation shows that the right-hand side equals $1$; this proves (\ref{limitAW}) for $n = 0$.

For each integer $n \geq 0$, let $P_n$ denote the left-hand side of (\ref{limitAW}), i.e., 
\beq\label{2p20}
P_n = \lim\limits_{\sigma_s \to \sigma_s^{(n)}} F_{\ren}\left[\substack{\theta_1\;\;\;\theta_t\vspace{0.1cm}\\ \theta_{\infty}\;\;\theta_0};\substack{\sigma_s \vspace{0.15cm} \\  \sigma_t}\right].
\eeq
The same kind of contour deformation argument used to establish the case $n = 0$ shows that the limit in (\ref{2p20}) exists for all $n$. 
The function $P_n$ depends on $\sigma_t$ as well as the four parameters $\theta_0, \theta_1, \theta_t, \theta_\infty$. To show that $P_n$ equals the Askey--Wilson polynomial $A_n$ on the right-hand side of (\ref{limitAW}) for $n \geq 1$, we consider the limit of the difference equation (\ref{hreneq1}).
Using the parameter correspondence \eqref{paramAW}, it is straightforward to verify that
\beq\label{eqMn}
\lim\limits_{\sigma_s \to \sigma_s^{(n)}} H_\ren\left[\substack{\theta_1\;\;\;\theta_t\vspace{0.1cm}\\ \theta_{\infty}\;\;\theta_0};\substack{b,\;\sigma_s}\right] = R_{A_n}, \qquad n \geq 0,
\eeq
where $R_{A_n}$ is the recurrence operator defined in \eqref{Mn}. 
Hence, taking the limit of the first difference equation \eqref{hreneq1} for $F_\ren$ as $\sigma_s \to \sigma_s^{(n)}$, we see that $P_n$ satisfies 
\begin{align}\label{MnPn}
R_{A_n} P_n = (z+z^{-1}) P_n, \qquad n \geq 0,
\end{align}
where $z=e^{2\pi b \sigma_t}$.
Thus the $P_n$ satisfy the same recurrence relation \eqref{recurrenceAW} as the Askey--Wilson polynomials evaluated at $z=e^{2\pi b \sigma_t}$. 
Since we have already shown that $P_0=A_0 = 1$, equation (\ref{MnPn}) with $n = 0$ implies that $P_1 = A_1$ (note that there is no term with $P_{-1}$ in (\ref{MnPn}) for $n = 0$ because the coefficient $a^{-}_n$ defined in (\ref{bndef}) vanishes for $n = 0$). Assuming that $P_n = A_n$ for all $n \leq N$, equation (\ref{MnPn}) with $n = N$ shows that $P_{N+1} = A_{N+1}$; thus $P_n = A_n$ for all $n \geq 0$ by induction, where $A_n$ is evaluated at $z=e^{2\pi b \sigma_t}$. This completes the proof of (\ref{limitAW}).
\end{proof}

\begin{remark}\label{sigmanmremark}
The result of Theorem \ref{FAWthm} can be generalized as follows. Instead of considering the limit of $F_\ren$ as $\sigma_s$ approaches one of the points $\sigma_s^{(n)}$ defined in (\ref{limsigmas}), we can consider the limit 
\beq\label{sigmasnm}
\sigma_s \to \sigma_s^{(n,m)} := \sigma_s^{(n)} + \frac{im}b,
\eeq
for any integers $n,m \geq 0$. In this limit, the Virasoro fusion kernel reduces to a product of two Askey--Wilson polynomials of the form $A_n \times A_m$. The first polynomial $A_n$ is expressed in terms of the quantum deformation parameter $q=e^{2i\pi b^2}$, while the second polynomial $A_m$ is expressed in terms of $\tilde{q}=e^{2i\pi b^{-2}}$. In the case $m = 0$ treated in Theorem \ref{FAWthm}, the second polynomial reduces to $A_0 = 1$.  
\end{remark}

\begin{remark}[Limits of the other three difference equations]
We saw in the proof of Theorem \ref{FAWthm} that the first difference equation \eqref{hreneq1} for $F_\ren$ reduces to the recurrence relation \eqref{recurrenceAW} for the Askey--Wilson polynomials as $\sigma_s \to \sigma_s^{(n)}$. A similar argument using that, under the parameter correspondence \eqref{paramAW},
\beq\label{2p28}
-e^{2\pi b(\tfrac{ib}2+\theta_0+\theta_t)} H_\ren\left[\substack{\theta_0\;\;\;\theta_t\vspace{0.1cm}\\ \theta_{\infty}\;\;\theta_1};\substack{b,\;\sigma_t}\right] = \Delta_{A_n},
\eeq
where $\Delta_{A_n}$ is the operator defined in \eqref{L}, shows that the third difference equation \eqref{hreneq3} reduces to the difference equation (\ref{differenceAW}) for the Askey--Wilson polynomials as $\sigma_s \to \sigma_s^{(n)}$. 
On the other hand, the fourth difference equation \eqref{hreneq4} is trivially satisfied in the limit $\sigma_s \to \sigma_s^{(n)}$. Indeed, let $P_n$ be the limit of  $F_\ren$ as in (\ref{2p20}). Recalling that $A_n$ is a polynomial of order $n$ in $z + z^{-1}$, we deduce from (\ref{limitAW}) that $P_n$ is an $n$th order polynomial in $\cosh(2\pi b \sigma_t)$. In particular, $e^{\pm ib^{-1} \partial_{\sigma_t}} P_n = P_n$ so that the operator on the left-hand side of (\ref{hreneq4}) becomes a multiplication operator in the limit $\sigma_s \to \sigma_s^{(n)}$. The identity
$$C\left[\substack{\theta_0\;\;\;\theta_t\vspace{0.1cm}\\ \theta_{\infty}\;\;\theta_1};\substack{b^{-1},\;\sigma_t}\right] 
+ C\left[\substack{\theta_0\;\;\;\theta_t\vspace{0.1cm}\\ \theta_{\infty}\;\;\theta_1};\substack{b^{-1},\;-\sigma_t}\right] 
+ H_F^0\left[\substack{\theta_0\;\;\;\theta_t\vspace{0.1cm}\\ \theta_{\infty}\;\;\theta_1};\substack{b^{-1},\;\sigma_t}\right]
= 2\cosh{(2\pi b^{-1} \sigma_s^{(n)})}
$$
then shows that (\ref{hreneq4}) is trivially satisfied in the limit $\sigma_s \to \sigma_s^{(n)}$.
Finally, the limit of the second difference equation \eqref{hreneq2} is of a different nature: since it involves the shifts $\sigma_s \to \sigma_s \pm ib^{-1}$ induced by the operators $e^{\pm ib^{-1} \partial_{\sigma_s}}$, a proper description of its limit involves the more general family of functions $A_n \times A_m$ mentioned in Remark \ref{sigmanmremark}.
\end{remark}

\section{From $\mathcal{C}_k$ to the continuous dual $q$-Hahn polynomials}\label{CktoHahnsec}
In this section, we show that (up to normalization) the confluent fusion kernel $\mathcal{C}_k\lb b,\boldsymbol{\theta},\nu, \sigma_s\rb$ degenerates, for each $k \geq 1$, to the continuous dual $q$-Hahn polynomials $H_n$ when $\nu$ is suitably discretized. 

\subsection{A renormalized version of $\mathcal{C}_k$}
Define the renormalized version $\mathcal{C}_k^\ren$ of the confluent fusion kernel $\mathcal{C}_k$ by
\beq\label{ckren}
\mathcal{C}_k^\ren\lb b,\boldsymbol{\theta},\nu, \sigma_s\rb = N_1(\nu,\sigma_s) N_2(\boldsymbol{\theta})~\mathcal{C}_k \lb b,\boldsymbol{\theta},\nu, \sigma_s\rb,
\eeq
where
\beq\label{N1}\begin{split}
N_1(\nu,\sigma_s) = &\; e^{i \pi  \nu  (-1)^k (-\theta_0+\theta_t-\nu -i Q)} \lb b~e^{2 i \pi  \left(\left\lfloor \frac{k}{2}\right\rfloor -\frac{1}{2}\right)} \rb^{-\Delta (\theta_0)-\Delta (\theta_t)+\Delta (\sigma_s)-\frac{\theta_*^2}{2}+2 \nu ^2} \\
& \times \prod_{\epsilon=\pm 1} \frac{g_b\left(2\epsilon\sigma_s-\frac{i Q}{2}\right) g_b\left(\epsilon \left(\frac{\theta_*}{2}-\nu \right)-\theta_0\right) g_b\left(\theta_t+\epsilon \left(\frac{\theta_*}{2}+\nu \right)\right)}{g_b\left(\theta_*+\epsilon\sigma_s \right) g_b\left(-\theta_t-\sigma_s+\epsilon\theta_0 \right) g_b\left(-\theta_t+\sigma_s+\epsilon\theta_0 \right)}
\end{split}\eeq
and
\beq
N_2(\boldsymbol{\theta}) = e^{(-1)^{k+1} i \pi \lb\tfrac{\theta_*}2-\theta_0-\tfrac{iQ}2\rb \lb \theta_t-\frac{\theta_*}2-\tfrac{iQ}2\rb} \tfrac{s_b\lb \tfrac{iQ}2-2\theta_t\rb}{s_b\lb-\tfrac{iQ}2-\theta_0-\theta_*+\theta_t\rb}.
\eeq
Using the representation \eqref{gnm} and the identity $s_b(x)=g_b(x)/g_b(-x)$, we find that $\mathcal{C}^\ren_k$ is given by the following expression:
\beq\label{renormcn}
\mathcal{C}^\ren_k \lb b,\boldsymbol{\theta},\nu, \sigma_s\rb = \mathcal{P}_k\lb\boldsymbol{\theta},\nu, \sigma_s\rb \displaystyle \int_{\mathsf{C}} dx ~ I^{(k)}\lb x,\boldsymbol{\theta},\nu, \sigma_s\rb,
\eeq
where $I^{(k)}$ is given in \eqref{I} and
\beq\label{prefcn}
\mathcal{P}_k\lb\boldsymbol{\theta},\nu, \sigma_s\rb = N_2(\boldsymbol{\theta})~e^{i \pi  \nu  (-1)^k (-\theta_0+\theta_t-\nu -i Q)} \frac{s_b(\sigma_s-\theta_*) s_b(-\sigma_s-\theta_*)}{s_b\left(\theta_0-\frac{\theta_*}{2}+\nu \right) s_b\left(\tfrac{\theta_*}{2}-\theta_t+\nu \right) s_b\left(-\tfrac{\theta_*}{2}-\theta_t-\nu \right)}.
\eeq
It is easy to see that $\mathcal{C}^\ren_k = \mathcal{C}^\ren_{k+2}$ for each integer $k \geq 1$ and that
\beq
\mathcal{C}^\ren_k \lb b,\boldsymbol{\theta},\nu, \sigma_s\rb = \mathcal{C}^\ren_k \lb b^{-1},\boldsymbol{\theta},\nu, \sigma_s\rb, \qquad k \geq 1.
\eeq 

\subsection{From $\mathcal{C}^\ren_k$ to $H_n$}\label{CrenkHnsubsec}
Define $\{\nu_n\}_{n=0}^\infty \subset \mathbb{C}$ by
\beq \label{nun}
\nu_n = \theta_t-\tfrac{iQ}2-\tfrac{\theta_*}2-inb.
\eeq
The main result of this section (Theorem \ref{thhahn}) states that the continuous dual $q$-Hahn polynomials $H_n$ defined in \eqref{qhahn} emerge from $\mathcal{C}^\ren_k$ when $\nu$ is discretized according to (\ref{nun}).
We will need the following two lemmas for the proof.

\begin{lemma}\label{sbdiff} 
For any integer $m \geq 0$,  the following identities hold: 
\begin{subequations}\label{identitysb}\begin{align}
\label{sb1} & \frac{s_b(x+imb)}{s_b(x)} = e^{\frac{m^2i\pi b^2}{2}} e^{\pi b mx} \lb -e^{-i\pi b^2} e^{-2\pi b x};e^{-2i\pi b^2}\rb_m, \\
\label{sb2} & \frac{s_b(x+\frac{im}{b})}{s_b(x)} = e^{\frac{m^2i\pi b^{-2}}{2}} e^{\frac{\pi  mx}b} \lb -e^{-i\pi b^{-2}} e^{-\frac{2\pi x}b};e^{-2i\pi b^{-2}}\rb_m,
\end{align}\end{subequations}
where $(a;q)_m$ denotes the $q$-Pochhammer symbol defined in (\ref{qpochhammerdef}).
\end{lemma}
\begin{proof}
The identity \eqref{sb1} follows by applying the difference equation for $s_b$ in (\ref{gbsbdifferenceeqs}) recursively. The identity \eqref{sb2} is obtained by sending $b\to b^{-1}$ in \eqref{sb1} and using the symmetry $s_{b^{-1}}(x)=s_b(x)$.
\end{proof}

\begin{lemma}\label{Sigmanklemma}
Let $\Sigma_{n,k}$ denote the sum
\begin{align}\label{Sigmankdef}
\Sigma_{n,k} = \sum_{m=0}^n \alpha^{-m} \beta ^{-m} q^{-m n} \left(\alpha^{m} \beta ^{m} q^{m(n-1)}\right)^{\delta_{k,1}} \frac{\big(\frac{q^{1-m}}{q^{-n}}, \frac{q^{1-m}}{\gamma  z},  \frac{q^{1-m}}{\gamma z^{-1}} ;q\big){}_m}{\big(\frac{q^{1-m}}{q},\frac{q^{1-m}}{\beta  \gamma },\frac{q^{1-m}}{\alpha  \gamma };q\big){}_m},
\end{align}
where $\delta_{k,1}=1$ if $k=1$ and $\delta_{k,1}=0$ if $k\neq 1$. Then, for any integer $n \geq 1$, 
\begin{align}\label{Sigmankcomputed}
\Sigma_{n,k} = \begin{cases} 
{}_3\phi_2\left( \left. \begin{matrix} q^{-n},\gamma z ,\gamma z^{-1} \\ \beta \gamma, \alpha\gamma \end{matrix} \right|q; \alpha \beta q^n \right), & k = 1,
	\\
{}_3\phi_2\left( \left. \begin{matrix} q^{-n}, \gamma z, \gamma z^{-1} \\ \beta \gamma, \alpha \gamma \end{matrix}  \right| q ; q \right), & k = 2.
\end{cases}
\end{align}
\end{lemma}
\begin{proof}
Let $n \geq 1$ be an integer.
Using the general identity (see \cite[Eq. 17.2.9]{NIST})
\begin{align}\label{pochhammeridentity1}
 \left(\frac{q^{1-m}}{a};q\right)_{m} = \frac{(a;q)_{m}}{(-a)^{m} q^{\frac{m(m-1)}{2}}}
\end{align}
with, in turn, $a = q^{-n}$, $a = \gamma z$, and $a = \gamma z^{-1}$, we find
\begin{align*}
& \Big(\frac{q^{1-m}}{q^{-n}}, \frac{q^{1-m}}{\gamma z}, \frac{q^{1-m}}{\gamma z^{-1}};q\Big)_m
= \frac{\left(q^{-n};q\right)_m\left(\gamma z;q\right)_m\left(\gamma z^{-1};q\right)_m}{(-q^{-n})^m(-\gamma z)^m(-\gamma z^{-1})^m q^{\frac{3m(m-1)}{2}}}.
\end{align*}
Similarly, applying (\ref{pochhammeridentity1}) with $a = q$, $a = \beta \gamma$, and $a = \alpha \gamma$, we find
$$\frac{1}{\left(\frac{q^{1-m}}{q},\frac{q^{1-m}}{\beta \gamma },\frac{q^{1-m}}{\alpha \gamma };q\right){}_m}
= \frac{(-q)^m (-\beta \gamma)^m (-\alpha \gamma)^m q^{\frac{3m(m-1)}{2}}}{(q;q)_m(\beta \gamma;q)_m(\alpha \gamma;q)_m}.$$
Thus
\begin{align}\label{bigquotient}
\frac{\big(\frac{q^{1-m}}{q^{-n}}, \frac{q^{1-m}}{\gamma  z},  \frac{q^{1-m}}{\gamma z^{-1}} ;q\big){}_m}{\left(\frac{q^{1-m}}{q},\frac{q^{1-m}}{\beta\gamma },\frac{q^{1-m}}{\alpha\gamma };q\right){}_m}
= \alpha^m \beta^m q^{m(n+1)} \frac{\left(q^{-n}, \gamma z, \gamma z^{-1};q\right)_m}{(q, \beta \gamma, \alpha \gamma;q)_m}.
\end{align}
It follows that
\begin{align*}
\Sigma_{n,1} = \sum_{m=0}^n 
 \frac{\left(q^{-n}, \gamma z, \gamma z^{-1};q\right)_m}{(q, \beta \gamma, \alpha \gamma;q)_m} \alpha^{m} \beta^m q^{mn}, \qquad
\Sigma_{n,2}
 =  \sum_{m=0}^n 
 \frac{\left(q^{-n}, \gamma z, \gamma z^{-1};q\right)_m}{(q, \beta \gamma, \alpha \gamma;q)_m} q^m.	
\end{align*}
Since one of the entries of the Pochhammer symbols in the numerators is $q^{-n}$, we can replace the upper limit of summation with infinity without changing the values of the sums. Hence (\ref{Sigmankcomputed}) follows from the definition (\ref{hypergeometricphidef}) of {}$_3\phi_2$.
\end{proof}
In addition to Assumption \ref{assumption}, we make the following assumption.

\begin{assumption}[Restriction on the parameters]\label{assumptionhahnjacobi}
Assume that $b > 0$ is such that $b^2$ is irrational, and that
\beq\label{restrictionhahnjacobi}\begin{split}
& \sigma_s, \theta_0 \neq 0, \qquad \tfrac{\theta_*}{2}-\nu+\theta_t\pm \sigma_s \neq 0, \qquad \theta_t - \theta_* \pm \theta_0 \neq 0.
\end{split}\eeq
\end{assumption}
Assumption \ref{assumptionhahnjacobi} implies that the three increasing and the three decreasing sequences of poles of the integrand in \eqref{renormcn} are vertical and do not overlap. It also implies that all the poles of the integrand are simple. The assumption that $b^2$ is irrational ensures that $q=e^{2i\pi b^2}$ is not a root of unity; this is needed in order for the continuous dual $q$-Hahn polynomials to be well defined, see Remark \ref{rootofunityremark}.

We now present the main result of this section.

\begin{theorem}[confluent Virasoro fusion kernels $\to$ continuous dual $q$-Hahn polynomials] \label{thhahn}
Suppose that Assumptions \ref{assumption} and \ref{assumptionhahnjacobi} are satisfied. 
Let $n \geq 0$ be an integer.
For each integer $k \geq 1$, the renormalized confluent fusion kernel $\mathcal{C}^\ren_k$ defined in \eqref{renormcn} reduces to the continuous dual $q$-Hahn polynomial $H_n$ defined in \eqref{qhahn} in the limit $\nu \to \nu_n$ as follows:
\beq\label{degenqhahn}
\lim\limits_{\nu \to \nu_n} \mathcal{C}^{\ren}_{k}(b,\boldsymbol{\theta},\nu,\sigma_s) = \begin{cases}
H_n(e^{2\pi b \sigma_s};\alpha^{-1},\beta^{-1},\gamma^{-1},q^{-1}), & \text{$k$ odd},
	\\
H_n(e^{2\pi b \sigma_s};\alpha,\beta,\gamma,q), & \text{$k$ even},
\end{cases}
\eeq
where $\nu_n \in \mathbb{C}$ is defined in (\ref{nun}) and
\beq\label{paramhahn}
\alpha=-e^{-2\pi b(\theta_0-\theta_t+\tfrac{ib}2)}, \quad \beta=-e^{2\pi b (\theta_0+\theta_t-\tfrac{ib}2)}, \quad \gamma=-e^{-2\pi b (\theta_*+\tfrac{ib}2)}, \qquad q=e^{-2i\pi b^2}.
\eeq
\end{theorem}
\begin{proof}
Since $\mathcal{C}^\ren_k = \mathcal{C}^\ren_{k+2}$, it is enough to prove the result for $k = 1$ and $k = 2$. Thus let $k \in \{1, 2\}$. Let $m,l \geq 0$ be integers and define $x_{m,l}\in \mathbb{C}$ by
\beq\label{xml}
x_{m,l} = -\tfrac{iQ}2-\tfrac{\theta_*}2+\theta_t-imb - \frac{il}{b}-\nu.
\eeq 
The integrand $I^{(k)}$ defined in \eqref{I} contains the factor
\beq\label{sboversb}
\frac{s_b(x+\tfrac{\theta_*}2-\theta_t+\nu)}{s_b(x+\tfrac{iQ}2)}.
\eeq
The function $s_b(x+\tfrac{\theta_*}2-\theta_t+\nu)$ has a simple pole located at $x=x_{m,l}$ for any integers $m,l \geq 0$. 
In the limit $\nu \to \nu_n$, the pole $x_{n,0}$ moves upwards, crosses the contour $\mathsf{C}$, and collides with the pole of $s_b(x+\tfrac{iQ}2)$ located at $x=0$. Therefore, before taking the limit $\nu \to \nu_n$, we deform $\mathsf{C}$ into a contour $\mathsf{C}'$ which passes just below $x_{n,0}$. As $\mathsf{C}$ is deformed into $\mathsf{C}'$, the integral in \eqref{renormcn} picks up residue contributions from all the poles $x=x_{m,l}$ which satisfy $\im{x_{m,l}}\geq \im{x_{n,0}}$, i.e., from all the poles $x_{m,l}$ such that $(m,l)$ satisfies $mb+\tfrac{l}{b} \leq nb$. We find
\beq\label{5p18}
\int_{\mathsf{C}} dx ~ I^{(k)}\lb x,\boldsymbol{\theta},\nu, \sigma_s\rb =  -2i\pi \sum_{\substack{m,l \geq 0 \\ mb+\tfrac{l}{b} \leq nb}}\underset{x=x_{m,l}}{\text{Res}}\lb I^{(k)}\lb x,\boldsymbol{\theta},\nu, \sigma_s\rb\rb + \int_{\mathsf{C}'} dx~I^{(k)}\lb x,\boldsymbol{\theta},\nu, \sigma_s\rb.
\eeq
Using Lemma \ref{sbdiff} and the residue of the function $s_b$ in \eqref{ressb}, a straightforward computation shows that residue of $I^{(k)}$ at the simple pole $x = x_{m,l}$ is given by
\begin{align}\nonumber
-2i\pi \underset{x=x_{m,l}}{\text{Res}} &I^{(k)}\lb x,\boldsymbol{\theta},\nu, \sigma_s\rb
=  e^{\frac{i\pi}{2} \left(b^2 m (m+1)+\frac{l (l+1)}{b^2}+2 ml+m+l\right)}e^{\frac{i \pi (-1)^k}{2} \left(\frac{\theta_*}{2}+\theta_t+\nu +\frac{i Q}{2}\right) \left(2 i b m+\frac{i (2 l+1)}{b}+i b+\theta_*-2 \theta_t+2 \nu \right)}  
	\\ \nonumber
& \times \frac{s_b\left(-i b m-\frac{i l}{b}-\theta_0-\theta_*+\theta_t-\frac{i Q}{2}\right) s_b\left(-i b m-\frac{i l}{b}+\theta_0-\theta_*+\theta_t-\frac{i Q}{2}\right)}{s_b\left(-i b m-\frac{i l}{b}-\theta_*+\sigma_s\right) s_b\left(-i b m-\frac{i l}{b}-\theta_*-\sigma_s\right) s_b\left(-i b m-\frac{i l}{b}-\frac{\theta_*}{2}+\theta_t-\nu \right)} 
	\\ \label{5p20}
& \times \frac{1}{\left(e^{\frac{2 i \pi  l}{b^2}};e^{-\frac{2 i \pi }{b^2}}\right){}_l \left(e^{2i\pi m b^2};e^{-2 i\pi b^2}\right){}_m}.
\end{align}
Because of the factor $s_b(-i b m-\frac{il}b-\frac{\theta_*}{2}+\theta_t-\nu)^{-1}$ in \eqref{5p20}, we deduce from the properties (\ref{polesb}) of $s_b$ that the function $\text{Res}_{x=x_{m,l}} I^{(k)}\lb x,\boldsymbol{\theta},\nu, \sigma_s\rb$ has a simple pole at $\nu=\nu_n$ if the pair $(m,l)$ satisfies $m\in [0,n]$ and $l=0$, but is regular at $\nu=\nu_n$ for all other choices of $m \geq 0$ and $l \geq 0$. 
On the other hand, because of the factor $s_b(\tfrac{\theta_*}2-\theta_t+\nu)^{-1}$ appearing in \eqref{prefcn}, $\mathcal{P}_k(\boldsymbol{\theta},\nu,\sigma_s)$ has a simple zero at $\nu = \nu_n$.
Hence the product $P_k(\boldsymbol{\theta},\nu,\sigma_s) \text{Res}_{x=x_{m,l}} I^{(k)}( x,b,\boldsymbol{\theta},\nu, \sigma_s)$ is nonzero in the limit $\nu \to \nu_n$ only if $m\in [0,n]$ and $l=0$. We deduce that 
\beq\label{CkrenResIk}
\lim\limits_{\nu \to \nu_n} \mathcal{C}^\ren_k \lb b,\boldsymbol{\theta},\nu, \sigma_s\rb = \mathcal{C}^\ren_k \lb b,\boldsymbol{\theta},\nu_n, \sigma_s\rb = -2i\pi \lim\limits_{\nu \to \nu_n} \mathcal{P}_k\lb\boldsymbol{\theta},\nu, \sigma_s\rb \sum_{m=0}^n \underset{x=x_{m,0}}{\text{Res}} I^{(k)} \lb x, b, \boldsymbol{\theta},\nu, \sigma_s\rb.
\eeq
More explicitly, for $k=1,2$, we find that 
\begin{align}\nonumber
& \mathcal{C}^\ren_k ( b,\boldsymbol{\theta},\nu_n, \sigma_s) =e^{\pi b n (\theta_*+\theta_t-\theta_0)} e^{-2\pi b n\delta_{k,1}( \theta_*+\theta_t-\theta_0 )}\frac{s_b(\theta_0-\theta_*+\theta_t-\frac{i Q}{2})}{s_b(-i b n+\theta_0-\theta_*+\theta_t-\frac{i Q}{2})} \frac{s_b(\frac{i Q}{2}-2 \theta_t)}{s_b(i b n+\frac{i Q}{2}-2 \theta_t)} 
	\\ \nonumber
& \times \displaystyle \sum_{m=0}^n \frac{e^{-2\pi b m\delta_{k,1}( i b n-2 \theta_t )} e^{i \pi b m (\frac{Q}2+2 i \theta_t+b(n+\frac{m}2))}}{(e^{2 i \pi  b^2 m};e^{-2 i \pi  b^2}){}_m}  \frac{s_b(-i b m-\theta_0-\theta_*+\theta_t-\tfrac{i Q}{2})}{s_b(-\theta_0-\theta_*+\theta_t-\frac{i Q}{2})} \frac{s_b(-i b m+\theta_0-\theta_*+\theta_t-\frac{i Q}{2})}{s_b(\theta_0-\theta_*+\theta_t-\frac{i Q}{2})} 
	\\
& \times \frac{s_b(-\theta_*-\sigma_s)}{s_b(-i b m-\theta_*-\sigma_s)} \frac{s_b(\sigma_s-\theta_*)}{s_b(-i b m-\theta_*+\sigma_s)} \frac{s_b(i b n+\frac{i Q}{2})}{s_b(-i b m+i b n+\frac{i Q}{2})}.
\end{align}
Using \eqref{identitysb}, a long but straightforward computation gives
\beq\label{4p28}\begin{split}
\mathcal{C}^\ren_k &( b,\boldsymbol{\theta},\nu_n, \sigma_s) 
=  e^{i\pi n^2} e^{\pi  b n (4 \theta_t-i (n+1) Q)} e^{-2\pi  b n \delta_{k,1} (-\theta_0+\theta_*+\theta_t)} \frac{ \big(e^{2 \pi  b (i b n-\theta_0+\theta_*-\theta_t)};e^{-2 i \pi  b^2}\big){}_n}{\big(e^{2 \pi  b (2 \theta_t-i b)};e^{-2 i \pi  b^2}\big){}_n} 
	\\
& \times \sum_{m=0}^n e^{2 \pi  b m (i b (n+1)-2 \theta_t)} e^{-2\pi bm\delta_{k,1}(i b n-2 \theta_t)} 
	\\
&\times \frac{\big(e^{-2 i \pi  b^2} e^{2 i \pi  b^2 (m-n)},-e^{\pi  b (2 (\theta_*-\sigma_s)+i b (2 m-1))}, -e^{\pi  b (2 (\theta_*+\sigma_s)+i b (2 m-1))} ;e^{-2 i \pi  b^2}\big){}_m}{\big(e^{2 i \pi  b^2 m},e^{2 \pi  b (i b m-\theta_0+\theta_*-\theta_t)},  e^{2 \pi  b (i b m+\theta_0+\theta_*-\theta_t)} ;e^{-2 i \pi  b^2}\big){}_m},
\end{split}\eeq
Recalling the parameter correspondence \eqref{paramhahn} and letting $z=e^{2\pi b \sigma_s}$, equation \eqref{4p28} can be rewritten as
\beq\label{4p29}\begin{split}
\mathcal{C}^\ren_k ( b,\boldsymbol{\theta},\nu_n, \sigma_s) = &\; e^{i \pi  n^2} \alpha ^n \beta ^n q^{\frac{n (n-1)}{2}} \frac{\big(\frac{q^{1-n}}{\gamma  \beta };q\big){}_n}{(\alpha  \beta ;q)_n} \big(\alpha ^{n/2} \gamma ^{-n/2}\big)^{-2\delta_{k,1}} \Sigma_{n,k},
\end{split}\eeq
where $\Sigma_{n,k}$ denotes the sum in (\ref{Sigmankdef}).
Utilizing (\ref{pochhammeridentity1}) with $a = \gamma \beta$, we can write (\ref{4p29}) as
\beq\begin{split}
\mathcal{C}^\ren_k ( b,\boldsymbol{\theta},\nu_n, \sigma_s) = &\;  \frac{\alpha^n\left(\gamma \beta;q\right)_{n}}{\gamma^n (\alpha  \beta ;q)_n} \big(\alpha ^{n/2} \gamma ^{-n/2}\big)^{-2\delta_{k,1}} \Sigma_{n,k},
\end{split}\eeq
Hence the proof of equation \eqref{degenqhahn} reduces to proving the following identity:
\begin{align}\label{4p24}
\frac{\alpha^n\left(\gamma \beta;q\right)_{n}}{\gamma^n (\alpha  \beta ;q)_n}  \big(\alpha ^{n/2} \gamma ^{-n/2}\big)^{-2\delta_{k,1}} \Sigma_{n,k} 
= \begin{cases}
H_n(z;\alpha^{-1},\beta^{-1},\gamma^{-1},q^{-1}), & k=1,
	\\
H_n(z;\alpha,\beta,\gamma,q), & k=2.
\end{cases}
\end{align}

Let us first prove (\ref{4p24}) for $k = 1$. Recalling the definition (\ref{qhahn}) of $H_n$ and applying the identity (see \cite[Eq. 17.2.7]{NIST})
\begin{align}\label{pochhammeridentity2}
(a; q^{-1})_n= (a^{-1}; q)_n (-a)^n q^{-\frac{n(n-1)}{2}}
\end{align}
multiple times, we find
\begin{align*}
H_n(z;\alpha^{-1},\beta^{-1},\gamma^{-1},q^{-1})
& = {}_3\phi_2\left( \left. \begin{matrix} q^{n},\alpha^{-1} z ,\alpha^{-1} z^{-1} \\ (\alpha\beta)^{-1} , (\alpha\gamma)^{-1}  \end{matrix} \right|q^{-1} ;q^{-1} \right)
 = \sum_{m=0}^\infty \frac{(q^{n}, \alpha^{-1} z, \alpha^{-1} z^{-1};q^{-1})_m}{( (\alpha \beta)^{-1}, (\alpha \gamma)^{-1},q^{-1}; q^{-1})_m}q^{-m}
	\\
& = \sum_{m=0}^\infty \frac{(q^{-n}, \alpha z, \alpha z^{-1};q)_m}{(\alpha \beta, \alpha \gamma,q; q)_m} \frac{q^{mn}}
{\beta^{-m} \gamma^{-m}}
= {}_3\phi_2\left( \left. \begin{matrix} q^{-n},\alpha z^{-1},\alpha z \\ \alpha\beta, \alpha\gamma \end{matrix} \right|q; \beta \gamma q^n\right).
\end{align*}
In view of Lemma \ref{Sigmanklemma}, it follows that (\ref{4p24}) can be rewritten as follows when $k = 1$:
\begin{align}\label{identitykequals12}
\frac{(\gamma\beta; q)_n}{(\alpha \beta; q)_n} 
 {}_3\phi_2\left( \left. \begin{matrix} q^{-n},\gamma z ,\gamma z^{-1} \\ \beta \gamma, \alpha\gamma \end{matrix} \right|q; \alpha \beta q^n \right)
 = {}_3\phi_2\left( \left. \begin{matrix} q^{-n},\alpha z^{-1},\alpha z \\ \alpha\beta, \alpha\gamma \end{matrix} \right|q; \beta \gamma q^n\right).
\end{align}
Applying the identity (see \cite[Eq. 17.9.10]{NIST})
$$ {}_3\phi_2\left( \left. \begin{matrix} q^{-n}, b, c \\ d, e \end{matrix} \right|q; \frac{deq^n}{bc} \right)= \frac{(e/c;q)_n}{(e;q)_n}
 {}_3\phi_2\left( \left. \begin{matrix} q^{-n},c, d/b \\ d, c q^{1-n}/e \end{matrix} \right|q; q \right),$$
to the left-hand side of (\ref{identitykequals12}) with
 $$b = \gamma z,\quad c = \gamma z^{-1}, \quad d = \alpha \gamma, \quad e = \beta \gamma,$$ 
and to the right-hand side of (\ref{identitykequals12}) with
 $$b = \alpha z, \quad c = \alpha z^{-1}, \quad d = \alpha\gamma, \quad e = \alpha\beta,$$ 
we deduce that the two sides of (\ref{identitykequals12}) are indeed equal. This proves (\ref{4p24}) for $k = 1$.

For $k = 2$, Lemma \ref{Sigmanklemma} and the definition (\ref{qhahn}) of $H_n$ imply that (\ref{4p24}) can be written as
\begin{align}\label{identity3}
& \frac{\alpha^n (\gamma\beta; q)_n}{\gamma^n (\alpha \beta; q)_n} 
 {}_3\phi_2\left( \left. \begin{matrix} q^{-n}, \gamma z, \gamma z^{-1} \\ \beta \gamma, \alpha \gamma \end{matrix}  \right| q ; q \right)
 = {}_3\phi_2\left( \left. \begin{matrix} q^{-n},\alpha z ,\alpha z^{-1} \\ \alpha\beta , \alpha\gamma  \end{matrix} \right|q ; q \right).
\end{align}
This equation is obtained by setting $b = \alpha z$, $c = \alpha z^{-1}$, $d = \alpha \gamma$, and $e = \alpha \beta$ in the general identity (see \cite[Eq. 17.9.8]{NIST})
\begin{align}\label{NIST1798}
\frac{(de/(bc); q)_n}{(e;q)_n} \bigg(\frac{bc}{d}\bigg)^n 
{}_3\phi_2\left( \left. \begin{matrix} q^{-n}, d/b, d/c \\ d, de/(bc) \end{matrix} \right|q ; q \right)
= {}_3\phi_2\left( \left. \begin{matrix} q^{-n}, b, c \\ d, e \end{matrix} \right|q ; q \right).
\end{align}
This proves (\ref{4p24}) also for $k = 2$ and completes the proof of the theorem.
\end{proof}

\subsection{Limits of the difference equations for $\mathcal{C}_k$ as $\nu \to \nu_n$}
In this subsection, we explain how the recurrence relation and the difference equation for the continuous dual $q$-Hahn polynomials emerge from the difference equations for $\mathcal{C}_k$ in the limit $\nu \to \nu_n$. 

We first formulate the difference equations in terms of $\mathcal{C}^\ren_k$.
Using the relation \eqref{ckren} between $\mathcal{C}^\ren_k$ and $\mathcal{C}_k$, the four difference equations for $\mathcal{C}_k$ derived in Theorem \ref{thm6p2} and Theorem \ref{thm6p4} can be expressed as difference equations for  $\mathcal{C}^\ren_k$. Define the renormalized difference operators
\begin{subequations} \label{5p27} \begin{align} 
\label{5p27a} & H_{\mathcal{C}_k^\ren}(b,\nu) = N_1(\nu,\sigma_s)  H_{\mathcal{C}_k}(b,\nu)  N_1(\nu,\sigma_s)^{-1}, \\
\label{5p27b} & \tilde{H}_{\mathcal{C}_k^\ren}(b,\sigma_s) = N_1(\nu,\sigma_s)  \tilde{H}_{\mathcal{C}_k}(b,\sigma_s)  N_1(\nu,\sigma_s)^{-1},
\end{align} \end{subequations}
where the difference operators $H_{\mathcal{C}_k}$ and $\tilde{H}_{\mathcal{C}_k}$ are respectively defined in \eqref{Hcn} and \eqref{Hcntilde}, and $N_1$ is given in \eqref{N1}. Using the identity \eqref{gbsbdifferenceeqs} satisfied by the function $g_b$, it can be verified that $N_1$ satisfies the difference equations
\beq\label{N1nu}
\frac{N_1(\nu,\sigma_s)}{N_1(\nu+ib,\sigma_s)} = U(\nu,\theta_*), \qquad \frac{N_1(\nu,\sigma_s)}{N_1(\nu-ib,\sigma_s)}=e^{(-1)^k 2\pi b(ib+\theta_0-\theta_t)} U(-\nu,-\theta_*),
\eeq
where
\beq
U(\nu,\theta_*)=-e^{2 \pi  b (2 \nu +i b) \left(2 \left\lfloor \frac{k}{2}\right\rfloor -1\right)}e^{\pi  b (-1)^k (-2 i b-\theta_0+\theta_t-2 \nu )}\tfrac{\Gamma \left(\frac{b Q}{2}+ib \left(\theta_0+\frac{\theta_*}{2}-\nu \right)\right)}{\Gamma \left(\frac{1-b^2}{2}+ib \left(\theta_0+\nu -\frac{\theta_*}{2}\right)\right)} \tfrac{\Gamma \left(\frac{b Q}{2}-i b \left(\frac{\theta_*}{2}+\theta_t+\nu \right)\right)}{\Gamma \left(\frac{1-b^2}{2}+ib \left(\frac{\theta_*}{2}-\theta_t+\nu \right)\right)}.
\eeq
Using the identities \eqref{N1nu}, we find that $H_{\mathcal{C}_k^\ren}(b,\nu)$ is given explicitly by 
\begin{align} \label{Hnuren}
H_{\mathcal{C}_k^\ren}(b,\nu) = H^+_{\mathcal{C}_k^\ren}(\nu) e^{ib\partial_\nu}+ H^-_{\mathcal{C}_k^\ren}(\nu) e^{-ib\partial_\nu} + H_{\mathcal{C}_k}^{0}(\nu),
\end{align}
where $H_{\mathcal{C}_k}^{0}(\nu)$ is given in \eqref{Hn0} and 
\beq
H^\pm_{\mathcal{C}_k^\ren}(\nu) = -4 e^{\mp \pi  b (-1)^{k} (\theta_0-\theta_t\mp 2\nu )} \cosh \left(\pi  b \left(\tfrac{i b}{2}+\theta_0\mp \tfrac{\theta_*}{2}\pm\nu \right)\right) \cosh \left(\pi  b \left(\tfrac{i b}{2}-\theta_t\pm\tfrac{\theta_*}{2}\pm\nu \right)\right).
\eeq

The difference operator $\tilde{H}_{\mathcal{C}_k^\ren}(b,\sigma_s)$ is computed in a similar way. The normalization factor $N_1$ also satisfies the identities 
\beq\label{5p31}
\frac{N_1(\nu,\sigma_s)}{N_1(\nu,\sigma_s+ib)} = V(\sigma_s), \qquad \frac{N_1(\nu,\sigma_s)}{N_1(\nu,\sigma_s-ib)} = V(-\sigma_s),
\eeq
where
\beq\begin{split}
V(\sigma_s)= e^{\pi  b (2 \sigma_s+i b) \left(2 \left\lfloor \frac{k}{2}\right\rfloor -1\right)} & \tfrac{\Gamma \left(b^2-2 i b \sigma_s\right) \Gamma (-2 i b \sigma_s) \Gamma \left(\frac{1-b^2}{2}+ib (\sigma_s-\theta_*)\right)}{\Gamma \left(2 i b \sigma_s-2 b^2\right) \Gamma \left(2 i b \sigma_s-b^2\right) \Gamma \left(\frac{b Q}{2}-i b (\theta_*+\sigma_s)\right)} \prod _{\epsilon=\pm} \tfrac{\Gamma \left(\frac{1-b^2}{2}+ib \left(\epsilon \theta_0+\theta_t+\sigma_s\right)\right)}{\Gamma \left(\frac{b Q}{2}+ib \left(\epsilon \theta_0+\theta_t-\sigma_s\right)\right)}.
\end{split}\eeq
Using \eqref{5p31}, it is straightforward to verify that $\tilde{H}_{\mathcal{C}_k^\ren}(b,\sigma_s)$ takes the form
\beq\label{4p34}\begin{split}
\tilde{H}_{\mathcal{C}_k^\ren}(b,\sigma_s) & = -V_k(-\sigma_s, -\theta_t) e^{ib\partial_{\sigma_s}} -V_k(\sigma_s, -\theta_t) e^{-ib\partial_{\sigma_s}} + \tilde{H}_{\mathcal{C}_k}^{0}(\sigma_s),
 \end{split}\eeq
 where $\tilde{H}_{\mathcal{C}_k}^{0}(\sigma_s)$ is given in \eqref{Htilden0} and $V_k$ is defined in (\ref{Vkdef}). 
 
 The next lemma summarizes the difference equations satisfied by $\mathcal{C}_k^\text{ren}$.

\begin{lemma}[Difference equations for $\mathcal{C}_k^\ren$]\label{lemma4p1}
Let  $H_{\mathcal{C}_k^\ren}$ and $\tilde{H}_{\mathcal{C}_k^\ren}$ be the difference operators defined in \eqref{Hnuren} and \eqref{4p34}, respectively. For each integer $k \geq 1$, the renormalized confluent fusion kernel $\mathcal{C}_k^\ren$ satisfies the pair of difference equations
\begin{subequations} \label{differencecnren}\begin{align}
\label{37a} & H_{\mathcal{C}_k^\ren}(b,\nu) ~ \mathcal{C}^\ren_k \lb b,\boldsymbol{\theta},\nu, \sigma_s\rb = 2\cosh{(2\pi b \sigma_s)} ~ \mathcal{C}^\ren_k \lb b,\boldsymbol{\theta}, \nu, \sigma_s\rb, \\
\label{37b} & \tilde{H}_{\mathcal{C}_k^\ren}(b,\sigma_s) ~ \mathcal{C}^\ren_k \lb b,\boldsymbol{\theta},\nu, \sigma_s\rb = e^{(-1)^{k+1}2\pi b \nu}\mathcal{C}^\ren_k \lb b,\boldsymbol{\theta},\nu, \sigma_s\rb,
\end{align} \end{subequations}
as well as the pair of difference equations obtained by replacing $b\to b^{-1}$ in \eqref{differencecnren}. 
\end{lemma} 
\begin{proof}
The lemma follows immediately from \eqref{5p27} together with Theorems \ref{thm6p2} and \ref{thm6p4}.
\end{proof}
%

The next proposition describes how the recurrence relation and the difference equation for the continuous dual $q$-Hahn polynomials appear in the limit $\nu \to \nu_n$. 

\begin{proposition}
Let $k \geq 1$  and $n \geq 0$ be integers. In the limit $\nu \to \nu_n$, the difference equations \eqref{37a} and (\ref{37b}) for $\mathcal{C}^\ren_k$ reduce to the three-term recurrence relation (\ref{recurrenceHn}) and the difference equation \eqref{differencehahn} for the continuous dual $q$-Hahn polynomials, respectively.
More precisely,
\beq \label{recurrence1}
R_{H_n}\lb \alpha^{(-1)^k},\beta^{(-1)^k},\gamma^{(-1)^k};q^{(-1)^k}\rb~\mathcal{C}^\ren_k \lb b,\boldsymbol{\theta},\nu_n, \sigma_s\rb = 2 \cosh{(2\pi b \sigma_s)}~\mathcal{C}^\ren_k \lb b,\boldsymbol{\theta},\nu_n, \sigma_s\rb,
\eeq
and 
\beq\label{difference1hn}\begin{split}
\Delta_{H_n}\lb\alpha^{(-1)^k},\beta^{(-1)^k},\gamma^{(-1)^k};q^{(-1)^k},e^{(-1)^k2\pi b \sigma_s}\rb~\mathcal{C}^\ren_k \lb b,\boldsymbol{\theta},\nu_n, \sigma_s\rb = \lb q^{(-1)^{k+1} n}-1\rb~\mathcal{C}^\ren_k \lb b,\boldsymbol{\theta},\nu_n, \sigma_s\rb,
\end{split}\eeq
where the parameters are related according to \eqref{paramhahn}, $R_{H_n}$ is the recurrence operator defined in \eqref{sn}, and $\Delta_{H_n}$ is the difference operator defined in \eqref{deltahn}.
\end{proposition} 
\begin{proof}
We first prove that the difference equation \eqref{37a} reduces to the three-term recurrence relation \eqref{recurrence1} in the limit $\nu \to \nu_n$. It follows from the definition \eqref{nun} of $\nu_n$ that
\beq
\lim\limits_{\nu \to \nu_n} e^{\pm ib \partial_\nu}\mathcal{C}^\ren_k \lb b,\boldsymbol{\theta},\nu, \sigma_s\rb = T_{n\mp 1}\mathcal{C}^\ren_k \lb b,\boldsymbol{\theta},\nu_n, \sigma_s\rb,
\eeq
where $T_{n\pm 1}$ acts on $\mathcal{C}_k^\ren$ as $T_{n\pm 1}\mathcal{C}^\ren_k \lb b,\boldsymbol{\theta},\nu_n, \sigma_s\rb = \mathcal{C}^\ren_k \lb b,\boldsymbol{\theta},\nu_{n\pm 1}, \sigma_s\rb$.
 Moreover, it can be verified that
\begin{align}\label{6p42}\nonumber
& \lim\limits_{\nu \to \nu_n} H^-_{\mathcal{C}_k^\ren}(\nu) e^{-ib \partial_\nu} \mathcal{C}^\ren_k \lb b,\boldsymbol{\theta},\nu, \sigma_s\rb= b^{+}_n\lb \alpha^{(-1)^k},\beta^{(-1)^k},\gamma^{(-1)^k};q^{(-1)^k}\rb T_{n+1}\mathcal{C}^\ren_k \lb b,\boldsymbol{\theta},\nu_n, \sigma_s\rb, \\
& \lim\limits_{\nu \to \nu_n} H^+_{\mathcal{C}_k^\ren}(\nu) e^{ib\partial_\nu} \mathcal{C}^\ren_k \lb b,\boldsymbol{\theta},\nu, \sigma_s\rb = b^{-}_n\lb \alpha^{(-1)^k},\beta^{(-1)^k},\gamma^{(-1)^k};q^{(-1)^k}\rb T_{n-1}\mathcal{C}^\ren_k \lb b,\boldsymbol{\theta},\nu_n, \sigma_s\rb,
\end{align}
and
\beq
\lim\limits_{\nu \to \nu_n} H_{\mathcal{C}_k}^{0}(\nu) = \alpha+\alpha^{-1}
- b^{+}_n\lb\alpha^{(-1)^k},\beta^{(-1)^k},\gamma^{(-1)^k};q^{(-1)^k}\rb
- b^{-}_n \lb\alpha^{(-1)^k},\beta^{(-1)^k},\gamma^{(-1)^k};q^{(-1)^k} \rb,
\eeq
where $b^{+}_n$ and $b^{-}_n$ are the coefficients defined in (\ref{cndef}).
Hence, 
\beq
H_{\mathcal{C}_k^\ren}(b,\nu_n)\mathcal{C}^\ren_k \lb b,\boldsymbol{\theta},\nu_n, \sigma_s\rb
 = R_{H_n} \lb \alpha^{(-1)^k},\beta^{(-1)^k},\gamma^{(-1)^k};q^{(-1)^k}\rb\mathcal{C}^\ren_k \lb b,\boldsymbol{\theta},\nu_n, \sigma_s\rb,
\eeq
so (\ref{recurrence1}) follows from (\ref{37a}).

We now prove that in the limit $\nu \to \nu_n$, the difference equation \eqref{37b} reduces to \eqref{difference1hn}. Observe that the function $\tilde{H}_{\mathcal{C}_k}^{0}(\sigma_s)$ in \eqref{Htilden0} is an even function of $\theta_t$, i.e., 
\beq\label{5p42} \begin{split}
 \tilde{H}_{\mathcal{C}_k}^{0}(\sigma_s) &= -e^{(-1)^k \pi b(ib+\theta_*+2\theta_t)}+V_k(\sigma_s,\theta_t)+V_k(-\sigma_s,\theta_t) \\
 & = -e^{(-1)^k \pi b(ib+\theta_*-2\theta_t)}+V_k(\sigma_s, -\theta_t)+V_k(-\sigma_s, -\theta_t).
\end{split} \eeq
Using \eqref{paramhahn}, \eqref{4p34} and \eqref{5p42}, we can rewrite the limit of \eqref{37b} as $\nu \to \nu_n$ as follows:
\begin{align}\nonumber
& V_k(\sigma_s, -\theta_t)~ \mathcal{C}^\ren_k \lb b,\boldsymbol{\theta},\nu_n, \sigma_s-ib\rb + V_k(-\sigma_s, -\theta_t)~\mathcal{C}^\ren_k \lb b,\boldsymbol{\theta},\nu_n, \sigma_s+ib\rb 
	\\ \nonumber
& - \lb V_k(\sigma_s, -\theta_t) + V_k(-\sigma_s, -\theta_t) -e^{(-1)^k 2\pi b\big(\tfrac{ib}2+\tfrac{\theta_*}2-\theta_t\big)} \rb ~ \mathcal{C}^\ren_k \lb b,\boldsymbol{\theta},\nu_n, \sigma_s\rb
	\\\label{5p41}
& = e^{(-1)^{k} 2\pi b\lb\tfrac{ib}2+\frac{\theta_*}2-\theta_t+ibn\rb} \mathcal{C}^\ren_k \lb b,\boldsymbol{\theta},\nu_n, \sigma_s\rb.
\end{align}
We notice that
\beq \label{4p50}
V_k(\sigma_s, -\theta_t)=e^{(-1)^k 2\pi b\lb\tfrac{ib}2+\tfrac{\theta_*}2-\theta_t\rb} h\lb\alpha^{(-1)^k},\beta^{(-1)^k},\gamma^{(-1)^k};q^{(-1)^k},e^{(-1)^k 2\pi b \sigma_s}\rb, \eeq
where $h(\alpha,\beta,\gamma;q,z)$ is defined in \eqref{hdef}. Moreover, writing $z=e^{2\pi b \sigma_s}$ and recalling that $q=e^{-2i\pi b^2}$, we observe that the translation $\sigma_s \to \sigma_s \pm ib$ corresponds to the multiplication $z \to q^{\mp 1} z$. Using these observations, a straightforward computation shows that \eqref{5p41} implies \eqref{difference1hn}. 
\end{proof}

\subsection{An alternative proof of Theorem \ref{thhahn}}\label{rk6p6}
The proof of Theorem \ref{thhahn} presented in Section \ref{CrenkHnsubsec} is based on a direct evaluation of the limit $\nu \to \nu_n$ in the integral representation \eqref{gnm} for $\mathcal{C}^\ren_k$.
This approach has the advantage that it is direct, but it employs several identities for $q$-functions and some rather involved algebra. 
An alternative approach is based on the recurrence relation (\ref{recurrence1}).  This approach, which is similar to the approach adopted in the proof of Theorem \ref{FAWthm}, is in fact the approach we  originally used to arrive at Theorem \ref{thhahn}. In this alternative approach, the proof of Theorem \ref{thhahn} involves two steps. First, the residue computation that led to (\ref{CkrenResIk}) is carried out in the special (and relatively simple) case of $n = 0$; this yields
\beq
\mathcal{C}^\ren_k \lb b,\boldsymbol{\theta},\nu_0, \sigma_s\rb = 1, \qquad k \geq 1,
\eeq
and shows that (\ref{degenqhahn}) holds for $n = 0$. 
Second, the functions 
$$\mathcal{C}^\ren_k\lb b,\boldsymbol{\theta},\nu_n, \sigma_s\rb \quad \text{and} \quad H_n(e^{2\pi b \sigma_s}; \alpha^{(-1)^k},\beta^{(-1)^k},\gamma^{(-1)^k}, q^{(-1)^k})$$ 
obey the same recurrence relation (cf. (\ref{recurrenceHn}) and (\ref{recurrence1})). Since (\ref{degenqhahn}) holds for $n = 0$, equation \eqref{recurrence1} with $n=0$ implies that 
\beq
\mathcal{C}^\ren_k\lb b,\boldsymbol{\theta},\nu_1, \sigma_s\rb = H_1(e^{2\pi b \sigma_s}; \alpha^{(-1)^k},\beta^{(-1)^k},\gamma^{(-1)^k}, q^{(-1)^k}),
\eeq
thus (\ref{degenqhahn}) holds for $n=1$. More generally, assuming that (\ref{degenqhahn}) holds for all $n \leq N$, equation \eqref{recurrence1} shows that (\ref{degenqhahn}) holds also for $n=N+1$. By induction, (\ref{degenqhahn}) holds for all $n \geq 0$. This completes the alternative proof of Theorem \ref{thhahn}.

\section{From $\mathcal{C}_k$ to the big $q$-Jacobi polynomials}\label{CktoJacobisec}
In this section, we show that (up to normalization) the confluent Virasoro fusion kernel $\mathcal{C}_k\lb b,\boldsymbol{\theta},\nu, \sigma_s\rb$ degenerates, for each $k \geq 1$, to the big $q$-Jacobi polynomials $J_n$ when $\sigma_s$ is suitably discretized. 

\subsection{Another renormalized version of $\mathcal{C}_k$}
Define the renormalized version $\hat{\mathcal{C}}_k^\ren$ of the confluent fusion kernel $\mathcal{C}_k$ by
\beq\label{renormcnhat}
\hat{\mathcal{C}}_k^\ren(b,\boldsymbol{\theta},\nu,\sigma_s) = N_3(\nu,\sigma_s) N_4(\boldsymbol{\theta})~\mathcal{C}_k(b,\boldsymbol{\theta},\nu,\sigma_s),
\eeq
where
\beq\label{N3}\begin{split}
N_3(\nu,\sigma_s) = &\; e^{i \pi  \nu  (-1)^k (\theta_0-\theta_t-\nu -i Q)} \lb b e^{2i \pi  \left(\left\lfloor \frac{k}{2}\right\rfloor -\frac12\right)}\rb^{-\Delta (\theta_0)-\Delta (\theta_t)+\Delta (\sigma_s)-\frac{\theta_*^2}{2}+2 \nu ^2}\\
& \times \prod _{\epsilon=\pm 1} \frac{g_b\left(2 \epsilon \sigma_s-\frac{i Q}{2}\right) g_b\left(\theta_0+\epsilon \left(\frac{\theta_*}{2}-\nu \right)\right) g_b\left(\epsilon \left(\frac{\theta_*}{2}+\nu \right)-\theta_t\right)}{g_b\left(\epsilon \sigma_s-\theta_*\right) g_b\left(-\theta_0+\epsilon\theta_t +\sigma_s\right) g_b\left(-\theta_0+\epsilon\theta_t -\sigma_s\right)}
\end{split}\eeq
and
\beq\label{N4}
N_4(\boldsymbol{\theta}) = s_b\left(-2 \theta_0+\tfrac{i Q}{2}\right) s_b\left(-\theta_0-\theta_*+\theta_t+\tfrac{i Q}{2}\right) e^{i \pi (-1)^k \left(\theta_0+\frac{\theta_*}{2}-\frac{i Q}{2}\right) \left(\frac{\theta_*}{2}+\theta_t+\frac{i Q}{2}\right)}.
\eeq
It follows from (\ref{gnm}) and (\ref{renormcnhat}) that the function $\hat{\mathcal{C}}_k^\ren$ admits the integral representation
\beq\label{hatCkrendef}
\hat{\mathcal{C}}_k^\ren(b,\boldsymbol{\theta},\nu,\sigma_s) =\hat{P}^{(k)}\lb\boldsymbol{\theta},\nu, \sigma_s\rb \displaystyle \int_{\mathsf{C}} dx ~ I^{(k)}\lb x,\boldsymbol{\theta},\nu, \sigma_s\rb,
\eeq
where $I^{(k)}$ is given in \eqref{I} and
\beq \label{5p5}
\hat{P}^{(k)}\lb\boldsymbol{\theta},\nu, \sigma_s\rb = N_4(\boldsymbol{\theta}) e^{i \pi  \nu  (-1)^k (\theta_0-\theta_t-\nu -i Q)} s_b\left(\theta_0+\tfrac{\theta_*}{2}-\nu \right) s_b(\theta_0-\theta_t-\sigma_s) s_b(\theta_0-\theta_t+\sigma_s).
\eeq
In particular, $\hat{\mathcal{C}}_k^\ren$ is invariant under each of the shifts $k\to k+2$ and $b \to b^{-1}$.

\subsection{From $\hat{\mathcal{C}}^\ren_k$ to $J_n$}\label{section7p2}
Define $\{\sigma_s^{(n)}\}_{n=0}^\infty \subset \mathbb{C}$ by
\beq\label{sigmasdiscrete}
\sigma_s^{(n)}=\tfrac{iQ}2 +\theta_t-\theta_0+inb.
\eeq
The main result of this section (Theorem \ref{thjacobi}) states that the big $q$-Jacobi polynomials $J_n$ defined in \eqref{Jn} emerge from $\hat{\mathcal{C}}_k^{\text{ren}}$ when $\sigma_s$ is discretized according to \eqref{sigmasdiscrete}. The proof will require Lemma \ref{sbdiff} and the following lemma.

\begin{lemma}\label{JacobiSigmanklemma}
Let $\Sigma_{n,k}$ denote the sum
\begin{align}\label{JacobiSigmankdef}
\Sigma_{n,k} = 
& \sum_{m=0}^n q^{-m} \big(x\beta \gamma^{-1}\big)^{m\delta_{k,2}} 
\frac{\big(q^{-m+n+1}, \frac{q^{-m-n}}{\alpha  \beta}, \frac{q^{1-m}}{x};q\big){}_m}{\big(\frac{q^{-m}}{\alpha }, \frac{q^{-m}}{\gamma }, q^{-m};q\big){}_m},
\end{align}
where $\delta_{k,2}=1$ if $k=2$ and $\delta_{k,2}=0$ if $k\neq 2$. Then, for any integer $n \geq 1$, 
\begin{align}\label{JacobiSigmankcomputed}
\Sigma_{n,k} = \begin{cases} 
J_n(x^{-1};\alpha^{-1},\beta^{-1},\gamma^{-1};q^{-1}), & k =1, 	\\
J_n(x;\alpha,\beta,\gamma;q), & k=2. 
\end{cases}
\end{align}
\end{lemma}
\begin{proof}
Let $n \geq 1$ be an integer.
Using the identity (\ref{pochhammeridentity1}) with, in turn, $a = q^{-n}$, $a = \alpha \beta  q^{n+1}$, and $a = x$, we find
\begin{align*}
& \Big(q^{-m+n+1}, \frac{q^{-m-n}}{\alpha  \beta}, \frac{q^{1-m}}{x};q\Big)_m
= \frac{\left(q^{-n};q\right)_m\left(\alpha \beta  q^{n+1};q\right)_m\left(x;q\right)_m}
{(-q^{-n})^m(-\alpha \beta  q^{n+1})^m(-x)^m q^{\frac{3m(m-1)}{2}}}.
\end{align*}
Similarly, applying (\ref{pochhammeridentity1}) with $a = \alpha q$, $a = \gamma q$, and $a = q$, we find
$$\frac{1}{\left(\frac{q^{-m}}{\alpha }, \frac{q^{-m}}{\gamma }, q^{-m};q\right){}_m}
= \frac{(-\alpha q)^m(-\gamma q)^m(- q)^m q^{\frac{3m(m-1)}{2}}}{\left(\alpha q;q\right)_{m}\left(\gamma q;q\right)_{m}\left(q;q\right)_{m}}.$$
Thus
\begin{align*}
\frac{\big(q^{-m+n+1}, \frac{q^{-m-n}}{\alpha  \beta}, \frac{q^{1-m}}{x};q\big){}_m}{\big(\frac{q^{-m}}{\alpha }, \frac{q^{-m}}{\gamma }, q^{-m};q\big){}_m}
= \frac{\gamma^m q^{2m}}{\beta^m x^m}
\frac{\left(q^{-n}, \alpha \beta q^{n+1}, x;q\right)_m}{\left(\alpha q, \gamma q, q;q\right)_{m}}.
\end{align*}
It follows that
\begin{align*}
\Sigma_{n,2}
 =  \sum_{m=0}^n 
\frac{\left(q^{-n}, \alpha \beta q^{n+1}, x;q\right)_m}{\left(\alpha q, \gamma q, q;q\right)_{m}} q^{m} = {}_3\phi_2\left( \left. \begin{matrix} q^{-n}, \alpha \beta q^{n+1}, x \\ \alpha q, \gamma q \end{matrix}  \right| q ; q \right) = J_n(x;\alpha,\beta,\gamma;q),	
\end{align*}
which proves (\ref{JacobiSigmankcomputed}) for $k = 2$.
Similarly, it follows that
\begin{align*}
\Sigma_{n,1} = \sum_{m=0}^n 
\frac{\left(q^{-n}, \alpha \beta q^{n+1}, x;q\right)_m}{\left(\alpha q, \gamma q, q;q\right)_{m}} \frac{\gamma^m q^{m}}{\beta^m x^m}.
\end{align*}
Applying the identity (\ref{pochhammeridentity2}) six times, this can be rewritten as
\begin{align*}
\Sigma_{n,1} &= \sum_{m=0}^\infty \frac{(q^{n}, (\alpha \beta q^{n+1})^{-1}, x^{-1};q^{-1})_m}{( (\alpha q)^{-1}, (\gamma q)^{-1},q^{-1}; q^{-1})_m}q^{-m}
= {}_3\phi_2\left( \left. \begin{matrix} q^{n}, (\alpha \beta q^{n+1})^{-1}, x^{-1} \\ (\alpha q)^{-1} , (\gamma q)^{-1}  \end{matrix} \right|q^{-1} ;q^{-1} \right)
	\\
& = J_n(x^{-1};\alpha^{-1},\beta^{-1},\gamma^{-1};q^{-1})
\end{align*}
which proves (\ref{JacobiSigmankcomputed}) also for $k = 1$.
\end{proof}

The following theorem is the main result of this section.

\begin{theorem}[confluent Virasoro fusion kernels $\to$ big $q$-Jacobi polynomials]\label{thjacobi}
Suppose that Assumptions \ref{assumption} and \ref{assumptionhahnjacobi} are satisfied. 
Let $n \geq 0$ be an integer.
For each integer $k \geq 1$, the renormalized confluent fusion kernel $\hat{\mathcal{C}}^\ren_k$ defined in \eqref{renormcnhat} reduces to the big $q$-Jacobi polynomial $J_n$ defined in \eqref{Jn} in the limit $\sigma_s \to \sigma_s^{(n)}$ as follows:
\beq\label{degenjacobi}
\lim\limits_{\sigma_s \to \sigma_s^{(n)}} \hat{\mathcal{C}}_k^\ren(b,\boldsymbol{\theta},\nu,\sigma_s) = \begin{cases} 
J_n(x^{-1};\alpha^{-1},\beta^{-1},\gamma^{-1};q^{-1}), & \text{$k$ odd}, 	\\
J_n(x;\alpha,\beta,\gamma;q), & \text{$k$ even}, 
\end{cases}
\eeq
where
\beq\label{paramjacobi}
\alpha = e^{4\pi b \theta_0}, \quad \beta = e^{-4\pi b \theta_t}, \quad \gamma = e^{2\pi b(\theta_0+\theta_*-\theta_t)}, \quad x=-e^{\pi b (-ib+2\theta_0+\theta_*)}e^{-2\pi b \nu}, \quad q=e^{-2i\pi b^2}.
\eeq
\end{theorem}
\begin{proof}
Since $\hat{\mathcal{C}}_k^\ren=\hat{\mathcal{C}}_{k+2}^\ren$, it is enough to prove the result for $k=1$ and $k=2$. Thus let $k\in \{1,2\}$. Let $m,l\geq 0$ be integers and define $x_{m,l} \in \mathbb{C}$ by
\beq
x_{m,l} = -\tfrac{iQ}2+\theta_0+\tfrac{\theta_*}2-\nu-imb-\tfrac{il}b. \eeq
The integrand $I^{(k)}$ defined in \eqref{I} contains the factor \beq
\frac{s_b\big( x-\theta_0+\nu-\tfrac{\theta_*}2 \big)}{s_b\big( x+\tfrac{iQ}2+\nu-\tfrac{\theta_*}2-\theta_t+\sigma_s \big)}.
\eeq
The function $s_b\big( x-\theta_0-\tfrac{\theta_*}2+\nu \big)$ has a simple pole located at $x_{m,l}$ for any integers $m,l \geq 0$. In the limit $\sigma_s \to \sigma_s^{(n)}$, the pole of $s_b\big( x+\tfrac{iQ}2-\tfrac{\theta_*}2-\theta_t+\nu+\sigma_s \big)^{-1}$ located at $x=\tfrac{\theta_*}2+\theta_t-\nu-\sigma_s$ moves downwards, crosses the contour $\mathsf{C}$, and collides with the pole of $s_b\big( x-\theta_0-\tfrac{\theta_*}2+\nu \big)$ located at $x_{n,0}$. Therefore, before taking the limit $\sigma_s \to \sigma_s^{(n)}$, we deform the contour $\mathsf{C}$ into a contour $\mathsf{C}'$ which passes just below $x_{n,0}$. As $\mathsf{C}$ is deformed into $\mathsf{C}'$, the integral in (\ref{hatCkrendef}) picks up residue contributions from all the poles $x_{m,l}$ which satisfy $\im x_{m,l} \geq \im x_{n,0}$, i.e., from all the poles $x_{m,l}$ such that $(m,l)$ satisfies $mb+\tfrac{l}{b} \leq nb$. The integral in (\ref{hatCkrendef}) becomes
\beq
\int_{\mathsf{C}} dx ~ I^{(k)}\lb x,\boldsymbol{\theta},\nu, \sigma_s\rb =  -2i\pi \sum_{\substack{m,l \geq 0 \\ mb+\tfrac{l}{b} \leq nb}} \underset{x=x_{m,l}}{\text{Res}}\lb I^{(k)}\lb x,\boldsymbol{\theta},\nu, \sigma_s\rb\rb + \int_{\mathsf{C}'} I^{(k)}\lb x,\boldsymbol{\theta},\nu, \sigma_s\rb.
\eeq
Utilizing Lemma \ref{sbdiff}, it can be verified that the residues of $I^{(k)}$ at the simple poles $x_{m,l}$ are given by
\beq\label{5p11}\begin{split}
 -2&i\pi \underset{x=x_{m,l}}{\text{Res}}\lb I^{(k)}\lb x,\boldsymbol{\theta},\nu, \sigma_s\rb\rb \\
 =&\; e^{i \pi  \left(\frac{l}{b}+b m\right) \left(\frac{l}{2 b}+\frac{b m}{2}+\frac{Q}{2}\right)} e^{i \pi  (-1)^{k+1} \left(\frac{\theta_*}{2}+\theta_t+\nu +\frac{i Q}{2}\right) \left(-\frac{i l}{b}-i b m+\theta_0+\frac{\theta_*}{2}-\nu -\frac{i Q}{2}\right)} \frac{1}{\left(e^{\frac{2 i \pi  l}{b^2}};e^{-\frac{2 i \pi }{b^2}}\right){}_l \left(e^{2 i \pi  m b^2};e^{-2 i b^2 \pi }\right){}_m} \\
&\times \frac{s_b\left(-\frac{i l}{b}-i b m+2 \theta_0-\frac{i Q}{2}\right) s_b\left(-\frac{i l}{b}-i b m+\theta_0+\theta_*-\theta_t-\frac{i Q}{2}\right)}{s_b\left(-\frac{i l}{b}-i b m+\theta_0+\frac{\theta_*}{2}-\nu \right) s_b\left(-\frac{i l}{b}-i b m+\theta_0-\theta_t-\sigma_s\right) s_b\left(-\frac{i l}{b}-i b m+\theta_0-\theta_t+\sigma_s\right)}. \end{split}
\eeq
Because of the factor $s_b(-\tfrac{il}b-ibm+\theta_0-\theta_t+\sigma_s)^{-1}$, we deduce from the properties \eqref{polesb} of $s_b$ that the function $\text{Res}_{x=x_{m,l}}\lb I^{(k)}\lb x,\boldsymbol{\theta},\nu, \sigma_s\rb\rb$ has a simple pole at $\sigma_s=\sigma_s^{(n)}$ if the pair $(m,l)$ satisfies $m\in [0,n]$ and $l=0$, but is regular at $\sigma_s=\sigma_s^{(n)}$ for all other choices of $m\geq 0$ and $l\geq 0$. On the other hand, because of the factor $s_b(\theta_0-\theta_t+\sigma_s)$ appearing in \eqref{5p5}, $\hat{P}^{(k)}\lb\boldsymbol{\theta},\nu, \sigma_s\rb$ has a simple zero at $\sigma_s=\sigma_s^{(n)}$. Hence the product $\hat{P}^{(k)}\lb\boldsymbol{\theta},\nu, \sigma_s\rb \underset{x=x_{m,l}}{\text{Res}}\lb I^{(k)}\lb x,\boldsymbol{\theta},\nu, \sigma_s\rb\rb$ is nonzero in the limit $\sigma_s\to \sigma_s^{(n)}$ only if $m\in [0,n]$ and $l=0$. We deduce that
\beq\label{7p15}
\lim\limits_{\sigma_s \to \sigma_s^{(n)}} \hat{\mathcal{C}}_k^\ren(b,\boldsymbol{\theta},\nu,\sigma_s) = \hat{\mathcal{C}}_k^\ren\lb b,\boldsymbol{\theta},\nu,\sigma_s^{(n)}\rb = -2i\pi \lim\limits_{\sigma_s \to \sigma_s^{(n)}} \hat{P}^{(k)}\lb\boldsymbol{\theta},\nu, \sigma_s\rb \sum_{m=0}^n \underset{x=x_{m,0}}{\text{Res}}\lb I^{(k)}\lb x,\boldsymbol{\theta},\nu, \sigma_s\rb\rb.
\eeq
Employing \eqref{5p5} and \eqref{5p11}, we obtain
\beq\label{7p16}\begin{split}
 \hat{\mathcal{C}}_k^\ren\lb b,\boldsymbol{\theta},\nu,\sigma_s^{(n)}\rb = & \sum_{m=0}^n e^{\pi  b m \left(\frac{i b m}{2}+(-1)^{k+1} \left(\frac{\theta_*}{2}+\theta_t+\nu \right)+iQ \delta_{k,1}\right)} \frac{1}{\left(e^{2 i b^2 m \pi };e^{-2 i b^2 \pi }\right){}_m} \frac{s_b\left(i b n+\frac{i Q}{2}\right)}{s_b\left(-i b m+i b n+\frac{i Q}{2}\right)}  
 	\\
& \times \frac{s_b\left(-i b m+2 \theta_0-\frac{i Q}{2}\right)}{s_b\left(2 \theta_0-\tfrac{i Q}{2}\right)} \frac{s_b\left(-i b n+2 \theta_0-2 \theta_t-\frac{i Q}{2}\right)}{s_b\left(-i b m-i b n+2 \theta_0-2 \theta_t-\frac{i Q}{2}\right)}
 	\\
& \times \frac{s_b\left(-i b m+\theta_0+\theta_*-\theta_t-\frac{i Q}{2}\right)}{s_b\left(\theta_0+\theta_*-\theta_t-\frac{i Q}{2}\right)} \frac{s_b\left(\theta_0+\frac{\theta_*}{2}-\nu \right)}{s_b\left(-i b m+\theta_0+\frac{\theta_*}{2}-\nu \right)}.
\end{split}\eeq
Lemma \ref{sbdiff} allows us to express \eqref{7p16} in terms of $q$-Pochammer symbols as follows:
\beq\label{5p11b} \begin{split}
& \hat{\mathcal{C}}_k^\ren\lb b,\boldsymbol{\theta},\nu,\sigma_s^{(n)}\rb = \sum_{m=0}^n e^{-2\pi bm\left( \delta_{k,2} \left(\frac{\theta_*}{2}+\theta_t+\nu +\frac{i Q}{2}\right)-i Q\right)} \\
& \times \frac{\left(e^{2 i b^2 (m-n-1) \pi };e^{-2 i b^2 \pi }\right){}_m \left(e^{2 b \pi  (b i (m+n)-2 \theta_0+2 \theta_t)};e^{-2 i b^2 \pi }\right){}_m \left(-e^{b \pi  (b i (2 m-1)-2 \theta_0-\theta_*+2 \nu )};e^{-2 i b^2 \pi }\right){}_m}{\left(e^{2 i b^2 m \pi };e^{-2 i b^2 \pi }\right){}_m \left(e^{2 i b \pi  (b m+2 i \theta_0)};e^{-2 i b^2 \pi }\right){}_m \left(e^{2 b \pi  (b i m-\theta_0-\theta_*+\theta_t)};e^{-2 i b^2 \pi }\right){}_m}.
\end{split} \eeq
Recalling the parameter correspondence \eqref{paramjacobi}, we arrive at
\beq \begin{split}
\hat{\mathcal{C}}_k^\ren\big( b,\boldsymbol{\theta},\nu,\sigma_s^{(n)}\big) = \Sigma_{n,k},
\end{split} \eeq
where $\Sigma_{n,k}$ is the sum defined in (\ref{JacobiSigmankdef}). The theorem now follows from Lemma \ref{JacobiSigmanklemma}.
\end{proof}

\subsection{Limits of the difference equations for $\mathcal{C}_k$ as $\sigma_s \to \sigma_s^{(n)}$}
In this subsection, we explain how the recurrence and the difference equation for the big $q$-Jacobi polynomials emerge from the difference equations for $\mathcal{C}_k$ in the limit $\sigma_s \to \sigma_s^{(n)}$. 

We first formulate the difference equations in terms of $\hat{\mathcal{C}}_k^\text{ren}$. Using the relation \eqref{renormcnhat} between $\hat{\mathcal{C}}_k^\text{ren}$ and $\mathcal{C}_k$, the four difference equations for $\mathcal{C}_k$ derived in Theorem \ref{thm6p2} and Theorem \ref{thm6p4} can be expressed as difference equations for  $\hat{\mathcal{C}}^\ren_k$. Define the renormalized difference operators 
\begin{subequations} \label{DhatCkrendef} \begin{align}
\label{DhatCkrendef1} & H_{\hat{\mathcal{C}}_k^\ren}(b,\nu) = N_3(\nu,\sigma_s)  H_{\mathcal{C}_k}(b,\nu)  N_3(\nu,\sigma_s)^{-1}, \\
\label{DhatCkrendef2} & \tilde{H}_{\hat{\mathcal{C}}_k^\ren}(b,\sigma_s) = N_3(\nu,\sigma_s)  \tilde{H}_{\mathcal{C}_k}(b,\sigma_s)  N_3(\nu,\sigma_s)^{-1},
\end{align} \end{subequations} 
where the difference operators $H_{\mathcal{C}_k}$ and $\tilde{H}_{\mathcal{C}_k}$ are respectively defined by \eqref{Hcn} and \eqref{Hcntilde}, and the normalization factor $N_3$ is given by \eqref{N3}. Using the identity \eqref{gbsbdifferenceeqs} satisfied by the function $g_b$, it can be verified that $N_3$ satisfies the following difference equations:
 \beq\label{7p20}
\frac{N_3(\nu,\sigma_s)}{N_3(\nu+ib,\sigma_s)} = S(\theta_*,\nu), \qquad \frac{N_3(\nu,\sigma_s)}{N_3(\nu-ib,\sigma_s)} = e^{(-1)^{k+1} 2 \pi  b (-i b+\theta_0-\theta_t)} S(-\theta_*,-\nu), 
\eeq
where
\beq\label{S}
S(\theta_*,\nu) = -e^{4 \pi  b (2 \nu +i b) \left(\left\lfloor \frac{k}{2}\right\rfloor -\frac{1}{2}\right)} e^{\pi  b (-1)^k (-2 i b+\theta_0-\theta_t-2 \nu )} \tfrac{\Gamma \left(\frac{b Q}{2}-i b \left(\theta_0+\nu -\frac{\theta_*}{2}\right)\right) \Gamma \left(\frac{b Q}{2}-i b \left(\frac{\theta_*}{2}-\theta_t+\nu \right)\right)}{\Gamma \left(\frac{1-b^2}{2} -i b \left(\theta_0+\frac{\theta_*}{2}-\nu \right)\right) \Gamma \left(\frac{1-b^2}{2} + ib \left(\frac{\theta_*}{2}+\theta_t+\nu \right)\right)}.
\eeq
Using \eqref{7p20}, we find that the difference operator in the left-hand side of \eqref{DhatCkrendef1} is given by 
\beq
H_{\hat{\mathcal{C}}_k^\ren}(b,\nu) = H_{\hat{\mathcal{C}}_k^\ren}^+(b,\nu) e^{ib \partial_\nu} +  H_{\mathcal{C}_k}^{0}(\nu) + H_{\hat{\mathcal{C}}_k^\ren}^-(b,\nu) e^{-ib \partial_\nu},
\eeq
where $H_{\mathcal{C}_k}^{0}(\nu)$ is defined by \eqref{Hn0} and
\beq\label{7p23}
H_{\hat{\mathcal{C}}_k^\ren}^\pm (b,\nu) = -4 e^{\pm \pi  b (-1)^k (\theta_0-\theta_t\pm 2 \nu )}
 \cosh \left(\pi  b \left(\tfrac{i b}{2}-\theta_0\mp\tfrac{\theta_*}{2}\pm\nu \right)\right) \cosh \left(\pi  b \left(\tfrac{i b}{2}+\theta_t\pm\tfrac{\theta_*}{2}\pm\nu \right)\right). \eeq
Moreover, it can be showed in a similar way that $N_3$ satisfies the difference equations
\beq\label{7p24}
\frac{N_3(\nu,\sigma_s)}{N_3(\nu,\sigma_s+ib)} = R(\sigma_s), \qquad \frac{N_3(\nu,\sigma_s)}{N_3(\nu,\sigma_s-ib)} = R(-\sigma_s),
\eeq
where
\beq\label{R}\begin{split}
R(\sigma_s)=& e^{4 \pi  b \left(\sigma_s+\frac{i b}{2}\right) \left(\left\lfloor \frac{k}{2}\right\rfloor -\frac{1}{2}\right)} \tfrac{\Gamma \left(b^2-2 i b \sigma_s\right) \Gamma (-2 i b \sigma_s)}{\Gamma \left(2 i b \sigma_s-b^2\right) \Gamma \left(2 i b \sigma_s-2 b^2\right)} \tfrac{\Gamma \left(\frac{1-b^2}{2} +ib (\theta_*+\sigma_s)\right)}{\Gamma \left(\frac{b Q}{2}+ib (\theta_*-\sigma_s)\right)} \prod _{\epsilon=\pm} \tfrac{\Gamma \left(\frac{1-b^2}{2} +ib \left(\theta_0+\epsilon \theta_t+\sigma_s\right)\right)}{\Gamma \left(\frac{b Q}{2}+ib \left(\theta_0+\epsilon \theta_t-\sigma_s\right)\right)}.
\end{split}\eeq
Using \eqref{7p24}, we find that the difference operator in the left-hand side of \eqref{DhatCkrendef2} takes the form
 \beq
 \tilde{H}_{\hat{\mathcal{C}}_k^\ren}(b,\sigma_s) = \tilde{H}_{\hat{\mathcal{C}}_k^\ren}^+(b,\sigma_s) e^{ib \partial_{\sigma_s}} + \tilde{H}^0_{\mathcal{C}_k}(\sigma_s) +  \tilde{H}_{\hat{\mathcal{C}}_k^\ren}^+(b,-\sigma_s) e^{-ib \partial_{\sigma_s}}, 
 \eeq
 where $ \tilde{H}^0_{\mathcal{C}_k}$ is defined by \eqref{Htilden0} and
 \beq
\tilde{H}_{\hat{\mathcal{C}}_k^\ren}^+(b,\sigma_s) = -2 e^{\pi  b (-1)^k \left(\sigma_s+\frac{i b}{2}\right)} \frac{\cosh \left(\pi  b \left(\tfrac{i b}{2}-\theta_*+\sigma_s\right)\right)}{\operatorname{sinh}{(2\pi b \sigma_s)}\operatorname{sinh}{(\pi b(ib+2\sigma_s)}} \prod _{\epsilon=\pm 1} \cosh \left(\pi  b \left(\tfrac{i b}{2}-\theta_0+\sigma_s+\epsilon\theta_t\right)\right).
 \eeq

 \begin{lemma}[Difference equations for $\hat{\mathcal{C}}_k^\ren$]
Let  $H_{\hat{\mathcal{C}}_k^\ren}$ and $\tilde{H}_{\hat{\mathcal{C}}_k^\ren}$ be the difference operators defined in \eqref{DhatCkrendef}. For each integer $k \geq 1$, the renormalized confluent fusion kernel $\hat{\mathcal{C}}_k^\ren$ satisfies the pair of difference equations 
\begin{subequations} \label{differenceCkrenhat} 
\begin{align}\label{differenceCkrenhata}
 & H_{\hat{\mathcal{C}}_k^\ren}(b,\nu)~\hat{\mathcal{C}}_k^\ren(b,\boldsymbol{\theta},\nu,\sigma_s) = 2\cosh{(2\pi b \sigma_s)}~\hat{\mathcal{C}}_k^\ren(b,\boldsymbol{\theta},\nu,\sigma_s), 
 	\\ \label{differenceCkrenhatb}
 & \tilde{H}_{\hat{\mathcal{C}}_k^\ren}(b,\sigma_s)~\hat{\mathcal{C}}_k^\ren(b,\boldsymbol{\theta},\nu,\sigma_s) = e^{(-1)^{k+1}2\pi b \nu}~\hat{\mathcal{C}}_k^\ren(b,\boldsymbol{\theta},\nu,\sigma_s),
\end{align} \end{subequations}
as well as the pair of difference equations obtained by replacing $b\to b^{-1}$ in \eqref{differenceCkrenhat}. 
\end{lemma}
\begin{proof}
The lemma follows immediately from \eqref{DhatCkrendef} together with Theorem \ref{thm6p2} and Theorem \ref{thm6p4}.
\end{proof}

The next proposition describes how the recurrence relation and the difference equation for the big q-Jacobi polynomials emerge from \eqref{differenceCkrenhat} in the limit $\sigma_s \to \sigma_s^{(n)}$. 

\begin{proposition}
Let $k \geq 1$  and $n \geq 0$ be integers. In the limit $\sigma_s \to \sigma_s^{(n)}$, the difference equations \eqref{differenceCkrenhata} and (\ref{differenceCkrenhatb}) for $\hat{\mathcal{C}}^\ren_k$ reduce to the three-term recurrence relation (\ref{recurrenceJn}) and the difference equation \eqref{differenceJn} for the big q-Jacobi polynomials, respectively.
More precisely,
\beq \label{7p29}
R_{J_n}(\alpha^{(-1)^k},\beta^{(-1)^k},\gamma^{(-1)^k};q^{(-1)^k})~\hat{\mathcal{C}}_k^\ren\big( b,\boldsymbol{\theta},\nu,\sigma_s^{(n)}\big) = x^{(-1)^k}~\hat{\mathcal{C}}_k^\ren\big( b,\boldsymbol{\theta},\nu,\sigma_s^{(n)}\big),
\eeq
and 
\begin{align}\nonumber\label{7p30}
& \Delta_{J_n}(\alpha^{(-1)^k},\beta^{(-1)^k},\gamma^{(-1)^k};q^{(-1)^k},x^{(-1)^k})~\hat{\mathcal{C}}_k^\ren\big( b,\boldsymbol{\theta},\nu,\sigma_s^{(n)}\big) \\
& = q^{(-1)^{k+1}n}\big(1-q^{(-1)^k n}\big)\big(1-\alpha^{(-1)^k} \beta^{(-1)^k} q^{(-1)^k n+1}\big) x^{2(-1)^k}  \hat{\mathcal{C}}_k^\ren\big( b,\boldsymbol{\theta},\nu,\sigma_s^{(n)}\big),
\end{align}
where the parameters are related according to \eqref{paramjacobi}, $R_{J_n}$ is the recurrence operator defined in \eqref{RJndef}, and $\Delta_{J_n}$ is the difference operator defined in \eqref{DeltaJndef}.
\end{proposition} 
\begin{proof}
We first prove that in the limit $\sigma_s \to \sigma_s^{(n)}$ the difference equation \eqref{differenceCkrenhatb} reduces to the three-term recurrence relation \eqref{7p29}. It follows from the definition \eqref{sigmasdiscrete} of $\sigma_s^{(n)}$ that
\beq
\lim\limits_{\sigma_s\to\sigma_s^{(n)}} e^{\pm ib \sigma_s} \hat{\mathcal{C}}_k^\ren\big( b,\boldsymbol{\theta},\nu,\sigma_s\big) = T_{n\pm 1} \hat{\mathcal{C}}_k^\ren\big( b,\boldsymbol{\theta},\nu,\sigma_s^{(n)}\big),
\eeq
where $T_{n\pm1}$ acts on $\hat{\mathcal{C}}_k^\ren$ as $T_{n\pm1}\hat{\mathcal{C}}_k^\ren\big( b,\boldsymbol{\theta},\nu,\sigma_s^{(n)}\big) = \hat{\mathcal{C}}_k^\ren\big( b,\boldsymbol{\theta},\nu,\sigma_s^{(n\pm1)}\big)$. Moreover, it can be verified that the following identities hold:
\begin{align}\nonumber\label{7p32}
\lim\limits_{\sigma_s\to\sigma_s^{(n)}} & \tilde{H}_{\hat{\mathcal{C}}_k^\ren}^+\big(b,\pm \sigma_s\big) e^{\pm ib\partial_{\sigma_s}} \hat{\mathcal{C}}_k^\ren\big( b,\boldsymbol{\theta},\nu,\sigma_s\big) \\
& = -e^{(-1)^k\pi b(ib-2\theta_0-\theta_*)}c_n^\pm \big(\alpha^{(-1)^k},\beta^{(-1)^k},\gamma^{(-1)^k};q^{(-1)^k}\big)T_{n\pm1}\hat{\mathcal{C}}_k^\ren\big( b,\boldsymbol{\theta},\nu,\sigma_s^{(n)}\big), \end{align} 
\begin{align}\label{7p33}\nonumber
& \tilde{H}^0_{\mathcal{C}_k}\big(\sigma_s^{(n)}\big) \\
& = -e^{(-1)^k\pi b(ib-2\theta_0-\theta_*)} \lb 1-c_n^+\big(\alpha^{(-1)^k},\beta^{(-1)^k},\gamma^{(-1)^k};q^{(-1)^k}\big)-c_n^-\big(\alpha^{(-1)^k},\beta^{(-1)^k},\gamma^{(-1)^k};q^{(-1)^k}\big)\rb,
\end{align}
where $c_n^\pm$ are defined in \eqref{cnpm}. Finally, it follows from \eqref{7p32} and \eqref{7p33}, together with 
\beq
e^{(-1)^{k+1}2\pi b \nu}=-e^{(-1)^k\pi b(ib-2\theta_0-\theta_*)} x^{(-1)^k}
\eeq
that the limit $\sigma_s\to \sigma_s^{(n)}$ of the difference equation \eqref{differenceCkrenhatb} reduces to
\begin{align}\nonumber
\tilde{H}_{\hat{\mathcal{C}}_k^\ren}(b,\sigma_s^{(n)}) \hat{\mathcal{C}}_k^\ren\big( b,\boldsymbol{\theta},\nu,\sigma_s^{(n)}\big) & = -e^{(-1)^k (ib-2\theta_0-\theta_*)} R_{J_n}(\alpha^{(-1)^k},\beta^{(-1)^k},\gamma^{(-1)^k};q^{(-1)^k}) \hat{\mathcal{C}}_k^\ren\big( b,\boldsymbol{\theta},\nu,\sigma_s^{(n)}\big)
	\\ \label{7p34} 
& = -e^{(-1)^k (ib-2\theta_0-\theta_*)} x^{(-1)^k} \hat{\mathcal{C}}_k^\ren\big( b,\boldsymbol{\theta},\nu,\sigma_s^{(n)}\big).
\end{align}
This shows that \eqref{differenceCkrenhatb} reduces to \eqref{7p29} as $\sigma_s \to \sigma_s^{(n)}$. 

We now show that the difference equation \eqref{differenceCkrenhata} reduces to the difference equation \eqref{7p30} satisfied by the big $q$-Jacobi polynomials in the limit $\sigma_s \to \sigma_s^{(n)}$. Observe that the potential $H_{\mathcal{C}_k}^{0}(\nu)$ and the coefficients $H_{\hat{\mathcal{C}}_k^\ren}^\pm (b,\nu)$, which are defined in \eqref{Hn0} and \eqref{7p23} respectively, are related as follows:
 \beq\label{7p35}
 H_{\mathcal{C}_k}^{0}(\nu) = -2\cosh{\big(2\pi b \big( \theta_t-\theta_0+\tfrac{ib}2\big) \big)} - H_{\hat{\mathcal{C}}_k^\ren}^+ (b,\nu)-H_{\hat{\mathcal{C}}_k^\ren}^- (b,\nu).
 \eeq
Therefore, the limit $\sigma_s \to \sigma_s^{(n)}$ of the difference equation \eqref{differenceCkrenhata} can be written as 
\begin{equation}\label{7p36}\begin{split}
& H_{\hat{\mathcal{C}}_k^\ren}^+(b,\nu) \hat{\mathcal{C}}_k^\ren\big( b,\boldsymbol{\theta},\nu+ib,\sigma_s^{(n)}\big) + H_{\hat{\mathcal{C}}_k^\ren}^-(b,\nu) \hat{\mathcal{C}}_k^\ren\big(b,\boldsymbol{\theta},\nu-ib,\sigma_s^{(n)}\big) \\
& - \lb H_{\hat{\mathcal{C}}_k^\ren}^+ (b,\nu)+H_{\hat{\mathcal{C}}_k^\ren}^- (b,\nu)\rb \hat{\mathcal{C}}_k^\ren\big(b,\boldsymbol{\theta},\nu,\sigma_s^{(n)}\big) \\
& = \lb 2\cosh{\big(2\pi b \big( \theta_t-\theta_0+\tfrac{ib}2\big) \big)}+2\cosh{(2\pi b \sigma_s^{(n)})} \rb \hat{\mathcal{C}}_k^\ren\big(b,\boldsymbol{\theta},\nu,\sigma_s^{(n)}\big). \end{split}\end{equation}
We observe that the following identities hold:
\beq\label{7p37}
H_{\hat{\mathcal{C}}_k^\ren}^\pm (b,\nu) = - e^{(-1)^{k+1}2\pi b\lb \theta_0-\theta_t-\tfrac{ib}2\rb}x^{2(-1)^{k+1}} d^\pm\big(\alpha^{(-1)^k},\beta^{(-1)^k},\gamma^{(-1)^k},q^{(-1)^k},x^{(-1)^k}\big),
\eeq
where $d^\pm$ are defined in \eqref{dpm}. Moreover, under the parameter correspondence \eqref{paramjacobi} we have 
\begin{align}\nonumber
\frac{2\cosh{\big(2\pi b \big( \theta_t-\theta_0+\tfrac{ib}2\big) \big)}+2\cosh{(2\pi b \sigma_s^{(n)})}}{-e^{(-1)^{k+1}2\pi b\big(\theta_0-\theta_t-\tfrac{ib}2\big)}x^{2(-1)^{k+1}}} 
= &\;q^{(-1)^{k+1} n} \big(1-q^{(-1)^k n}\big) 
	\\ \label{7p38}
&\times \left(1-\alpha ^{(-1)^k} \beta ^{(-1)^k} q^{(-1)^k (n+1)}\right) x^{2 (-1)^k}.
\end{align} 
Finally, according to \eqref{paramjacobi}, a shift $\nu \to \nu\pm ib$ implies a multiplication $x\to q^{\pm1} x$. Thus, using \eqref{7p37} and \eqref{7p38}, is straightforward to see that \eqref{7p36} reduces to \eqref{7p30}.
\end{proof}

\begin{remark}[An alternative proof of Theorem \ref{thjacobi}] 
The proof of Theorem \ref{thjacobi} presented in Section \ref{section7p2} is based on a direct evaluation of the limit $\sigma_s \to \sigma_s^{(n)}$ in the integral representation \eqref{hatCkrendef} for $\hat{\mathcal{C}}_k^\ren(b,\boldsymbol{\theta},\nu,\sigma_s)$. An alternative approach is based on the recurrence relation \eqref{7p29}. This approach is the one we originally used to arrive at Theorem \ref{thjacobi}. In this alternative approach, the proof of Theorem \ref{thjacobi} involves two steps. First, the residue computation that led to \eqref{7p15} is carried out in the special case of $n=0$; this yields
\beq
\hat{\mathcal{C}}_k^\ren(b,\boldsymbol{\theta},\nu,\sigma_s^{(0)}) = 1, \qquad k\geq1,
\eeq
and shows that \eqref{degenjacobi} holds for $n=1$. Second, since the functions
$$ \hat{\mathcal{C}}_k^\ren(b,\boldsymbol{\theta},\nu,\sigma_s^{(n)}) \qquad \text{and} \qquad J_n(x^{(-1)^k};\alpha^{(-1)^k},\beta^{(-1)^k},\gamma^{(-1)^k};q^{(-1)^k})$$
satisfy the same three-term recurrence relation (see \eqref{7p29} and \eqref{recurrenceJn}), an inductive argument shows that \eqref{degenjacobi} holds also for $n\geq 1$. 
\end{remark}

\section{Conclusions and perspectives}\label{conclusionssec}
We have studied the family of confluent Virasoro fusion kernels $\mathcal{C}_k(b,\boldsymbol{\theta},\sigma_s,\nu)$ defined in \eqref{gnm}.
We have shown in Theorems \ref{thm6p2} and \ref{thm6p4} that $\mathcal{C}_k$ is a joint eigenfunction of four difference operators for each $k$. Furthermore, we have proved in Theorems \ref{thhahn} and Theorem \ref{thjacobi} that $\mathcal{C}_k(b,\boldsymbol{\theta},\sigma_s,\nu)$ reduces (up to normalization) to the continuous dual $q$-Hahn polynomials when $\nu$ is suitably discretized and to the big $q$-Jacobi polynomials when $\sigma_s$ is suitably discretized. We have also shown that the Virasoro fusion kernel $F$ reduces (up to normalization) to the Askey--Wilson polynomials when $\sigma_s$  is suitably discretized (Theorem \ref{FAWthm}). As described in the introduction, these results have led us to propose the existence of a non-polynomial version of the $q$-Askey scheme with the Virasoro fusion kernel as its top member. Our results have been summarized in Figure \ref{schemefig2}. 


Let us point out that the confluent Virasoro fusion kernels $\mathcal{C}_k(b,\boldsymbol{\theta},\sigma_s,\nu)$ are not independent for different values of $k$. In fact, the kernels $\mathcal{C}_k$ can be viewed as infinite dimensional generalizations of the connection matrices which relate the solutions of the confluent hypergeometric equation at the singular points $z=0$ and $z=\infty$ in different Stokes sectors \cite{LR}. However, the confluent hypergeometric equation possesses two independent Stokes matrices, and any two consecutive connection matrices are related by a Stokes matrix. We conjecture that two consecutive confluent Virasoro fusion kernels are related by the following integral transform:
\beq
\mathcal{C}_{k+1}(b,\boldsymbol{\theta},\sigma_s,\nu_{n+1}) = \int_{\mathbb{R}} d\nu_n ~  \mathcal{S}_{n}\left[\substack{\theta_t\vspace{0.08cm} \\ \theta_{*}\;\;\;\theta_0};\substack{\nu_{n+1}\vspace{0.15cm} \\  \nu_n}\right] ~ \mathcal{C}_{k}(b,\boldsymbol{\theta},\sigma_s,\nu_{n}),
\eeq
where the kernel $\mathcal{S}_{n}\left[\substack{\theta_t\vspace{0.08cm} \\ \theta_{*}\;\;\;\theta_0};\substack{\nu_{n+1}\vspace{0.15cm} \\  \nu_n}\right]$ is the Stokes kernel which was introduced in \cite[Eq.(5.9)]{LR}.

Finally, it would be interesting to understand the difference operators introduced in this article from the viewpoint of integrable systems. It was shown in \cite{R20} that under a certain parameter correspondence the difference operator $H_\text{ren}$ defined in \eqref{3p18} corresponds to the quantum relativistic hyperbolic Calogero-Moser Hamiltonian tied to the root system $BC_1$. Relativistic Toda system were found in \cite{R1990} and various Toda limits of relativistic Calogero-Moser systems were studied in \cite{vandiejen95}. In particular, a Toda limit of a one-parameter specialization of Ruijsenaars' hypergeometric function was considered in \cite{R2011}, and similar functions were obtained in \cite{KLS2002} from a quantum group perspective. Both of these works seem to involve one-parameter specializations of the family of confluent Virasoro fusion kernels studied in this article. It would be interesting to understand this better.

\appendix

\section{$q$-hypergeometric series}\label{appendixA}
The $q$-Pochammer symbols $(a;q)_n$ and $(a_1,a_2,...,a_m;q)_n$ are defined by
\begin{align}\label{qpochhammerdef}
(a;q)_n = \prod_{k=0}^{n-1} (1-a q^k) \quad \text{and} \quad (a_1,a_2,...,a_m;q)_n = \prod_{j=1}^m (a_j;q)_n.
\end{align}
The $q$-hypergeometric series $_{s+1}\phi_s$ is a $q$-deformation of the hypergeometric series. It is defined by
\begin{align}\label{hypergeometricphidef}
_{s+1}\phi_s\left[ \begin{matrix} a_1,...a_{s+1} \\ b_1...b_s \end{matrix} ;q,z\right] = \sum_{k=0}^\infty \frac{(a_1,...,a_{s+1};q)_k}{(b_1,...,b_s,q;q)_k}z^k.
\end{align}
The series terminates if one of the $a_i$ in the numerator is equal to $q^{-n}$ for some integer $n\geq 1$. Otherwise, the series converges for $|z|<1$.

\section{The first two levels of the $q$-Askey scheme}\label{appendixB}

%

\subsection{Askey--Wilson polynomials}

The Askey--Wilson polynomials $A_n$ are the most general polynomials of the $q$-Askey scheme. They are defined by
\beq\label{AW}
A_n(z;\alpha,\beta,\gamma,\delta,q) = {}_4 \phi_3\lb \left. \begin{split}\begin{matrix}  q^{-n}, \alpha\beta\gamma\delta q^{n-1}, \alpha z, \alpha z^{-1} \\ \alpha \beta, \alpha\gamma, \alpha\delta \end{matrix} \end{split}\right| q;q \rb
\eeq
The normalization used in $\eqref{AW}$ for the Askey--Wilson polynomials is related to the standard normalization of \cite[Eq. (3.1.1)]{KS} by
\beq\label{1p2}
p_n\lb\tfrac{z+z^{-1}}2;\alpha,\beta,\gamma,\delta,q\rb = \alpha^{-n} (\alpha\beta,\alpha\gamma,\alpha\delta;q)_n ~ A_n(z;\alpha,\beta,\gamma,\delta,q).
\eeq
The right-hand side of \eqref{1p2} is symmetric in its four parameters $\alpha,\beta,\gamma,\delta$, whereas $A_n(z;\alpha,\beta,\gamma,\delta,q)$ is only symmetric in $\beta,\gamma,\delta$. Since $p_n(x;\alpha,\beta,\gamma,\delta,q)$ is a polynomial of order $n$ in $x$, $A_n$ is a polynomial of order $n$ in $z + z^{-1}$.
The polynomials $A_n$ satisfy the three-term recurrence relation
\beq\label{recurrenceAW}
(R_{A_n} A_n)(z;\alpha,\beta,\gamma,\delta,q) = (z+z^{-1}) A_n(z;\alpha,\beta,\gamma,\delta,q),
\eeq
where the operator $R_{A_n}$ is given by
\beq\label{Mn}
R_{A_n} = a^{+}_n T_{n+1} + (\alpha+\alpha^{-1}-a^{+}_n-a^{-}_n) + a^{-}_n T_{n-1},
\eeq
with $T_{n \pm 1}p_n(x) = p_{n \pm 1}(x)$ and
\beq \label{bndef}
\begin{split}
& a^{+}_n = \frac{\lb 1-\alpha\beta q^n\rb\lb 1-\alpha\gamma q^n\rb\lb 1-\alpha\delta q^n\rb\lb 1-\alpha\beta\gamma\delta q^{n-1}\rb}{\alpha \lb 1-\alpha\beta\gamma\delta q^{2n-1}\rb\lb 1-\alpha\beta\gamma\delta q^{2n}\rb}, \\
& a^{-}_n = \frac{\alpha\lb 1-q^n\rb\lb 1-\beta \gamma q^{n-1}\rb\lb 1-\beta\delta q^{n-1}\rb\lb 1-\gamma\delta q^{n-1}\rb}{\lb 1-\alpha\beta\gamma\delta q^{2n-2}\rb\lb 1-\alpha\beta\gamma\delta q^{2n-1}\rb}.
\end{split}\eeq
They also satisfy the difference equation
\beq\label{differenceAW}
(\Delta_{A_n}A_n)(z;\alpha,\beta,\gamma,\delta,q) = \lb q^{-n} +\alpha \beta \gamma \delta q^{n-1} \rb A_n(z;\alpha,\beta,\gamma,\delta,q),
\eeq
where the $q$-difference operator $\Delta_{A_n}$ is defined by
\beq\label{L}\begin{split}
(\Delta_{A_n}f)(z)=& \lb 1+\tfrac{\alpha\beta\gamma\delta}q\rb f(z) + \frac{(1-\alpha z)(1-\beta z)(1-\gamma z)(1-\delta z)}{(1-z^2)(1-q z^2)} (f(qz)-f(z)) \\
 & + \frac{(\alpha-z)(\beta-z)(\gamma-z)(\delta-z)}{(1-z^2)(q-z^2)} (f(q^{-1}z)-f(z)).
\end{split}\eeq

We next describe the second level in the $q$-Askey which consists of the continuous dual $q$-Hahn and the big $q$-Jacobi polynomials. These families of polynomials arise as a limit of the Askey--Wilson polynomials.

\subsection{Continuous dual $q$-Hahn polynomials} 

The continuous dual $q$-Hahn polynomials, denoted by $H_n(z;\alpha,\beta,\gamma,q)$, are obtained from the Askey--Wilson polynomials by setting $\delta=0$ in \eqref{AW}:
\beq\label{qhahn}
H_n(z;\alpha,\beta,\gamma,q) = A_n(z;\alpha,\beta,\gamma,0,q) = {}_3\phi_2\left(\left. \begin{matrix} q^{-n},\alpha z ,\alpha z^{-1} \\ \alpha\beta, \alpha\gamma \end{matrix}\right|q;q\right).
\eeq
The polynomials $H_n$ satisfy the three-term recurrence relation 
\beq\label{recurrenceHn}
\lb R_{H_n}(\alpha,\beta,\gamma;q) H_n\rb(z;\alpha,\beta,\gamma,q)=(z+z^{-1})~H_n(z;\alpha,\beta,\gamma,q),
\eeq
where the operator $R_{H_n}$ is defined by
\beq \label{sn}
R_{H_n}(\alpha,\beta,\gamma;q) = b^{+}_n(\alpha,\beta,\gamma;q) T_{n+1}+\lb \alpha+\alpha^{-1}-b^{+}_n(\alpha,\beta,\gamma;q)-b^{-}_n(\alpha,\beta,\gamma;q)\rb +b^{-}_n(\alpha,\beta,\gamma;q) T_{n-1}
\eeq
with 
\beq\label{cndef}
b^{+}_n(\alpha,\beta,\gamma;q)=\alpha^{-1}(1-\alpha \beta q^n)(1-\alpha \gamma q^n), \qquad 
b^{-}_n(\alpha,\beta,\gamma;q)=\alpha(1-q^n)(1-\beta \gamma q^{n-1}). 
\eeq
They also satisfy the difference equation
\beq\label{differencehahn}
\lb \Delta_{H_n}(\alpha,\beta,\gamma;q,z) H_n\rb(z;\alpha,\beta,\gamma,q) = (q^{-n}-1) H_n(z;\alpha,\beta,\gamma,q),
\eeq
where the $q$-difference operator $\Delta_{H_n} \equiv \Delta_{H_n}(\alpha,\beta,\gamma;q,z)$ is defined by
\beq\label{deltahn}\begin{split}
\lb \Delta_{H_n} f\rb(z) = ~ & h(\alpha,\beta,\gamma;q,z) f(qz) + h(\alpha,\beta,\gamma;q,z^{-1}) f(q^{-1}z) \\
& - \lb h(\alpha,\beta,\gamma;q,z)+h(\alpha,\beta,\gamma;q,z^{-1}) \rb f(z),
\end{split}\eeq
with
\beq\label{hdef}
h(\alpha,\beta,\gamma;q,z)=\frac{(1-\alpha z)(1-\beta z)(1-\gamma z)}{(1-z^2)(1-q z^2)}.
\eeq
 
\subsection{Big $q$-Jacobi polynomials}
The big $q$-Jacobi polynomials $J_n(x;\alpha,\beta,\gamma;q)$ arise from the Askey--Wilson polynomials in a more subtle way:
\beq\label{Jn}
J_n(x;\alpha,\beta,\gamma;q) = \lim\limits_{\lambda \to 0} A_n\lb\frac{x}{\lambda};\lambda,\frac{\alpha q}\lambda,\frac{\gamma q}\lambda,\frac{\lambda \beta}\gamma,q\rb
= {}_3\phi_2\left( \left. \begin{split} \begin{matrix} q^{-n},\alpha \beta q^{n+1}, x \\ \alpha q, \gamma q\end{matrix} \end{split} \right|q;q\right).
\eeq
The polynomials $J_n$ satisfy the three-term recurrence relation
\beq\label{recurrenceJn}
R_{J_n}(\alpha,\beta,\gamma;q) J_n(x;\alpha,\beta,\gamma;q)=x J_n(x;\alpha,\beta,\gamma;q),
\eeq
where the operator $R_{J_n}$ is defined by
\beq\label{RJndef}
R_{J_n}(\alpha,\beta,\gamma;q) = c^{+}_n T_{n+1} + \lb 1-c^{+}_n-c^{-}_n\rb + c^{-}_n T_{n-1},
\eeq
with
\beq\label{cnpm}
\left\{\begin{split} 
& c^{+}_n = \frac{\lb 1-\alpha q^{n+1}\rb\lb 1-\alpha \beta q^{n+1}\rb\lb 1-\gamma q^{n+1}\rb}{\lb 1-\alpha \beta q^{2n+1}\rb\lb 1-\alpha \beta q^{2n+2}\rb}, \\
& c^{-}_n = -\alpha \gamma q^{n+1} \frac{\lb 1-q^n\rb\lb 1-\alpha \beta \gamma^{-1} q^n\rb\lb 1-\beta q^n\rb}{\lb 1-\alpha \beta q^{2n}\rb\lb 1-\alpha \beta q^{2n+1}\rb}.
\end{split}\right.
\eeq
They also satisfy the difference equation
\beq\label{differenceJn}
(\Delta_{J_n}(\alpha,\beta,\gamma;q,x)J_n)(x;\alpha,\beta,\gamma;q) = q^{-n} (1-q^n)(1-\alpha \beta q^{n+1})x^2 J_n(x;\alpha,\beta,\gamma;q),
\eeq
where the $q$-difference operator $\Delta_{J_n}\equiv \Delta_{J_n}(\alpha,\beta,\gamma;q,x)$ is defined by
\begin{align}\nonumber
(\Delta_{J_n}f)(x) = ~ & d^{+}(\alpha,\beta,\gamma,q,x) f(q x) - (d^{+}(\alpha,\beta,\gamma,q,x)+d^{-}(\alpha,\beta,\gamma,q,x)) f(z) \\\label{DeltaJndef}
& +  d^{-}(\alpha,\beta,\gamma,q,x)f(q^{-1}x),
\end{align}
with
\beq\label{dpm}
d^{+}(\alpha,\beta,\gamma,q,x) = \alpha q(x-1)(\beta x-\gamma), \qquad d^{-}(\alpha,\beta,\gamma,q,x)=(x-\alpha q)(x-\gamma q). 
\eeq

\begin{remark}\label{rootofunityremark}
In the definitions of the polynomials $A_n$, $H_n$, and $J_n$, we assume that $q$ is not a root of unity, because otherwise the polynomials are in general not well-defined. Indeed, the $q$-hypergeometric series on the right-hand sides of (\ref{AW}), (\ref{qhahn}), and (\ref{Jn}) involve the ratio $(q^{-n}; q)_k/(q; q)_k$ of $q$-Pochhammer symbols, and this ratio is indeterminate for all sufficiently large $k$ if $q$ is a root of unity.
\end{remark}

\bigskip
\noindent
{\bf Acknowledgement} {\it J.R. acknowledges support from the European Research Council, Grant Agreement No. 682537 and the Ruth and Nils-Erik Stenb\"ack Foundation. J.L. acknowledges support from the European Research Council, Grant Agreement No. 682537, the Swedish Research Council, Grant No. 2015-05430, and the Ruth and Nils-Erik Stenb\"ack Foundation. }

\end{document}